\documentclass[12pt]{article}

\usepackage[utf8]{inputenc}
\usepackage[T1]{fontenc}
\usepackage{lmodern,microtype,placeins,needspace,xr-hyper}
\usepackage[margin=1in]{geometry}
\usepackage{amsmath,amssymb,amsthm,mathtools,bm}
\usepackage{booktabs,comment}
\usepackage{enumitem}
\usepackage{graphicx}
\usepackage{subcaption}
\usepackage{setspace}
\usepackage[authoryear]{natbib}
\usepackage[colorlinks=true,linkcolor=black,citecolor=black,urlcolor=black,filecolor=black]{hyperref}
\allowdisplaybreaks
\newtheorem{assumption}{Assumption}
\newtheorem{theorem}{Theorem}
\newtheorem{proposition}{Proposition}
\newtheorem{lemma}{Lemma}
\newtheorem{corollary}{Corollary}
\newtheorem{remark}{Remark}

\newtheorem{algorithm}{Algorithm}

\newcommand{\R}{\mathbb{R}}
\newcommand{\E}{E}

\newcommand{\Pn}{\mathbb{P}_{n}}
\newcommand{\calX}{\mathcal{X}}

\newcommand{\calF}{\mathcal{F}}
\newcommand{\calR}{\mathcal{R}}

\newcommand{\calH}{\mathcal{H}}

\newcommand{\ind}{\mathbf{1}}

\newcommand{\weak}{\rightsquigarrow}
\DeclareMathOperator{\Var}{Var}
\DeclareMathOperator{\Cov}{Cov}

\title{Specification Testing for Dyadic Regression Models}
\ifdefined\BlindReview
\author{}\date{}
\else
\author{%
{\large \textbf{Ulrich Hounyo}}\thanks{Department of Economics, University at Albany-State University of New York, Albany, NY 12222, United States. E-mail: \texttt{khounyo@albany.edu}.}
\and
{\large \textbf{Jiahao Lin}}\thanks{School of Digital Economy and Management, Fuyao University of Science and Technology, Fuzhou, 350109, China. E-mail: \texttt{jhlin@fyust.edu.cn}.} 
\and
{\large \textbf{Xiaojun Song}}\thanks{Department of Business Statistics and Econometrics, Peking University, Beijing, China. E-mail: \texttt{sxj@gsm.pku.edu.cn}.}
}
\date{}
\fi

\begin{document}
\maketitle

\begin{abstract}
\singlespacing
Of the 88 empirical papers on dyadic or other relational data that we
reviewed in the top five economics journals, 55 (62.5\%) use a linear
specification. To assess whether such a specification describes the
conditional mean, we develop Kolmogorov-Smirnov and Cram\'er-von Mises
tests for undirected dyadic data that also accommodate smooth nonlinear
specifications. To our knowledge, this paper is the first to develop omnibus
specification tests for undirected dyadic regression models. A corrected first-order multiplier bootstrap is asymptotically exact in
two standard Gaussian cases: the first-order nondegenerate case and independent dyads case.  For a
general class of such processes, we establish that a second-order
multiplier bootstrap consistently approximates the limiting law
uniformly over the indexing function class, even when the indexing
functions are estimated. For the specification tests, we augment this bootstrap to account for
two additional sources of non-Gaussianity: parameter estimation and
the curvature of the regression function. We then construct an envelope test that is
asymptotically exact in both Gaussian cases and controls size in the
other degenerate cases. Simulations illustrate the size and power of the
tests. In an application to housing sales, we reject a common mean price gap within distance bands but do not reject
separate linear relationships by metro access.
\end{abstract}

\newpage

\section{Introduction}

Linear and nonlinear specifications are widely used in empirical work
with dyadic data. We review 88 empirical papers that use dyadic or other
relational data and were published in the \emph{American Economic Review},
\emph{Econometrica}, \emph{Journal of Political Economy},
\emph{Quarterly Journal of Economics}, or \emph{Review of Economic Studies}.
We identify linear specifications in 55 papers (62.5\%), including 17
that also use nonlinear specifications. In 19 additional papers, we
identify nonlinear or structural specifications but no linear
specification.\footnote{The review covers selected papers published
during 1993-2025. We record specifications verified in the available
published or working-paper versions. The sample includes directed,
bipartite, and panel settings beyond the undirected framework studied here.} A complementary question, for both linear and nonlinear regressions,
is whether the specified functional form correctly describes the
conditional mean given the observed covariates.

We develop Kolmogorov-Smirnov (KS) and Cram\'er-von Mises (CvM)
specification tests for smooth parametric conditional means with
undirected dyadic data. The framework covers linear and nonlinear
regressions. Dyads sharing a node may be dependent; collections of dyads
with no nodes in common are independent.  We first develop the tests in the two standard Gaussian
cases and then turn to the general degenerate case, where the limit need
not be Gaussian.

Our first contribution is to develop, to our knowledge, the first
omnibus specification tests for undirected dyadic regression models.
Because KS and CvM combine moments across covariate values, their
critical values depend on the joint behavior of the estimated moments.
In the nondegenerate case, we combine a uniform expansion of these
moments with the empirical-process theory of \citet{davezies2021} and
the first-order multiplier results of \citet{chiang2023inference} to
establish valid inference at the $\sqrt n$ rate. We then extend inference
to independent dyads, a standard benchmark and a special case of
degeneracy with a Gaussian limit at the faster $\sqrt{N_n}$ rate,
where $N_n=n(n-1)/2$. The usual node-multiplier bootstrap counts each
dyad's own contribution to the covariance twice, once for each endpoint.
An exact covariance identity gives a correction that is negligible in
the nondegenerate case and necessary under independent dyads. The
resulting corrected first-order multiplier bootstrap yields
asymptotically exact KS and CvM tests in both Gaussian cases. The
covariance calculation is related to those used for dyadic and multiway
inference \citep{fafchamps2007formation,aronow2015cluster,
tabord2019inference,cameron2011robust,mackinnon2021wild,chiang2023standard}.

Our second contribution is a bootstrap approximation for general
degenerate dyadic empirical processes.  We establish that weighting each dyad by the
product of its two node multipliers consistently reproduces the joint
law of the potentially non-Gaussian node interaction and the Gaussian
dyad-specific component, uniformly over the indexing function class.  The same bootstrap draw recovers both
components without estimating either separately, and the approximation
remains valid when the indexing functions are estimated. \citet{davezies2021} characterize
the sampling limit in this case, but their bootstrap does not generally
extend to it; the Gaussian approximation of \citet{chiang2023inference}
requires nondegeneracy.
The dyad-specific component distinguishes this problem from multiplier
approximation for U-processes based on i.i.d.\ nodes
\citep{han2022multiplier}.  Our uniform approximation for estimated dyad marks complements the
multiplier bootstrap of \citet{chatterjee2026higher} for finite
collections of graphon motif densities.

Our third contribution is to show how estimating the regression
parameters changes the second-order limit of the specification
statistics. The moments used for testing can be degenerate even when
the regression score for parameter estimator is not. The parameter estimator can then converge
at the slower $\sqrt n$ rate, introducing two additional sources of
non-Gaussianity: a product involving the parameter-estimation error
and a curvature term that vanishes for linear regression. Our
augmented second-order multiplier bootstrap reproduces these terms
jointly with the node interaction by using the same multipliers
throughout. This preserves their dependence, including possible
cancellation, and yields asymptotically exact tests in the covered
degenerate cases. Related differences in convergence rates arise in
clustered GMM \citep{davezies2025analytic}. 

Finally, we take the larger
of the corrected first-order and augmented second-order critical values.
The resulting \emph{envelope test}  is asymptotically exact in the two
Gaussian cases and has asymptotic size at most the nominal level in the
remaining possibly non-Gaussian cases.

The tests build on conditional-moment and regression specification tests
\citep{bierens1982consistent,bierens1990,stute1997,whang2000consistent,
escanciano2006}. Related approaches use kernel smoothing
\citep{zheng1996,fan1996consistent} or accommodate spatial dependence
\citep{su2017specification,gupta2024consistent,yang2024model,
lee2025heteroskedasticity}. For dyadic data, existing work develops
nonparametric regression and density estimation
\citep{graham2021minimax,graham2024kernel} and inference that accommodates
degeneracy \citep{chiang2023empirical,cattaneo2024uniform,chiang2023multiway}.
Our null concerns the conditional mean given observed covariates, rather
than transitivity, neglected heterogeneity, or strategic interaction in a
network formation model
\citep{dzemski2019empirical,hahn2025neglected,pelican2026optimal}.

Simulations examine independent dyads, weak and strong nondegenerate
Gaussian designs, and a degenerate design in which all three second-order
terms contribute. In the housing application, we reject the restriction that the expected
price gap depends only on geographic distance, but do not reject a linear
specification that allows the intercept and slope to vary with metro access.

Section~\ref{sec:setup} introduces the model.
Section~\ref{sec:bootstrap} defines the statistics and develops first-order
inference. Section~\ref{sec:nongaussian} gives the second-order process
results and the envelope test. Sections~\ref{sec:simulation} and
\ref{sec:empirical} report the simulations and housing application.
Sections~\ref{app:foundations}-\ref{app:second-order} of the supplement
give the proofs.

\section{Model Setup}\label{sec:setup}
\subsection{Undirected dyadic sampling}

For each unordered pair of nodes $1\le i<j\le n$, we observe
$\bm Z_{ij}=(Y_{ij},\bm X_{ij}^\top)^\top$, where $\bm X_{ij}\in\calX$ and $\calX\subset\R^d$ is a compact rectangle. Let
\(\mathcal D_n=\{(i,j):1\leq i<j\leq n\}\) and \(N_n=|\mathcal D_n|=n(n-1)/2\). Write $\mathcal Z_n=\{\bm Z_{ij}:i<j\}$ for the observed sample. For notational convenience, whenever \(j<i\), set \(\bm{Z}_{ij}=\bm{Z}_{ji}\), and use the same convention for all dyad-level functions and latent dyad variables. For a measurable function \(f\), define
\(\Pn f=N_n^{-1}\sum_{i<j}f(\bm{Z}_{ij})\) and \(Pf=\E[f(\bm{Z}_{ij})]\). The expectation is the same for every $i<j$ under Assumption~\ref{ass:sampling} below.

Dyads sharing a node may be dependent; collections of dyads involving
disjoint sets of nodes are independent (dissociation). We also impose
joint exchangeability: any common permutation of the node labels
preserves the joint distribution of the array. The following
representation incorporates both properties
\citep{aldous1981,hoover1979,kallenberg1989representation}.
\begin{assumption}[Sampling]
\label{ass:sampling}
The undirected array \((\bm{Z}_{ij})_{i<j}\) admits a representation \(\bm{Z}_{ij}=\bm{\tau}(U_i,U_j,U_{ij})\), where the node variables \((U_i)\) and dyad variables \((U_{ij})\) are mutually independent and identically distributed, and \(\bm{\tau}\) is symmetric in its first two arguments. 
\end{assumption}
Exchangeability and dissociation, as imposed in
Assumption~\ref{ass:sampling}, are standard in the literature on
dyadic and multiway-clustered data; see, e.g., \citet{davezies2021,menzel2021bootstrap,chiang2023inference,graham2024kernel}. Under Assumption~\ref{ass:sampling}, the dyadic observations $\bm Z_{ij}$ are identically distributed but need not be independent.

\subsection{Regression model and null hypothesis}

Let $m:\calX\times\Theta\to\R$ be a specified smooth regression
function, where $\Theta\subset\R^p$ and $p$ is fixed. Let
$\bm\theta_*$ be an interior minimizer of
$E[(Y_{ij}-m(\bm X_{ij},\bm\theta))^2]$ over $\bm\theta\in\Theta$,
and define $\varepsilon_{ij}=Y_{ij}-m(\bm X_{ij},\bm\theta_*)$.
We test
\begin{equation}
 H_0:\quad E[Y_{ij}\mid\bm X_{ij}]
 =m(\bm X_{ij},\bm\theta_*)\quad\text{a.s.},
 \label{eq:null}
\end{equation}
against conditional means outside the specified family. Examples
include linear regression, $m(\bm X,\bm\theta)=(1,\bm X^\top)\bm\theta$,
and logistic regression for binary outcomes,
$m(\bm X,\bm\theta)=\Lambda\{\bm z(\bm X)^\top\bm\theta\}$,
where $\Lambda(t)=(1+e^{-t})^{-1}$ and $\bm z(\bm X)$ may include
an intercept, powers, and interactions.  The choice of $m$ need not follow from a structural model. A researcher may choose a linear specification as a convenient
approximation and ask whether it misses systematic variation in the
conditional mean. This is the question posed by \eqref{eq:null},
and it applies equally to a chosen nonlinear specification.\footnote{In the linear case, $\bm\theta_*$ gives the coefficients of the best
linear projection even when \eqref{eq:null} fails. Rejection means
that this projection does not equal the conditional mean; it does
not invalidate inference on the projection coefficients that is robust
to misspecification and dyadic dependence.}

Under $H_0$, $E[\varepsilon_{ij}\mid\bm X_{ij}]=0$ almost surely.
To assess this restriction, consider the population moment
\[
 \Delta(\bm x)=E[\varepsilon\ind\{\bm X\preceq\bm x\}],
\]
where $\bm X\preceq\bm x$ means that $X_k\le x_k$ for every
$k=1,\ldots,d$. Varying $\bm x$ gives the following characterization
of the null without nonparametric smoothing \citep{stute1997}.

\begin{lemma}[Conditional-mean characterization]
\label{lem:identification}
Suppose $E|\varepsilon_{ij}|<\infty$. Then $H_0$ holds if and only if
$\Delta(\bm x)=0$ for every $\bm x\in\R^d$.
\end{lemma}

 Lemma~\ref{lem:identification} establishes that testing $H_0$ is equivalent to
testing whether $\Delta(\bm x)=0$ for every $\bm x\in\R^d$. To estimate $\Delta(\bm x)$, replace the unobserved errors by
least-squares residuals. Specifically, let
\[
 \widehat{\bm\theta}\in\arg\min_{\bm\theta\in\Theta}
 \Pn\bigl[(Y-m(\bm X,\bm\theta))^2\bigr],
 \qquad
 \widehat\varepsilon_{ij}
 =Y_{ij}-m(\bm X_{ij},\widehat{\bm\theta}).
\]
The resulting estimated process is
\begin{equation}
 \widehat R_n(\bm x)
 =\Pn[\widehat\varepsilon\ind\{\bm X\preceq\bm x\}].
 \label{eq:Rhat}
\end{equation}

\subsection{Effect of parameter estimation}
\label{sec:estimated process}

Because $\widehat R_n(\bm x)$ uses estimated residuals, estimating
$\bm\theta_*$ can affect its sampling distribution. To describe this
effect, write
$\bm g_{\bm\theta}(\bm X)
=\partial m(\bm X,\bm\theta)/\partial\bm\theta$
and let $\bm g(\bm X)=\bm g_{\bm\theta_*}(\bm X)$ denote the
regression gradient at $\bm\theta_*$. Define
\[
 \bm Q=E[\bm g(\bm X)\bm g(\bm X)^\top],
 \qquad
 \bm M(\bm x)=E[\bm g(\bm X)\ind\{\bm X\preceq\bm x\}].
\]
For nonsingular $\bm Q$, define the adjusted weight $q_{\bm x}$ and
the corresponding mark $r_{\bm x}$ by
\begin{equation}
 q_{\bm x}(\bm X)
 =\ind\{\bm X\preceq\bm x\}
  -\bm M(\bm x)^\top\bm Q^{-1}\bm g(\bm X),
 \qquad
 r_{\bm x}(\bm Z)=\varepsilon q_{\bm x}(\bm X).
 \label{eq:pop-mark}
\end{equation}
The subtracted term is the population linear projection of the
indicator $\ind\{\bm X\preceq\bm x\}$ onto the components of $\bm g(\bm X)$. Consequently,
$E[\bm g(\bm X)q_{\bm x}(\bm X)]=\bm0$. Under $H_0$, $r_{\bm x}$
is the mark in the first-order expansion of $\widehat R_n(\bm x)$
and thus accounts for the leading effect of estimating $\bm\theta_*$.

The adjustment preserves the population moment under both the null
and misspecification. Since $\bm\theta_*$ is an interior minimizer,
the least-squares first-order condition gives
$E[\bm g(\bm X)\varepsilon]=\bm0$ under Assumption~\ref{ass:moments}.
Hence,
\begin{equation}
 Pr_{\bm x}
 =\Delta(\bm x)
  -\bm M(\bm x)^\top\bm Q^{-1}E[\bm g(\bm X)\varepsilon]
 =\Delta(\bm x).
 \label{eq:delta}
\end{equation}
By Lemma~\ref{lem:identification}, $H_0$ therefore holds if and only
if $Pr_{\bm x}=0$ for every $\bm x\in\R^d$.

For the sample counterpart, let
$\widehat{\bm g}_{ij}=\bm g_{\widehat{\bm\theta}}(\bm X_{ij})$
and define
\[
 \widehat{\bm Q}=\Pn(\widehat{\bm g}\widehat{\bm g}^\top),
 \qquad
 \widehat{\bm M}(\bm x)
 =\Pn[\widehat{\bm g}\ind\{\bm X\preceq\bm x\}].
\]
When $\widehat{\bm Q}$ is nonsingular, the estimated weight and mark are
\begin{equation}
 \begin{aligned}
 \widehat q_{\bm x}(\bm X_{ij})
 &=\ind\{\bm X_{ij}\preceq\bm x\}
   -\widehat{\bm M}(\bm x)^\top
    \widehat{\bm Q}^{-1}\widehat{\bm g}_{ij},\qquad
 \widehat r_{ij}(\bm x)
 =\widehat\varepsilon_{ij}\widehat q_{\bm x}(\bm X_{ij}).
 \end{aligned}
 \label{eq:estimated-mark}
\end{equation}
If $\widehat{\bm\theta}$ is an interior least-squares solution,
then $\Pn(\widehat{\bm g}\widehat\varepsilon)=\bm0$, so
\begin{equation}
 \Pn\widehat r(\bm x)
 =\widehat R_n(\bm x)
  -\widehat{\bm M}(\bm x)^\top\widehat{\bm Q}^{-1}
   \Pn(\widehat{\bm g}\widehat\varepsilon)
 =\widehat R_n(\bm x).
 \label{eq:exact-orthogonal}
\end{equation}
Thus we can use the estimated process $ \Pn\widehat r(\bm x)$ for inference on $Pr_{\bm x}$
without changing the observed process.

\begin{assumption}[Regression and moments]\label{ass:moments}
(i) $\calX$ and $\Theta$ are compact, and the population least-squares
criterion has a unique minimizer
$\bm\theta_*\in\operatorname{int}(\Theta)$.
(ii) On a neighborhood of $\Theta$, $m$ is three times continuously
differentiable in $\bm\theta$; $m$ and its first three derivatives
are uniformly bounded over $\calX$ and this neighborhood.
(iii) $E|Y_{ij}|^q<\infty$ for some $q>4$, and $\bm Q$ is positive
definite.
(iv) $\widehat{\bm\theta}$ is a measurable choice of the sample
least-squares minimizer.
\end{assumption}

These are standard conditions for expanding the estimated process
and estimating its covariance. Under Assumptions~\ref{ass:sampling}
and~\ref{ass:moments}, $\widehat{\bm\theta}$ is interior and
$\widehat{\bm Q}$ is nonsingular with probability tending to one.
On this event, \eqref{eq:exact-orthogonal} holds for every $\bm x$.
\section{Specification Tests and First-Order Inference}\label{sec:bootstrap}

\subsection{KS and CvM statistics}\label{sec:statistics}

We construct KS and CvM statistics from the estimated process
$\Pn\widehat r$. The KS statistic measures its largest absolute
deviation from zero, whereas the CvM statistic averages its squared
deviations over the observed covariates.

Let $F_{\bm X}$ denote the distribution of $\bm X_{ij}$ and
$\widehat F_n$ its empirical counterpart based on the $N_n$ dyads.
Write $F_k$ and $\widehat F_{n,k}$ for the corresponding marginal
cumulative distribution functions, $k=1,\ldots,d$.
Choose $G_n$ combinations of quantile levels
$\bm u_g=(u_{g1},\ldots,u_{gd})^\top\in(0,1]^d$,
$g=1,\ldots,G_n$, including $(1,\ldots,1)^\top$.
The associated empirical quantiles define the evaluation points
\begin{equation}
 \widehat{\bm x}_g
 =
 \bigl(\widehat F_{n,1}^{-1}(u_{g1}),\ldots,
       \widehat F_{n,d}^{-1}(u_{gd})\bigr)^\top,
 \qquad g=1,\ldots,G_n,
 \label{eq:random-grid}
\end{equation}
and the KS grid $\widehat{\calX}_{G_n}=\{\widehat{\bm x}_g:g=1,\ldots,G_n\}$.
Here $G_n$ counts the chosen combinations of quantile levels;
several combinations may give the same point. For $d=1$ the grid
consists of empirical quantiles. For $d>1$ the coordinates are
marginal quantiles, so a grid point need not be an observed covariate
vector. The quantile levels themselves may depend on the data,
subject to Assumption~\ref{ass:grid}.

The statistics are
\begin{equation}
 \begin{aligned}
 T_n^{\mathrm{KS}}
 &=\max_{\bm x\in\widehat{\calX}_{G_n}}
       |\Pn\widehat r(\bm x)|,\\
 T_n^{\mathrm{CvM}}
 &=\int\{\Pn\widehat r(\bm x)\}^2\,d\widehat F_n(\bm x)
   =\frac1{N_n}\sum_{i<j}\{\Pn\widehat r(\bm X_{ij})\}^2.
 \end{aligned}
 \label{eq:statistics}
\end{equation}
KS uses the grid in \eqref{eq:random-grid}. CvM uses the observed
covariate vector from every dyad, with weight $1/N_n$ even when
covariate values repeat.

\begin{assumption}[Increasing grid]\label{ass:grid}
(i) 
$G_n\to\infty$ and $\log(2G_n)/\sqrt n\to0$ in probability.
(ii) For some deterministic $h_n\to0$, with probability tending to one,
every $\bm t\in[0,1]^d$ with $\min_k t_k>h_n$ admits an index
$g\le G_n$ such that
\(
 0\le t_k-u_{gk}\le h_n,\  k=1,\ldots,d.
\)
\end{assumption}

Assumption~\ref{ass:grid} ensures that the grid becomes sufficiently fine
while controlling covariance approximation error. It accommodates
continuous and discrete covariates and permits polynomial growth in $G_n$.

\subsection{Projections and Gaussian process limits}\label{sec:asymptotics}

Each $\bm x$ in \eqref{eq:statistics} indexes a different function
$f=r_{\bm x}$. We first decompose the dyadic average $\Pn f-Pf$
for general $f$. We then apply this decomposition to $r_{\bm x}$
and show that $\Pn r_{\bm x}$ approximates $\Pn\widehat r(\bm x)$
uniformly over $\bm x$ in the two Gaussian cases.
Throughout the paper and supplement, each function class is
pointwise measurable: it contains a countable subclass such that
every function in the class is the pointwise limit of a sequence
from that subclass.

For a square-integrable function $f$ of a dyadic observation, let
$\mu_f=Pf=E[f(\bm Z_{ij})]$ and define its first- and second-order
projection terms by
\begin{align}
 f_1(u)
 &=E[f(\bm Z_{ij})\mid U_i=u]-\mu_f,
 \label{eq:f1}\\
 f_2(u,v,w)
 &=f\{\bm\tau(u,v,w)\}-\mu_f-f_1(u)-f_1(v).
 \label{eq:f2}
\end{align}
The first-order projection $f_1$ is the contribution of a single node;
$f_2$ is what remains after subtracting the contributions of both nodes.
By symmetry, the two first-order projections coincide:
\(
 f_1(u)=E[f(\bm Z_{ij})\mid U_j=u]-\mu_f.
\)
Moreover, the projections satisfy
\[
 E[f_1(U_i)]=0,\qquad
 E[f_2(U_i,U_j,U_{ij})\mid U_i]=0,\qquad
 E[f_2(U_i,U_j,U_{ij})\mid U_j]=0.
\]
These projections yield the exact decomposition
\begin{equation}
\sqrt n(\Pn f-Pf)
 =
 \underbrace{\frac{2}{\sqrt n}\sum_{i=1}^n f_1(U_i)}_{\text{first-order: node effects}}
 +\underbrace{\sqrt n\,\mathbb U_n f_2}_{\text{second-order: dyad effects}},
 \label{eq:main-projection}
\end{equation}
where
\(
 \mathbb U_n f_2
 =
 \frac1{N_n}\sum_{i<j}f_2(U_i,U_j,U_{ij}).
\)
For a VC-type class $\calF$ with envelope $F$ satisfying $PF^2<\infty$,
the projection bounds imply
$E\sup_{f\in\calF}|\mathbb U_n f_2|\le C\|F\|_{P,2}/n$, as established in
Lemma~\ref{lem:uniform-degenerate} of the supplement. Thus
$\sup_{f\in\calF}|\sqrt n\,\mathbb U_n f_2|=o_p(1)$.
When $E[f_1(U_i)^2]>0$ for at least one $f\in\calF$, the first-order projections dominates
determine a nontrivial Gaussian process at the $\sqrt n$ rate. This is
the standard nondegenerate case \citep{davezies2021,chiang2023inference};
see Lemma~\ref{lem:ddg}(i) of the supplement.

In the degenerate case, $f_1=0$ almost surely for every $f\in\calF$.
The second-order projection then determines the process at the
$\sqrt{N_n}$ rate. Independent dyads are a special case for which
the limit is Gaussian. Section~\ref{sec:general-process} establishes empirical process theory for the
general degenerate case, where the limit need not be Gaussian.

For the marks $f=r_{\bm x}$, the first-order node projection is
\begin{equation}
 \psi_{\bm x}(u)=(r_{\bm x})_1(u)
 =E[r_{\bm x}(\bm Z_{ij})\mid U_i=u]-\Delta(\bm x),
 \label{eq:first-order-node-projection}
\end{equation}
where $\Delta(\bm x)=Pr_{\bm x}=0$ under $H_0$.
Define the node and dyad covariance kernels by
\begin{align}
 \Omega(\bm x,\bm y)
 &=4E[\psi_{\bm x}(U_i)\psi_{\bm y}(U_i)],
 \label{eq:Omega}\\
 \Gamma(\bm x,\bm y)
 &=\Cov\{r_{\bm x}(\bm Z_{ij}),r_{\bm y}(\bm Z_{ij})\}.
 \label{eq:Gamma}
\end{align}

\begin{assumption}[Two Gaussian cases]\label{ass:nondeg}
Under $H_0$, one of the following holds:
(i) $\sup_{\bm x\in\calX}\Omega(\bm x,\bm x)>0$;
(ii) the collection $\{\bm Z_{ij}:i<j\}$ is i.i.d. and
$\sup_{\bm x\in\calX}\Gamma(\bm x,\bm x)>0$.
\end{assumption}

In the nondegenerate case (i), $\Omega(\bm x,\bm x)>0$ for at least
one $\bm x\in\calX$; it may be zero elsewhere. In case (ii), independent
dyads imply $\Omega(\bm x,\bm x)=0$ for every $\bm x$. This is a special
case of degeneracy with a nontrivial Gaussian limit at the $\sqrt{N_n}$ rate.

These cases parallel the Gaussian benchmarks studied in the
multiway-clustering literature; see, for example,
\citet[Section 3.2]{mackinnon2021wild} and
\citet[Assumption 3(iv)(1)-(2)]{chiang2023standard}.
Their results concern finite-dimensional regression parameters,
whereas our specification tests require approximation of the
residual process uniformly over $\bm x\in\calX$.
The two cases do not exhaust the possible limiting behavior
under dyadic dependence. Section~\ref{sec:nongaussian} discusses the case where limits need not be Gaussian.

\begin{lemma}[First-order expansion of the estimated process]\label{lem:linearization}
Under $H_0$ and Assumptions~\ref{ass:sampling}, \ref{ass:moments}, and~\ref{ass:nondeg},
let $\rho_n=n^{-1/2}$ in case (i) and $\rho_n=N_n^{-1/2}$ in case (ii).
Then
\begin{align}
 \sup_{\bm x\in\calX}|\Pn\widehat r(\bm x)-\Pn r_{\bm x}|
 =o_p(\rho_n).
 \label{eq:uniform-linearization}
\end{align}
\end{lemma}

Lemma~\ref{lem:linearization} shows that replacing
$\Pn\widehat r(\bm x)$ by $\Pn r_{\bm x}$ is asymptotically
negligible uniformly over $\bm x$ at the rate appropriate to each
case. In the nondegenerate case, \eqref{eq:main-projection} reduces
the analysis to the first-order projection based on the i.i.d.\
node variables $U_i$. In the independent-dyad case, standard i.i.d.\
empirical process theory applies directly to the dyad marks
$r_{\bm x}(\bm Z_{ij})$. These arguments yield the following
Gaussian limits.

\begin{theorem}[Gaussian limits of the estimated process]\label{thm:weak}
Under $H_0$ and Assumptions~\ref{ass:sampling}, \ref{ass:moments}, and~\ref{ass:nondeg}, in
$\ell^\infty(\calX)$,
\[
 \sqrt n\,\Pn\widehat r\weak\mathbb G_R
 \quad\text{in case (i)},\qquad
 \sqrt{N_n}\,\Pn\widehat r\weak\mathbb G_D
 \quad\text{in case (ii)},
\]
where $\mathbb G_R$ and $\mathbb G_D$ are centered Gaussian
processes with covariance kernels $\Omega$ and $\Gamma$ in \eqref{eq:Omega} and \eqref{eq:Gamma},
respectively.
\end{theorem}

The two limits have different rates and covariance kernels. The
following bootstrap approximates both without choosing between them.

\subsection{Corrected first-order multiplier bootstrap}\label{sec:corrected}

Write $E^*$ and $P^*$ for expectation and probability conditional on
the observations and any auxiliary grid randomization.

To motivate the bootstrap, center the estimated marks and form
their node averages:
\begin{equation}
 \widetilde r_{ij}(\bm x)
 =\widehat r_{ij}(\bm x)-\Pn\widehat r(\bm x),
 \qquad
 \widehat\psi_i(\bm x)
 =\frac1{n-1}\sum_{j\ne i}\widetilde r_{ij}(\bm x).
 \label{eq:psihat}
\end{equation}
In case (i), the leading term of the estimated process is
$2n^{-1}\sum_i\psi_{\bm x}(U_i)$.
The node average $\widehat\psi_i(\bm x)$ estimates the contribution
of node $i$. The uncorrected first-order multiplier bootstrap generates
\begin{align}\label{eq:uncorrected first}
 \frac2n\sum_{i=1}^n\xi_i\widehat\psi_i(\bm x),
 \qquad \xi_i\stackrel{\mathrm{iid}}{\sim}N(0,1),
\end{align}
with multipliers independent of the observations.
Using the same multipliers for every $\bm x$ reproduces the
dependence across values of the process.
The general multiplier result of
\citet[Theorem~3.3]{davezies2021}, together with the approximation
of the estimated marks, justifies this procedure in case (i).

This argument, however, does not provide a valid approximation at the
faster rate in case (ii). Each dyad enters two node averages,
and the resulting multiplier covariance, multiplied by $N_n$,
converges to $2\Gamma$ rather than $\Gamma$.
We therefore correct this excess covariance so that the same
procedure applies to both Gaussian cases.

The double-counting correction follows from an exact covariance identity.
Under $H_0$, define the population node averages
$\bar r_i(\bm x)=(n-1)^{-1}\sum_{j\ne i}
r_{\bm x}(\bm Z_{ij})$.
Separating identical dyads from distinct dyads sharing a node gives
\begin{align}
 \Cov(\Pn r_{\bm x},\Pn r_{\bm y})
 &=
 \frac{n-2}{n(n-1)}\Omega(\bm x,\bm y)
 +\frac{\Gamma(\bm x,\bm y)}{N_n},
 \nonumber\\
 E\left[
   \frac4{n^2}\sum_{i=1}^n
   \bar r_i(\bm x)\bar r_i(\bm y)
 \right]
 &=
 \Cov(\Pn r_{\bm x},\Pn r_{\bm y})
 +\frac{\Gamma(\bm x,\bm y)}{N_n}.
 \label{eq:correction-identity}
\end{align}
Each dyad appears in two node averages, so the multiplier covariance
counts $\Gamma(\bm x,\bm y)/N_n$ twice. Subtracting an estimate of
the extra contribution gives
\begin{align}
 \widehat V_0(\bm x,\bm y)
 &=
 \frac1{N_n}\sum_{i<j}
 \widetilde r_{ij}(\bm x)\widetilde r_{ij}(\bm y),
 \nonumber\\
 \widehat C_n(\bm x,\bm y)
 &=
 \frac4{n^2}\sum_{i=1}^n
 \widehat\psi_i(\bm x)\widehat\psi_i(\bm y)
 -\frac{\widehat V_0(\bm x,\bm y)}{N_n}.
 \label{eq:covariance-estimator}
\end{align}
Thus $\widehat C_n(\bm x,\bm y)$ estimates
$\Cov(\Pn r_{\bm x},\Pn r_{\bm y})$, accounting for the
double counting in the node-multiplier covariance.
The correction is asymptotically negligible under $\sqrt n$
normalization in case (i), but necessary under $\sqrt{N_n}$
normalization in case (ii). The counting argument is related
to the covariance estimators of \citet{cameron2011robust}, \citet{aronow2015cluster}, and \citet{chiang2023standard}.

Apply the following construction separately to each statistic.
The covariance matrix $\widehat{\bm C}_n$ has dimension $J\times J$, with
\[
 \widehat{\bm C}_n=
 \begin{cases}
  \bigl(\widehat C_n(\widehat{\bm x}_g,\widehat{\bm x}_h)\bigr)_{g,h=1}^{G_n},
   & \text{KS: }J=G_n,\\[3pt]
  \bigl(\widehat C_n(\bm X_{ij},\bm X_{k\ell})\bigr)_{i<j,\,k<\ell},
   & \text{CvM: }J=N_n.
 \end{cases}
\]
For CvM, rows and columns follow the dyad order
$(1,2),(1,3),\ldots,(n-1,n)$, retaining repeated covariate values.
We suppress the dependence of the matrix on the statistic.
The covariance subtraction can produce an indefinite matrix in a
finite sample. For Gaussian simulation, project it onto the
positive-semidefinite matrices in the entrywise maximum norm:\footnote{Setting
negative eigenvalues to zero can enlarge the error in individual
entries of a large matrix. The projection in \eqref{eq:psd-maxnorm}
bounds this error independently of the number of evaluation points.}
\begin{equation}
 \widehat{\bm C}_{n,+}\in
 \arg\min_{\bm C\succeq0}
     \|\bm C-\widehat{\bm C}_n\|_{\max},
 \qquad \|\bm A\|_{\max}=\max_{g,h}|A_{gh}|.
 \label{eq:psd-maxnorm}
\end{equation}
If there are several minimizers, select the one with smallest Frobenius
norm. The corrected first-order multiplier bootstrap generates a Gaussian
vector with this covariance matrix.

\begin{algorithm}[Corrected first-order multiplier bootstrap]
\label{alg:corrected}
\leavevmode
\small

\begin{enumerate}[label=\arabic*.,leftmargin=*]
\item 
Compute $\widehat{\bm\theta}$, the residuals
$\widehat\varepsilon_{ij}$, and the estimated marks
$\widehat r_{ij}(\bm x)$.

\item 
Construct $\widehat{\calX}_{G_n}$ using
\eqref{eq:random-grid}, and compute
$T_n^{\mathrm{KS}}$ and $T_n^{\mathrm{CvM}}$
from \eqref{eq:statistics}. 
\item 
For each statistic, compute its $\widehat{\bm C}_{n,+}$ from
\eqref{eq:psd-maxnorm}.
For $b=1,\ldots,B$, independently generate
$\bm z_b\sim N(\bm0,\bm I_J)$, where $J=G_n$ for KS and $J=N_n$
for CvM, and set
\begin{equation}
 \begin{aligned}
 \bigl(R_{n,b}^{*}(\widehat{\bm x}_1),\ldots,
 R_{n,b}^{*}(\widehat{\bm x}_{G_n})\bigr)^\top
 &=\widehat{\bm C}_{n,+}^{1/2}\bm z_b
 &&\text{for KS},\\
 \bigl(R_{n,b}^{*}(\bm X_{12}),\ldots,
 R_{n,b}^{*}(\bm X_{n-1,n})\bigr)^\top
 &=\widehat{\bm C}_{n,+}^{1/2}\bm z_b
 &&\text{for CvM}.
 \end{aligned}
 \label{eq:corrected-draw}
\end{equation}
 Compute
\(
 T_{n,b}^{\mathrm{KS},*}
 =\max_{1\le g\le G_n}
    |R_{n,b}^{*}(\widehat{\bm x}_g)| \)
and \( T_{n,b}^{\mathrm{CvM},*}
 =\frac1{N_n}\sum_{i<j}
    \{R_{n,b}^{*}(\bm X_{ij})\}^{2}.
\)
\item 
For $S\in\{\mathrm{KS},\mathrm{CvM}\}$, reject $H_0$ if $T_n^S>c_n^{S,*}(1-\alpha)$, where the bootstrap
critical value 
\begin{equation}
 c_n^{S,*}(1-\alpha)
 =\inf\left\{t\in\R:
   \frac1B\sum_{b=1}^B\ind\{T_{n,b}^{S,*}\le t\}
   \ge1-\alpha\right\}.
 \label{eq:empirical-critical}
\end{equation}

\end{enumerate}
\end{algorithm}

The algorithm requires neither a choice between cases (i) and
(ii) nor an estimate of the convergence rate. In the asymptotic results
below, the number of draws satisfies $B\to\infty$ as $n\to\infty$.
Rejection probabilities include randomness from the observations,
the grid, and the bootstrap draws.

\subsection{Validity and power}

Let $\Phi_{\mathrm{KS}}(z)=\sup_{\bm x\in\calX}|z(\bm x)|$ and
$\Phi_{\mathrm{CvM}}(z)=\int z(\bm x)^2dF_{\bm X}(\bm x)$.
The following condition concerns the statistic and significance level
actually used.

\begin{assumption}[Critical values]\label{ass:critical}
The distribution of
$\Phi_S(\mathbb G_R)$ in case (i), or $\Phi_S(\mathbb G_D)$ in case (ii),
is continuous and strictly increasing at its positive $(1-\alpha)$ quantile.
\end{assumption}

Assumption~\ref{ass:critical} makes the critical value stable under approximation.
It permits zero variance at some values of $\bm x$ and singular
covariance matrices. For CvM, variation must occur on a set with
positive $F_{\bm X}$ mass.

\begin{theorem}[Validity of the corrected first-order multiplier bootstrap]
\label{thm:bootstrap-tests}
Suppose $H_0$ and Assumptions~\ref{ass:sampling}-\ref{ass:critical}
hold, and $B\to\infty$. Then 
\[
 P\{T_n^S>c_n^{S,*}(1-\alpha)\}\to\alpha,
 \qquad S\in\{\mathrm{KS},\mathrm{CvM}\}.
\]
\end{theorem}

Theorem~\ref{thm:bootstrap-tests} justifies the simulated critical
values in both Gaussian cases, with KS evaluated on the growing
data-dependent grid and CvM averaged over the observed covariate vectors.

\begin{theorem}[Consistency against fixed alternatives]
\label{thm:consistency}
Suppose Assumptions~\ref{ass:sampling}, \ref{ass:moments}, and
\ref{ass:grid} hold, $\sup_{\bm x\in\calX}|\Delta(\bm x)|>0$,
and $B\to\infty$. Then
$P\{T_n^{\mathrm{KS}}>c_n^{\mathrm{KS},*}(1-\alpha)\}\to1$.
If, in addition, $\int\Delta(\bm x)^2\,dF_{\bm X}(\bm x)>0$,
then
$P\{T_n^{\mathrm{CvM}}>c_n^{\mathrm{CvM},*}(1-\alpha)\}\to1$.
\end{theorem}

Assumptions~\ref{ass:nondeg} and~\ref{ass:critical} are not
required for fixed-alternative consistency: the statistic
converges to a positive population value under the remaining assumptions, while the bootstrap
critical value converges to zero. With multivariate covariates, $\Delta(\bm x)\ne0$ at some $\bm x$
need not imply a positive CvM integral $\int\Delta(\bm x)^2\,dF_{\bm X}(\bm x)>0$, which explains the additional
condition for CvM consistency. The KS grid evaluates combinations
of marginal quantiles that need not lie in the joint support of
$\bm X$. It can therefore detect fixed alternatives for which
the population CvM criterion is zero.

\subsection{Local asymptotic power}

Starting from a null model in either case of Assumption~\ref{ass:nondeg},
consider local alternatives formed from the same covariates and disturbances:
\begin{equation}
 Y_{ij,n}=m(\bm{X}_{ij},\bm\theta_0)+\varepsilon_{ij,0}
                  +\rho_n d(\bm{X}_{ij}),
 \label{eq:local-model}
\end{equation}
where $E[\varepsilon_{ij,0}\mid\bm X_{ij}]=0$, $d$ is bounded,
$\rho_n=n^{-1/2}$ in case (i), and
$\rho_n=N_n^{-1/2}$ in case (ii), as in Lemma~\ref{lem:linearization}. Assume the smoothness
conditions continue to hold in a neighborhood of $\bm\theta_0$.
Define
\[
 \mu_d(\bm{x})=E[d(\bm{X})q_{\bm{x}}(\bm{X})],
\]
where $q_{\bm x}$ and the expectation are evaluated under the null
baseline at $\bm\theta_0$.
After accounting for parameter estimation, the local departure
$\rho_n d(\bm X_{ij})$ changes the moment
$\Delta(\bm x)$ by $\rho_n\mu_d(\bm x)$ to first order. Let $\mathbb Z$ denote the centered Gaussian limit under the null,
and let $c_{S,1-\alpha}$ be the $(1-\alpha)$ quantile of
$\Phi_S(\mathbb Z)$, for $S\in\{\mathrm{KS},\mathrm{CvM}\}$.

\begin{theorem}[Local power in the two Gaussian cases]
\label{thm:local}
Suppose \eqref{eq:local-model} holds, the null baseline satisfies
Assumptions~\ref{ass:sampling}-\ref{ass:critical}, and $B\to\infty$.
Then $\rho_n^{-1}\Pn\widehat r\weak\mathbb Z+\mu_d$ in
$\ell^\infty(\calX)$. If
$P\{\Phi_S(\mathbb Z+\mu_d)=c_{S,1-\alpha}\}=0$, then
\begin{equation}
 P\{T_n^S>c_n^{S,*}(1-\alpha)\}
 \to
 P\{\Phi_S(\mathbb Z+\mu_d)>c_{S,1-\alpha}\}
 \ge\alpha.
 \label{eq:local-power}
\end{equation}
\end{theorem}

\begin{remark}
The inequality in \eqref{eq:local-power} is strict when the
statistic detects the local change represented by $\mu_d$.
For CvM, a sufficient condition is
$\int\mu_d(\bm x)^2\,dF_{\bm X}(\bm x)>0$.
For KS, it suffices that $\mu_d(\bm x)\ne0$ for some $\bm x$
and that the local change follows directions in which the
limiting Gaussian process varies. The precise KS condition
is given in \eqref{eq:local-drift} of the supplement.
\end{remark}

\section{Second-Order Theory and Inference}\label{sec:nongaussian}

We now study the general degenerate case beyond independent dyads. With
dependent dyads, node interactions can lead to a non-Gaussian limit.
We first establish a second-order multiplier bootstrap approximation
for general degenerate dyadic empirical processes. For specification
testing, we then show that two additional non-Gaussian terms can enter the limit:
one from parameter estimation and one from the curvature of a nonlinear
regression function. We augment the bootstrap to account for both terms
and take the larger of its critical value and the corrected first-order
critical value from Section~\ref{sec:bootstrap}.

\subsection{Theory for degenerate dyadic empirical processes}\label{sec:general-process}

In the degenerate case of \eqref{eq:main-projection}, $f_1=0$ almost surely for every $f\in\calF$.
To characterize second-order limit of $\sqrt{N_n}(\Pn f-Pf)= \sqrt{N_n}\mathbb U_n f_2$, write
\begin{equation}
 f_2(U_i,U_j,U_{ij})
 =\underbrace{\gamma(U_i,U_j,f)}_{\text{node interaction}}
  +\underbrace{e_{ij}(f)}_{\text{dyad noise}},
 \label{eq:components}
\end{equation}
where
\begin{align}
 \gamma(u,v,f)
 &=E[f_2(U_i,U_j,U_{ij})\mid U_i=u,U_j=v],\\*
 e_{ij}(f)
 &=f_2(U_i,U_j,U_{ij})-\gamma(U_i,U_j,f).
 \label{eq:components}
\end{align}
The node interaction $\gamma(\cdot,\cdot,f)$ determines the potentially
non-Gaussian component of the limit. Conditional on the node variables,
the dyad-specific terms $e_{ij}(f)$ are independent and centered.
Their normalized sum contributes the Gaussian component.

To describe the interaction limit, let
$(\phi_k)_{k\ge1}$ be a common orthonormal basis of the centered
node $L_2$ space. Since $\gamma(\cdot,\cdot,f)$ is square-integrable and has zero
conditional mean given either node, it admits the expansion
\begin{equation}
 \begin{aligned}
 \gamma(u,v,f)
 &=\sum_{k,l\ge1}\beta_{kl}(f)\phi_k(u)\phi_l(v),\qquad
 \beta_{kl}(f)
 =E[\gamma(U_i,U_j,f)\phi_k(U_i)\phi_l(U_j)].
 \end{aligned}
 \label{eq:interaction-expansion}
\end{equation}
The series converges in mean square under the distribution of
the independent nodes $U_i,U_j$, with
$\sum_{k,l\ge1}\beta_{kl}(f)^2=E[\gamma(U_i,U_j,f)^2]<\infty$.
Thus each interaction kernel can be approximated by finite sums
of products of centered node functions. Such expansions are
standard tools for analyzing degenerate interactions; see, for
example, \citet{menzel2021bootstrap} for interactions between
row and column effects. In our setting, the
mean-square approximation follows from square-integrability
and requires no additional finite-rank assumption.

Lemma~\ref{lem:ddg}(ii) of the supplement gives the limit of
$\sqrt{N_n}\mathbb U_n f_2$ as
$\mathbb L(f)=\mathbb W(\gamma(\cdot,\cdot,f))+\mathbb G_e(f)$, where
\begin{equation}
 \mathbb W(\gamma(\cdot,\cdot,f))
 =\frac1{\sqrt2}\sum_{k,l\ge1}
 \beta_{kl}(f)(Z_kZ_l-\ind\{k=l\}),
 \label{eq:chaos}
\end{equation}
and the $Z_k$ are independent standard normal variables. The
covariance between the interaction limits indexed by two functions $f$ and $h$ is
$E[\gamma(U_i,U_j,f)\gamma(U_i,U_j,h)]$.
The process $\mathbb G_e$ is centered Gaussian, independent of
$\mathbb W$, with covariance $E[e_{ij}(f)e_{ij}(h)]$. 

Two difficulties arise in approximating the limiting distribution.
First, the node interaction can make the limit non-Gaussian, so its
covariance kernel alone does not determine its law. Second, the data
do not separately reveal the node interaction and the dyad-specific
component. We therefore introduce a second-order multiplier bootstrap
that approximates the joint law of $\{\mathbb L(f):f\in\calF\}$
directly from the observed dyads. The bootstrap assigns one multiplier
to each node and weights each dyad by the product of its two node
multipliers. Specifically, let $\xi_1,\ldots,\xi_n$ be i.i.d.\
standard normal variables, independent of the observations, and define
\begin{equation}
 \mathbb Q_n^*(f)
 =\frac1{\sqrt{N_n}}\sum_{i<j}
 \xi_i\xi_j\{f(\bm Z_{ij})-\Pn f\}.
 \label{eq:general-quadratic}
\end{equation}

\begin{theorem}[Second-order multiplier bootstrap approximation]
\label{thm:quadratic}
Under Assumption~\ref{ass:sampling}, let $\calF$ be VC type
with envelope $F$ satisfying $PF^q<\infty$ for some $q>4$.
If $f_1=0$ for every $f\in\calF$, then, conditionally on the
observations,
\[
 \mathbb Q_n^*\weak\mathbb L
 \quad\text{in }\ell^\infty(\calF),\quad\text{in probability}.
\]
\end{theorem}

Theorem~\ref{thm:quadratic} treats $f$ as known. In the specification
test, $r_{\bm x}$ depends on unknown parameters and moments. The
following corollary permits estimated indexing functions more generally.

\begin{corollary}[Estimated indexing functions]
\label{cor:quadratic-plugin}
Under Assumption~\ref{ass:sampling}, let $\calH$ be a fixed VC-type
class with envelope $H$ satisfying $PH^q<\infty$ for some $q>4$,
and let $\calF\subset\calH$. For each $f\in\calF$, let
$\widehat f_n(f)$ be an estimate constructed from the observations.
Suppose that
\(
 P\bigl(\{\widehat f_n(f):f\in\calF\}\subset\calH\bigr)\to1,
 \ 
 \sup_{f\in\calF}\|\widehat f_n(f)-f\|_{P,2}=o_p(1).
\)
Define $\mathbb Q_n^*(\widehat f_n(f))$ by replacing $f$
throughout \eqref{eq:general-quadratic} with $\widehat f_n(f)$. Then
\begin{equation}
 E^*\sup_{f\in\calF}
 \left|\mathbb Q_n^*(\widehat f_n(f))-\mathbb Q_n^*(f)\right|
 =o_p(1).
 \label{eq:quadratic-plugin-general}
\end{equation}
\end{corollary}

Theorem~\ref{thm:quadratic} and Corollary~\ref{cor:quadratic-plugin}
complement the literature by establishing a bootstrap approximation
for degenerate dyadic empirical processes that holds uniformly over
$f\in\calF$. \citet{davezies2021} characterize the sampling limit
under degeneracy, but their bootstrap does not generally reproduce it.
The Gaussian approximation of \citet{chiang2023inference} requires a
nondegenerate first-order projection. \citet{han2022multiplier}
studies product multipliers for U-processes based on i.i.d.\
observations. In our setting, each dyad may also contain variation
specific to that pair. The same node multipliers reproduce the joint
law of the Gaussian component from this variation and the potentially
non-Gaussian component from node interactions, even when the indexing
functions are estimated.

\subsection{A second-order expansion of the estimated process}

Section~\ref{sec:general-process} gives the limit of $\sqrt{N_n}\Pn r_{\bm x}$.
The tests use $\Pn\widehat r(\bm x)$ instead. In the degenerate case,
the difference between these processes need not be negligible at
the $\sqrt{N_n}$ rate: parameter estimation and nonlinear curvature
can both affect the limit. We state the degenerate condition in terms of the
node covariance kernel $\Omega$ in \eqref{eq:Omega}.
\begin{assumption}[Degenerate case for the marks]\label{ass:zero}
Under $H_0$, $\sup_{\bm x\in\calX}\Omega(\bm x,\bm x)=0$.
\end{assumption}
Thus the node projection of $r_{\bm x}$ is zero for every $\bm x$.
Under $H_0$, this condition and case (i) of Assumption~\ref{ass:nondeg}
exhaust the possibilities for the marks. Independent dyads in case (ii)
are a special case of Assumption~\ref{ass:zero}.
Let $\bm\delta_n=\widehat{\bm\theta}-\bm\theta_*$ and define
\begin{equation}
 \bm H(\bm X)=\frac{\partial^2m(\bm X,\bm\theta_*)}
 {\partial\bm\theta\partial\bm\theta^\top},\qquad
 \bm D_n(\bm x)=\Pn[\bm gq_{\bm x}],\qquad
 \bm C(\bm x)=E[\bm Hq_{\bm x}].
 \label{eq:second-order-objects}
\end{equation}
Here $\bm H$ is the Hessian of $m$, $\bm D_n(\bm x)$ has mean
zero by \eqref{eq:pop-mark}, and $\bm C(\bm x)$ is the curvature
coefficient in the following expansion.

\begin{proposition}[Second-order estimation effects]\label{prop:second-order}
Suppose $H_0$ and Assumptions~\ref{ass:sampling}, \ref{ass:moments},
and~\ref{ass:zero} hold. Then $\bm\delta_n=O_p(n^{-1/2})$ and, uniformly in $\bm x$,
\begin{equation}
 \Pn\widehat r(\bm x)
 =\Pn r_{\bm x}
   -\bm D_n(\bm x)^\top\bm\delta_n
   -\frac12\bm\delta_n^\top\bm C(\bm x)\bm\delta_n
   +o_p(n^{-1}).
 \label{eq:second-order-expansion}
\end{equation}
If $m(\bm X,\bm\theta)=\bm g(\bm X)^\top\bm\theta$, the interior
score equation yields the exact identity
$\Pn\widehat r(\bm x)=\Pn r_{\bm x}-\bm D_n(\bm x)^\top\bm\delta_n$
whenever $\widehat{\bm Q}$ is nonsingular.
\end{proposition}

\begin{remark}
Assumption~\ref{ass:zero} does not imply $\bm\delta_n=O_p(N_n^{-1/2})$.
It restricts the node projection of the mark $r_{\bm x}$, whereas
the rate of parameter estimation $\bm\delta_n$ depends on the node projection of the
score $\bm g(\bm X_{ij})\varepsilon_{ij}$. If the first-order score projection is
nonzero, $\bm\delta_n$ is of order $n^{-1/2}$ even though $\Pn r_{\bm x}$
is of order $N_n^{-1/2}$. Parameter estimation can then affect the
second-order limit of the specification test.
\end{remark}
\begin{remark}
Proposition~\ref{prop:second-order} identifies three potential sources
of non-Gaussianity: the node interaction contained in $\Pn r_{\bm x}$,
the parameter-estimation term $-\bm D_n(\bm x)^\top\bm\delta_n$,
and the curvature term
$-\bm\delta_n^\top\bm C(\bm x)\bm\delta_n/2$ arising from
nonlinearity in $m$. Each can contribute at order $n^{-1}$.
The last two terms share $\bm\delta_n$, and both depend on the same
nodes as the interaction in $\Pn r_{\bm x}$. Consequently, the
three contributions need not be independent and can partly or fully
cancel. The inference must preserve their joint distribution.
\end{remark}

\subsection{Augmented second-order multiplier bootstrap}\label{sec:augmented}

We augment the second-order multiplier bootstrap of
Section~\ref{sec:general-process} to account for the parameter-estimation
and curvature terms in Proposition~\ref{prop:second-order}. We refer
to the resulting procedure as the \emph{augmented second-order multiplier
bootstrap}.

First, average the estimated score and adjusted covariate moments over
the dyads containing each node:
\begin{equation}
 \widehat{\bm s}_i
 =\frac1{n-1}\sum_{j\ne i}
      \widehat{\bm g}_{ij}\widehat\varepsilon_{ij},
 \qquad
 \widehat{\bm d}_i(\bm x)
 =\frac1{n-1}\sum_{j\ne i}
      \widehat{\bm g}_{ij}\widehat q_{\bm x}(\bm X_{ij}).
 \label{eq:estimated-node-moments}
\end{equation}
These averages estimate the node contributions to parameter estimation
and to $\bm D_n(\bm x)$. Their averages across nodes are zero by the
sample score equations and the definition of $\widehat q_{\bm x}$.
Estimate curvature by
\[
 \widehat{\bm C}(\bm x)=\Pn[\widehat{\bm H}\widehat q_{\bm x}],
 \qquad
 \widehat{\bm H}_{ij}
 =\frac{\partial^2m(\bm X_{ij},\widehat{\bm\theta})}
        {\partial\bm\theta\partial\bm\theta^\top}.
\]

Draw $\xi_1,\ldots,\xi_n$ independently from $N(0,1)$, independently
of the observations, and set
\begin{equation}
 \bm\delta_n^*
 =\frac2n\widehat{\bm Q}^{-1}\sum_i\xi_i\widehat{\bm s}_i,
 \qquad
 \bm D_n^*(\bm x)
 =\frac2n\sum_i\xi_i\widehat{\bm d}_i(\bm x).
 \label{eq:node-estimation-draws}
\end{equation}
Using the same multipliers in every term, form the augmented draw
\begin{equation}
 \begin{aligned}
 R_n^{A,*}(\bm x)
 ={}&\frac1{N_n}\sum_{i<j}\xi_i\xi_j\widetilde r_{ij}(\bm x)
       -\bm D_n^*(\bm x)^\top\bm\delta_n^*-\frac12(\bm\delta_n^*)^\top
                  \widehat{\bm C}(\bm x)\bm\delta_n^*.
 \end{aligned}
 \label{eq:curvature-bootstrap}
\end{equation}
The first sum reproduces the node interaction and dyad-specific
variation described in Section~\ref{sec:general-process}. The two
subtractions account for the parameter-estimation and curvature
terms in Proposition~\ref{prop:second-order}. Using the same
multipliers $\{\xi_i\}$ in all three terms is essential to reproduce their
joint limiting distribution, including possible cancellation.
The products are left uncentered to retain their potentially
nonzero means. For linear regression,
$m(\bm X,\bm\theta)=\bm g(\bm X)^\top\bm\theta$ implies
$\widehat{\bm H}=\bm0$, so the curvature term in
\eqref{eq:curvature-bootstrap} is omitted.

Keep the KS grid fixed and generate $B$ independent repetitions of
\eqref{eq:curvature-bootstrap}. For each repetition $b$, calculate
\begin{equation}
 \begin{aligned}
 T_{n,b}^{\mathrm{KS},A,*}
 &=\max_{\bm x\in\widehat{\calX}_{G_n}}
          |R_{n,b}^{A,*}(\bm x)|,\\
 T_{n,b}^{\mathrm{CvM},A,*}
 &=\frac1{N_n}\sum_{i<j}R_{n,b}^{A,*}(\bm X_{ij})^2.
 \end{aligned}
 \label{eq:quadratic-statistics}
\end{equation}
The corresponding empirical critical value is
\begin{equation}
 c_{A,n}^{S,*}(1-\alpha)
 =\inf\left\{t\in\R:
     \frac1B\sum_{b=1}^B\ind\{T_{n,b}^{S,A,*}\le t\}
     \ge1-\alpha\right\}.
 \label{eq:additional-critical}
\end{equation}
\subsection{Envelope inference}\label{sec:envelope}
Table~\ref{tab:bootstrap-validity} summarizes which methods are exact
under the conditions below. The corrected first-order multiplier
bootstrap covers the nondegenerate case and independent dyads; the
augmented second-order multiplier bootstrap covers the degenerate
case, including independent dyads. The \emph{envelope test} rejects $H_0$ when $T_n^S$ exceeds
\begin{equation}
 c_{E,n}^{S,*}(1-\alpha)
 = \max\bigl\{c_n^{S,*}(1-\alpha),\,c_{A,n}^{S,*}(1-\alpha)\bigr\},
 \qquad S\in\{\mathrm{KS},\mathrm{CvM}\}.
 \label{eq:envelope-critical}
\end{equation}

\begin{table}[t!]
\centering 
\caption{Bootstrap validity and asymptotic size of the envelope test}
\label{tab:bootstrap-validity}
\begin{tabular}{lll}
\toprule
Case & Asymptotically exact methods & Envelope size \\
\midrule
Nondegenerate case & Uncorrected, Corrected & $\to\alpha$\\[1mm]
Independent dyads (degenerate) & Corrected, Augmented & $\to\alpha$\\[1mm]
Other degenerate cases & Augmented & $\limsup\le\alpha$\\
\bottomrule
\end{tabular}
\end{table}

In the nondegenerate case, the augmented critical value is of smaller
order than the corrected first-order critical value. Taking their maximum
therefore leaves the test asymptotically unchanged.\footnote{Used alone,
the augmented bootstrap would overreject in the nondegenerate case.}
Under independent dyads, both critical values, after normalization,
estimate the same Gaussian quantile. In other degenerate cases, the
maximum is at least as large as the valid augmented critical value;
no comparison of the Gaussian and non-Gaussian limits is needed.

Let $\mathbb T$ denote the limit of
$\sqrt{N_n}\Pn\widehat r$ under Assumption~\ref{ass:zero}, derived in
Lemma~\ref{lem:SD-fitted} of the supplement.
\begin{theorem}[Validity of the envelope test]\label{thm:envelope}
Suppose $H_0$ and Assumptions~\ref{ass:sampling}, \ref{ass:moments},
and \ref{ass:grid} hold, and $B\to\infty$.
For $S\in\{\mathrm{KS},\mathrm{CvM}\}$:
\begin{enumerate}[label=(\roman*),leftmargin=*]
\item Under Assumptions~\ref{ass:nondeg} and \ref{ass:critical},
\(
 P\{T_n^S>c_{E,n}^{S,*}(1-\alpha)\}\to\alpha.
\)
\item Under Assumption~\ref{ass:zero}, if the distribution of
$\Phi_S(\mathbb T)$ is continuous and strictly increasing at its
positive $(1-\alpha)$ quantile, then \[
 P\{T_n^S>c_{A,n}^{S,*}(1-\alpha)\}\to\alpha,
 \qquad
 \limsup_{n\to\infty}
 P\{T_n^S>c_{E,n}^{S,*}(1-\alpha)\}\le\alpha.
\]
\end{enumerate}
\end{theorem}

Theorem~\ref{thm:envelope} gives the envelope test asymptotic size
$\alpha$ in the two Gaussian cases and at most $\alpha$ in the other
degenerate cases covered by part (ii).\footnote{Asymptotic exactness is
pointwise in the data-generating process. This is consistent with the
impossibility argument in the proof of Proposition~4.1 of
\citet{menzel2021bootstrap}.}
Section~\ref{app:envelope-inference} of the supplement establishes that the
envelope test also retains the fixed-alternative consistency and local
power properties in Theorems~\ref{thm:consistency} and~\ref{thm:local}.

\section{Monte Carlo Simulation}\label{sec:simulation}

\subsection{Simulation designs}

We generate $Y_{ij}=(1+X_{ij})^2+\varepsilon_{ij}$, where
\begin{equation}
 \begin{aligned}
 X_{ij}&=\omega(U_i^x+U_j^x)+U_{ij}^x,\\
 \varepsilon_{ij}&=(1+X_{ij})\{\lambda\omega(U_i^\varepsilon+U_j^\varepsilon)
                         +U_{ij}^\varepsilon\}.
 \end{aligned}
 \label{eq:sim-squared}
\end{equation}
The variables $U_i^x$ take values $-1$ and $1$ with equal probability,
and $U_{ij}^x$ is the sum of two independent variables with this same
distribution. Thus $U_{ij}^x$ takes values $-2,0,2$ with probabilities
$1/4,1/2,1/4$. The error variables $U_i^\varepsilon$ and
$U_{ij}^\varepsilon$ are uniform on $(-1,1)$, and all primitive variables
are mutually independent. We estimate
$m(X,\bm\theta)=(\theta_0+\theta_1X)^2$ by nonlinear least squares over
$\Theta=[1/2,3/2]^2$. The null holds at $\bm\theta_*=(1,1)^\top$.

The three baseline designs differ only in $\omega$:
\begin{enumerate}[label=(\roman*),leftmargin=*]
\item Independent dyads ($\omega=0$): the dyads are i.i.d., a special
case of degeneracy satisfying Assumption~\ref{ass:nondeg}(ii).
\item Nondegenerate weak Gaussian ($\omega=1/2$):
Assumption~\ref{ass:nondeg}(i) holds.
\item Degenerate non-Gaussian ($\omega=1$):
Assumption~\ref{ass:zero} holds. This design includes a node
interaction, parameter estimation, and nonlinear curvature.
\end{enumerate}

For the nondegenerate strong Gaussian design, keep
$\omega=1/2$ and replace only the error loading:
\begin{equation}
 \varepsilon_{ij}=(1+X_{ij}^2)
 \{\tfrac\lambda2(U_i^\varepsilon+U_j^\varepsilon)+U_{ij}^\varepsilon\}.
 \label{eq:sim-strong}
\end{equation}
This design also satisfies Assumption~\ref{ass:nondeg}(i), with stronger
node variation than in design (ii). Section~\ref{app:simulation-checks}
of the supplement verifies the projection conditions and identification
of $\bm\theta_*$.

We set $\lambda\in\{1,4,10\}$. Increasing $\lambda$ gives more weight
to the node errors without changing the variance regime. In the
non-Gaussian design, the non-Gaussian component accounts for 3.9\%,
37.1\%, and 78.5\% of the integrated limiting variance, respectively
(Section~\ref{app:simulation-checks} of the supplement). When
$\omega=0$, $\lambda$ drops out, so we report the independent-dyad
design once.

We examine size at the 5\% level for
$n\in\{30,40,60,90,120,150,200\}$. Each of the 70 settings uses
20,000 replications and 999 bootstrap draws per replication. 
The KS grid consists of the distinct empirical quantiles at
$\max\{64,\lceil n^{1/4}\rceil\}$ equally spaced levels in $(0,1]$,
including $1$.\footnote{This choice satisfies Assumption~\ref{ass:grid} and is simple to
implement. The results in these designs are insensitive to the
precise grid specification.} For the corrected draws, we set negative eigenvalues of
the covariance matrix on this grid to zero. Both first-order procedures
compute CvM by assigning each observation to the grid point nearest
in empirical cdf value (ties go to the smaller point), or to zero if
the lower boundary is at least as close. The observed and augmented
CvM statistics average over all dyads. Section~\ref{app:mc-implementation}
of the supplement gives implementation details.\footnote{On boundary fits, we use \(\widehat R_n(\bm x)\) for the observed process and to center the bootstrap marks, as detailed in Section SE.2 of the supplement.}

The figures compare four procedures: the uncorrected first-order
multiplier bootstrap in \eqref{eq:uncorrected first}, the corrected
first-order multiplier bootstrap in Algorithm~\ref{alg:corrected},
the augmented second-order multiplier bootstrap in
\eqref{eq:curvature-bootstrap}, and the envelope test in
\eqref{eq:envelope-critical}. We label them Uncorrected, Corrected,
Augmented, and Envelope, respectively.

\subsection{Finite-sample size}

Figures~\ref{fig:size-independent}-\ref{fig:size-nong} report the four
procedures over the seven sample sizes. Tables~\ref{tab:size-independent}-\ref{tab:size-nong}
in the supplement give all numerical results.

\begin{figure}[!t]\centering
\includegraphics[width=\textwidth]{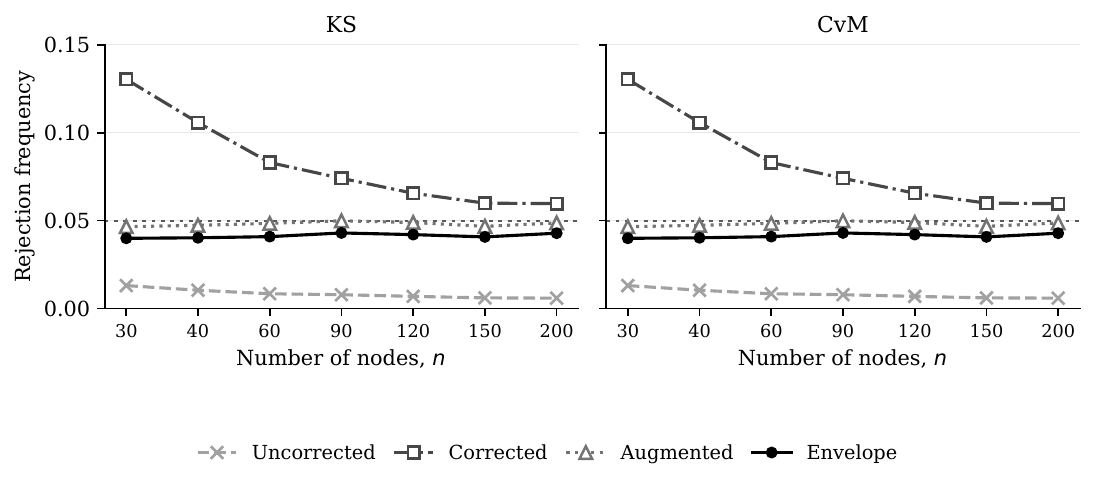}
\caption{Independent dyads. Panels show KS and CvM; $\lambda$ does not affect this design. Each point uses 20,000 replications and 999 draws.
The horizontal dotted line marks the nominal 5\% level.}
\label{fig:size-independent}
\end{figure}

Under independent dyads (Figure~\ref{fig:size-independent}), the
uncorrected procedure counts each dyad contribution twice in estimating
the covariance. It therefore rejects too rarely even at $n=200$:
the KS rejection frequency is 0.006. After correction, this frequency
is 0.060; the augmented and envelope frequencies are 0.049 and 0.043.
The latter three procedures have the same asymptotic size in this
special case of degeneracy, although they differ in finite samples.

\begin{figure}[!t]
\centering
\includegraphics[
  width=\textwidth,
  height=0.4\textheight
]{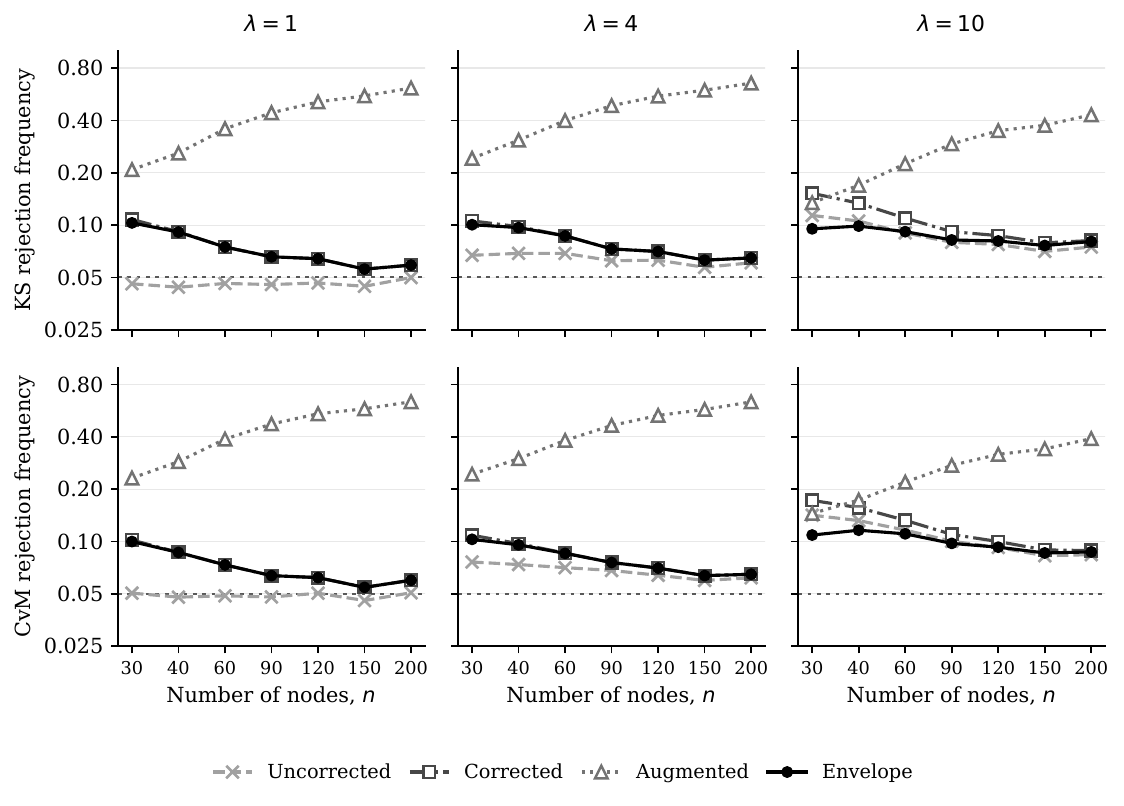}
\caption{Nondegenerate strong Gaussian. Columns vary $\lambda$;
rows show KS and CvM. Each point uses 20,000 replications and
999 draws. The vertical axis is logarithmic; tick labels give
rejection frequencies. The horizontal dotted line marks the
nominal 5\% level.}
\label{fig:size-strong}

\medskip

\includegraphics[
  width=\textwidth,
  height=0.4\textheight
]{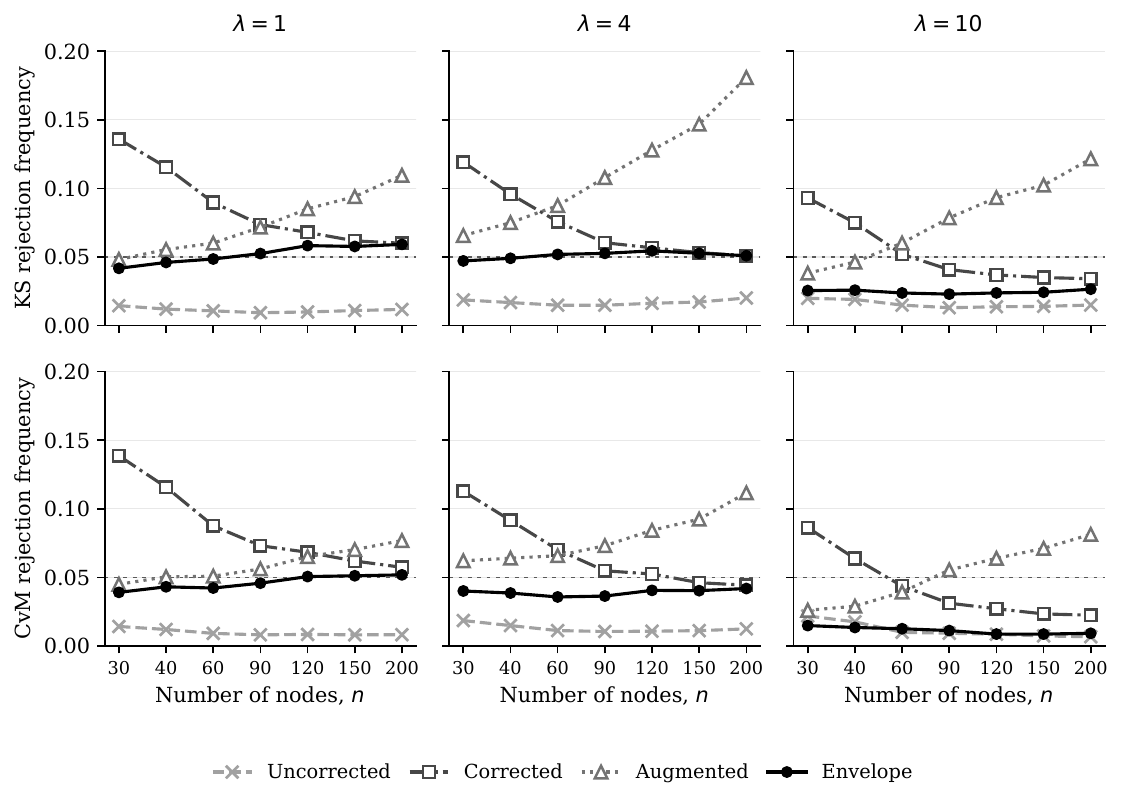}
\caption{Nondegenerate weak Gaussian. Columns vary $\lambda$;
rows show KS and CvM. Each point uses 20,000 replications and
999 draws. The horizontal dotted line marks the nominal 5\% level.}
\label{fig:size-weak}
\end{figure}

In the nondegenerate strong Gaussian design
(Figure~\ref{fig:size-strong}), the uncorrected and corrected
procedures become closer as $n$ increases. At $n=200$ and $\lambda=1$,
their KS rejection frequencies are 0.050 and 0.059. The envelope
frequency is also 0.059, consistent with the corrected critical value
determining the envelope in the nondegenerate case. At $\lambda=10$,
all three procedures still overreject at $n=200$. The augmented procedure
alone rejects increasingly often as $n$ grows: its critical value is
of smaller order than the statistic in this case.

The nondegenerate weak Gaussian design (Figure~\ref{fig:size-weak})
shows why the covariance correction can matter even when the
uncorrected procedure is asymptotically valid. At $n=200$ and
$\lambda=1$, the uncorrected and corrected KS rejection frequencies are
0.012 and 0.060. At that sample size, dyads sharing a node account for
only 23.4\% of the integrated variance in this design, compared with
93.6\% in the strong design (Table~\ref{tab:node-share} of the
supplement). Counting the individual dyad contribution twice therefore
has a much larger effect in the weak design.

In the degenerate non-Gaussian design (Figure~\ref{fig:size-nong}),
the augmented KS rejection frequencies at $n=200$ are 0.051, 0.052,
and 0.044 for $\lambda=1,4,10$, respectively. CvM is more conservative
at $\lambda=10$: its augmented rejection frequency rises from 0.012
at $n=30$ to 0.030 at $n=200$. The envelope rejects less often than
the augmented procedure, as the critical-value rule implies. Neither
first-order procedure stays close to 5\% across settings. For example,
at $\lambda=4$ the corrected CvM frequency falls from 0.054 at $n=90$
to 0.037 at $n=200$. The estimation terms also matter:
at $\lambda=10$ and $n=200$, omitting both terms raises the KS
frequency from 0.044 to 0.247 (Table~\ref{tab:second-order} of the
supplement).

\begin{figure}[!t]\centering
\includegraphics[width=\textwidth]{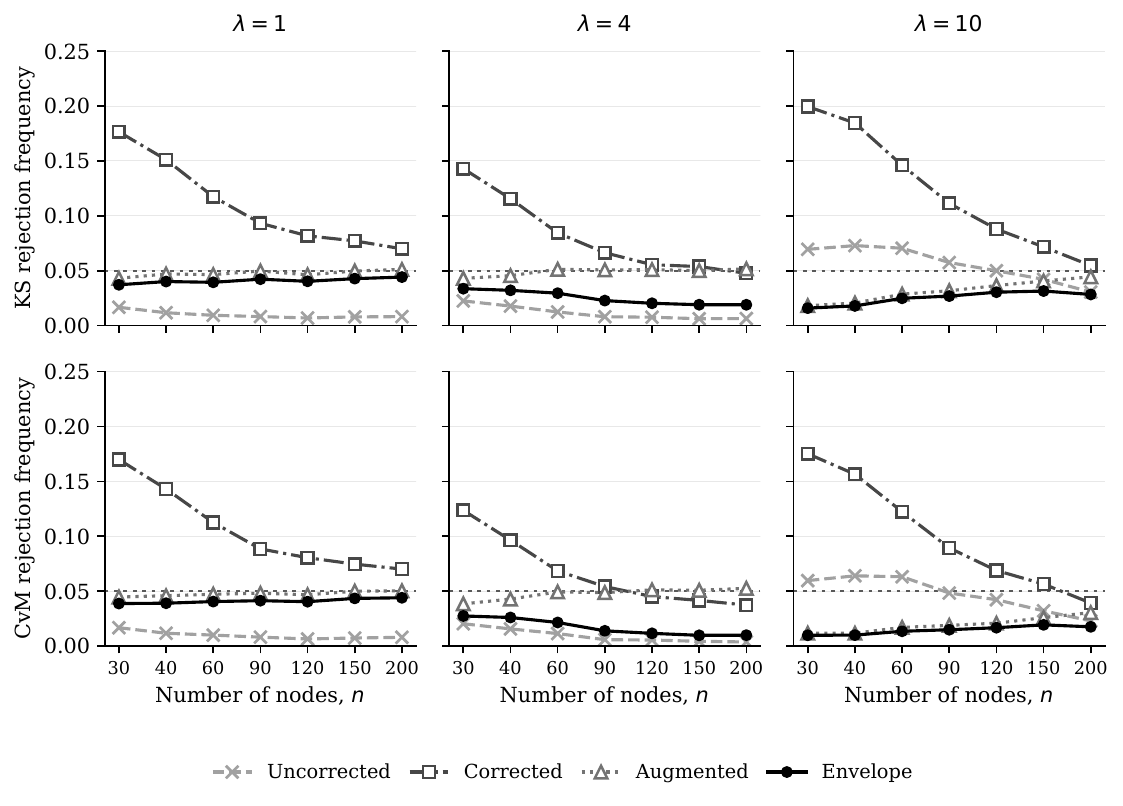}
\caption{Degenerate non-Gaussian case. Columns vary $\lambda$; rows show KS and CvM. Each point uses 20,000 replications and 999 draws.
The horizontal dotted line marks the nominal 5\% level.}
\label{fig:size-nong}
\end{figure}

\subsection{Local power}

For the three Gaussian designs, we add $a\rho_n X_{ij}^2$ to $Y_{ij}$
and estimate the same regression model. Here $\rho_n=N_n^{-1/2}$ for
independent dyads and $\rho_n=n^{-1/2}$ for the two nondegenerate
designs. The departure $X^2$ does not lie in the span of the regression
gradients, and the population CvM criterion for the local change is
positive. We use $n=120$, $a\in\{0,0.25,0.5,1,2,4\}$, 20,000
replications, and 999 bootstrap draws. Figure~\ref{fig:local-power}
reports $\lambda=1$; the values at $a=0$ give each procedure's size.

Rejection frequencies rise with $a$. At $a=1$, the envelope KS
frequencies are 0.181 under independent dyads, 0.990 in the
nondegenerate weak Gaussian design, and 0.113 in the nondegenerate
strong Gaussian design. At $a=4$, they reach 0.989, 1.000, and
0.604, respectively. The envelope detects departures in all three
designs, but the same change in the conditional mean is harder to
detect when node variation is stronger.

The null frequencies are important for interpreting these comparisons.
Under independent dyads, the uncorrected procedure rejects relatively
rarely even under the alternatives. In the nondegenerate designs,
the augmented procedure alone overrejects under the null, so its higher
rejection frequency need not indicate greater power.
Figures~\ref{fig:power-lambda4} and~\ref{fig:power-lambda10} of the
supplement report the corresponding comparisons for $\lambda=4$ and $10$.

\begin{figure}[!t]\centering
\includegraphics[width=\textwidth]{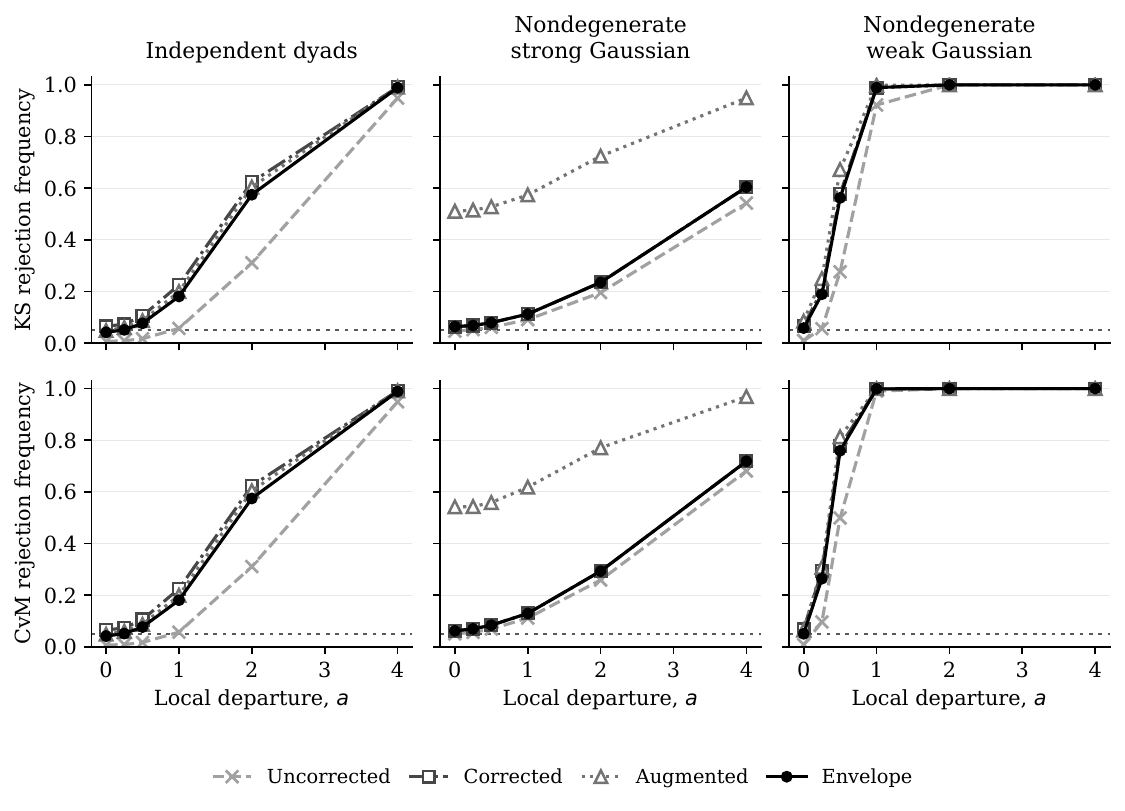}
\caption{Local power at $n=120$ and $\lambda=1$. The added conditional-mean
term is $a\rho_n X^2$. Each point uses 20,000 replications and 999 draws;
$a=0$ is the null.}\label{fig:local-power}
\end{figure}

The size and power results support using the envelope when the variance
regime is unknown. It follows the corrected procedure when node
variation dominates and protects against the size distortions of the
first-order procedures in the non-Gaussian design. Some finite-sample
distortions remain, particularly with strong node errors, but the
envelope retains power against the alternatives considered.

\section{Empirical application: Which housing sales are comparable?}\label{sec:empirical}

When choosing a comparable sale, an appraiser can observe both the distance
between two properties and their access to public transit. Does distance
alone capture how closely their prices agree? We examine this question in
Xindian District, New Taipei City, using the housing transactions of
\citet{yeh2018xindian}. \citet{yehhsu2018comparables} use these data to
predict individual property values from comparable sales. Our outcome
instead measures how far apart the prices of two locations are. This lets
us ask whether proximity to a metro station adds information even among
pairs the same distance apart.

The data contain 414 transactions in 2012-2013 at 259 distinct recorded
locations. We take each location as a node and let $P_i$ be its median
transaction price per unit of floor area. This gives equal weight to
locations, including those with several sales. For each of the 33,411
unordered pairs, the outcome is $Y_{ij}=|\log P_i-\log P_j|$.
Let $D_{ij}$ denote the geographic distance between locations $i$ and
$j$, calculated from their latitude and longitude using the great-circle
formula and rounded to the nearest 0.5 kilometre. Let $C_{ij}=1$ if exactly one location is within 800 metres
of its nearest metro station. Thus $C_{ij}=0$ includes pairs that are both
near or both far from a station. The 800-metre definition follows the
Taipei housing study of \citet{lee2017taipeimrt}.

We test four specifications of the dyadic regression
\begin{equation}
 Y_{ij}=m(D_{ij},C_{ij};\bm\theta)+\varepsilon_{ij},
 \qquad E[\varepsilon_{ij}\mid D_{ij},C_{ij}]=0.
 \label{eq:housing-regression}
\end{equation}
For distance band $d$ and metro-access indicator $c\in\{0,1\}$, the four
specifications are
\begin{equation}
 \begin{aligned}
 \text{Common straight line:}\qquad
   &m(d,c;\bm\theta)=a+b d,\\
 \text{Separate mean by distance band:}\qquad
   &m(d,c;\bm\theta)=\eta_d,\\
 \text{Access-specific straight lines:}\qquad
   &m(d,c;\bm\theta)=a_c+b_c d,\\
 \text{Access-specific power functions:}\qquad
   &m(d,c;\bm\theta)=\exp(a_c)(1+d)^{b_c}.
 \end{aligned}
 \label{eq:housing-specifications}
\end{equation}
The second specification gives each of the 20 distance bands its own mean
$\eta_d$. Its only restriction is that, within a band, the expected price
gap is the same for $c=0$ and $c=1$. The last two specifications allow
the two groups to have different price gaps at the same distance.

Figure~\ref{fig:housing-gaps} suggests that the mean absolute log-price
gap increases more rapidly with distance when only one location is near
a metro station. From the one-kilometre band to the six-kilometre band,
the mean gap rises from 0.44 to 0.85 for these pairs ($C_{ij}=1$),
compared with 0.24 to 0.53 when both locations are near a station or
both are far from one ($C_{ij}=0$). This pattern motivates allowing
both the intercept and the slope to vary with metro access.

\begin{figure}[!t]
\centering
\includegraphics[width=.88\linewidth]{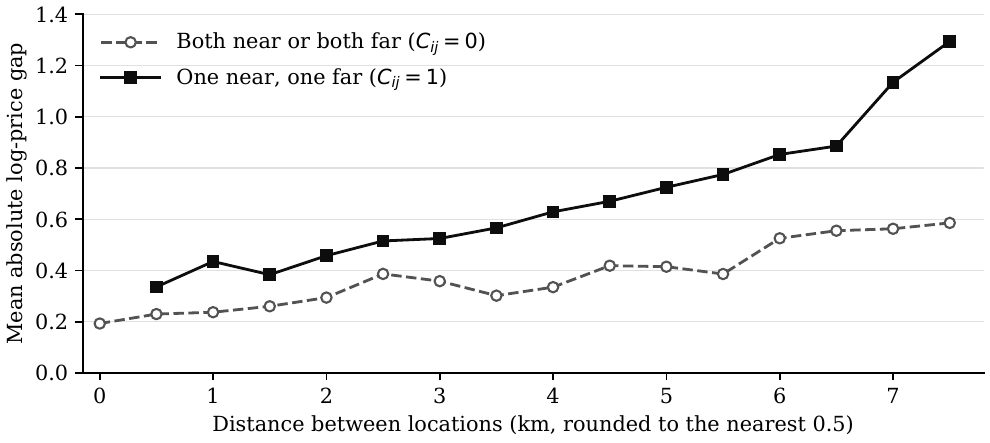}
\caption{Price differences between locations by geographic distance and
metro access. Each point is the mean absolute log-price gap within a
half-kilometre distance band. The figure shows cells with at least 50
pairs; the tests use all pairs.}\label{fig:housing-gaps}
\end{figure}

\begin{table}[!t]
\centering
\caption{Specification tests for absolute log-price gaps in Xindian}
\label{tab:empirical}
\begingroup\small\setstretch{1}\setlength{\tabcolsep}{7pt}
\begin{tabular}{lrr}
\toprule
Conditional mean & Envelope KS & Envelope CvM \\
\midrule
Common straight line          & $<0.001$ & 0.0013 \\
Separate mean by distance band & $<0.001$ & $<0.001$ \\
Access-specific straight lines & 0.1315 & 0.2488 \\
Access-specific power functions   & 0.0377 & 0.0801 \\
\bottomrule
\end{tabular}
\par\smallskip
\begin{minipage}{.92\linewidth}\footnotesize
\textit{Notes:} Entries are bootstrap $p$-values based on 19,999 draws per
component. KS uses all 40 combinations of the observed marginal covariate
values; CvM averages over all 33,411 pairs, retaining their empirical
frequencies. The envelope uses the larger of the corrected first-order
and augmented second-order critical values.
\end{minipage}\endgroup
\end{table}

Both envelope tests reject the two distance-only specifications
(Table~\ref{tab:empirical}). Allowing the distance profile to be
unrestricted across bands therefore does not explain away the difference
between the two groups. Neither test rejects the access-specific straight
lines at 5\%. The estimated slopes are 0.046 log points per kilometre for
pairs with the same metro classification and 0.101 for mixed pairs.
A nonlinear specification is not automatically preferable: KS rejects
the access-specific power functions at 5\%, while CvM does not.
Table~\ref{tab:housing-components} in the supplement reports the component $p$-values. Table~\ref{tab:housing-sensitivity} shows that both
tests still reject the unrestricted distance-only model at 600 and
1,000 metres and when the first recorded sale replaces each location's
median. Conclusions about the precise shape of the two relationships are more
sensitive to the cutoff.

Overall, both envelope KS and CvM tests reject the distance-only models, even with a separate mean
for each distance band, but neither rejects separate straight lines by
metro access. For choosing comparable sales in this market, the results
suggest considering both how far apart two properties are and whether
they have similar access to a metro station.

\section{Conclusion}\label{sec:conclusion}

We develop KS and CvM specification tests for linear and nonlinear
conditional means in undirected dyadic regressions. The corrected
first-order multiplier bootstrap yields asymptotically exact tests in
two standard Gaussian cases: the nondegenerate case and the independent dyads case. In the general degenerate case, we complement existing results on dyadic empirical processes by establishing a feasible
second-order multiplier bootstrap approximation. The bootstrap reproduces
the joint limiting law of the node-interaction and dyad-specific
components uniformly over the indexing function class, even when the
limit is non-Gaussian and the indexing functions are estimated.
For specification testing, we augment the bootstrap to account for
parameter estimation and nonlinear curvature, using the same node
multipliers to preserve dependence among these terms.

We recommend the envelope test when the dependence structure is unknown.
It requires only one regression estimate and uses the larger of the
corrected first-order and augmented second-order critical values. The envelope
test is asymptotically exact in the two Gaussian cases and has
asymptotic size at most the nominal level in the remaining cases.
In the housing application, both tests reject models based on distance
alone, but neither rejects separate linear relationships by metro access.

\bibliographystyle{chicago}
{\small\setstretch{1.0}\bibliography{multiway_clustering}}

\clearpage

\renewcommand{\thesection}{S\Alph{section}}
\numberwithin{equation}{section}
\numberwithin{theorem}{section}
\numberwithin{proposition}{section}
\numberwithin{assumption}{section}
\numberwithin{table}{section}
\numberwithin{figure}{section}
\begin{center}{ \large \textbf{Supplementary Materials \\  Specification Testing for Dyadic Regression Models}}

\end{center}

\appendix
Section~\ref{app:foundations} proves the projection and process results used
in Sections~\ref{sec:bootstrap} and~\ref{sec:general-process} of the main paper.
Section~\ref{app:regression} establishes the expansion of the estimated
process, consistency of the covariance estimator, and the approximations
used for KS and CvM. Section~\ref{app:inference} proves size and power
results for the corrected first-order multiplier bootstrap.
Section~\ref{app:second-order} derives the second-order expansion and
proves validity of the augmented second-order multiplier bootstrap and
the envelope test. Section~\ref{app:simulation-checks} verifies the
simulation designs and reports further results.
Section~\ref{app:housing} reports the component tests and sensitivity
checks for the housing application.

Write $\ind_{\bm x}(\bm X)=\ind\{\bm X\preceq\bm x\}$,
$\|h\|_{P,q}=(P|h|^q)^{1/q}$, and
$\|z\|_\infty=\sup_{\bm x\in\calX}|z(\bm x)|$.
Throughout, $E^*$ and $P^*$ condition on the observations and any
auxiliary grid randomization.
Conditional weak convergence is in probability in bounded-Lipschitz
distance, using the supremum norm on each indexed component.

\section{Technical Foundations for Dyadic Processes}\label{app:foundations}

\subsection{Projection and covariance identities}

Write $f_{ij}=f(\bm{Z}_{ij})$ and $\gamma(\cdot,\cdot,f)$ as in the main text, and let
$e_{ij}(f)=f_{ij}-Pf-f_1(U_i)-f_1(U_j)-\gamma(U_i,U_j,f)$.
Iterated expectation gives
\begin{equation}
 Ef_1=0,\qquad E[\gamma(U_i,U_j,f)\mid U_i]=0,\qquad
 E[e_{ij}(f)\mid U_i,U_j]=0.\label{eq:SA-centering}
\end{equation}
The analogous identities hold with $i$ and $j$ interchanged.
Every node occurs in $n-1$ dyads. Summing the decomposition and using
$(n-1)/N_n=2/n$ proves \eqref{eq:main-projection}.
The three sums are pairwise uncorrelated. In particular,
\begin{equation}
 \Cov(\Pn f,\Pn h)=\frac4n E[f_1h_1]
       +\frac1{N_n}E[\gamma(U_i,U_j,f)\gamma(U_i,U_j,h)]+
        \frac1{N_n}E[e_{ij}(f)e_{ij}(h)].\label{eq:SA-exact-cov}
\end{equation}
Equivalently, counting $N_n$ identical-dyad pairs and $2N_n(n-2)$
ordered pairs sharing one node gives
\[
 \Cov(\Pn f,\Pn h)
 =\frac{\Cov(f_{ij},h_{ij})}{N_n}
  +\frac{2(n-2)}{N_n}E[f_1h_1].
\]
Disjoint dyads have zero covariance by dissociation.

\subsection{Entropy, maximal bounds, and continuity}\label{app:entropy}

A measurable class $\calH$ with envelope $H$, meaning $|h|\le H$ for
all $h\in\calH$, is VC type if, for fixed constants $A,v<\infty$,
\[
 \sup_Q N(\epsilon\|H\|_{Q,2},\calH,L_2(Q))
 \le(A/\epsilon)^v,\qquad 0<\epsilon\le1,
\]
where $Q$ ranges over finitely supported probability measures and $N$
is the covering number. Constants in the following bounds may depend
on $A$ and $v$.
A conditional average $Tf(v)=E[f(W)\mid V=v]$ has envelope
$H_T(v)=\{E[H(W)^2\mid V=v]\}^{1/2}$. Let $Q^\uparrow$ denote the distribution of $W$ obtained by drawing
$V$ from $Q$ and then $W$ from its conditional distribution given $V$. Then
\begin{equation}
 \|Tf-Tg\|_{Q,2}\le\|f-g\|_{Q^\uparrow,2},\quad
 \|H_T\|_{Q,2}=\|H\|_{Q^\uparrow,2},\quad
 \|H_T\|_{P,q}\le\|H\|_{P,q}.\label{eq:SA-cover-transfer}
\end{equation}
Consequently the projection classes inherit polynomial covering bounds.
The uniform covering inequality extends from finite measures to these
induced measures by finite approximation on the countable determining
subclass. Sums and differences inherit the same type of bound. Products
use the inequality $|fg-f'g'|\le H(|f-f'|+|g-g'|)$ and the covering
measure weighted by $H^2$; their envelope is a constant times $H^2$.

\begin{lemma}[Maximal bounds]\label{lem:SA-maximal}
Suppose Assumption~\ref{ass:sampling} holds and $PH^2<\infty$.
For a VC-type class of centered functions,
\begin{align}
 E\Big\|\frac1n\sum_i h_1(U_i)\Big\|_{\calH}^2&\le C PH^2/n,\\
 E\Big\|\frac1{N_n}\sum_{i<j}\gamma(U_i,U_j,h)\Big\|_{\calH}^2&\le C PH^2/n^2,\\
 E\Big\|\frac1{N_n}\sum_{i<j}e_{ij}(h)\Big\|_{\calH}^2&\le C PH^2/n^2,\\
 E E^*\Big\|\frac1{\sqrt{N_n}}\sum_{i<j}\xi_i\xi_j h_{ij}\Big\|_{\calH}^2
 &\le C PH^2.\label{eq:SA-maximal}
\end{align}
Here $\|z\|_{\calH}=\sup_{h\in\calH}|z(h)|$. The fourth bound does
not require $Ph=0$ or $h_1=0$.
\end{lemma}

\begin{proof}
For the node sum and, conditionally on the nodes, the dyad-specific
sum, symmetrization gives Rademacher increments satisfying
\[
 \Big\|\sum_a\zeta_a c_a\Big\|_{L_s}
 \le C\sqrt{s}\Big(\sum_a c_a^2\Big)^{1/2},\qquad s\ge2.
\]
The covering bound and $\int_0^1\sqrt{1+\log(A/u)}\,du<\infty$ give
\[
 E\Big\|\sum_i h_1(U_i)\Big\|_{\calH}^2\le CnPH^2,
 \qquad
 E\Big\|\sum_{i<j}e_{ij}(h)\Big\|_{\calH}^2\le CN_nPH^2.
\]
For the interaction sum, degeneracy permits decoupling
\citep{de2012decoupling}:
\[
 E\Big\|\sum_{i<j}\gamma(U_i,U_j,h)\Big\|_{\calH}^2
 \le C E\Big\|\sum_{i\ne j}\zeta_i\zeta'_j
                    \gamma(U_i,U'_j,h)\Big\|_{\calH}^2,
\]
where $(U'_j)$ and the two sign sequences are independent copies.
Repeated Khintchine inequalities and the Gaussian quadratic-form bound yield
\begin{align*}
 \Big\|\sum_{i,j}\zeta_i\zeta'_j c_{ij}\Big\|_{L_s}
 &\le Cs\Big(\sum_{i,j}c_{ij}^2\Big)^{1/2},\\
 \Big\|\sum_{i<j}\xi_i\xi_j c_{ij}\Big\|_{L_s(P^*)}
 &\le Cs\Big(\sum_{i<j}c_{ij}^2\Big)^{1/2}.
\end{align*}
Chaining now uses $\int_0^1\{1+\log(A/u)\}\,du<\infty$.
The expected squared envelope sums are at most $Cn^2PH^2$ and
$CN_nPH^2$, respectively. Division by the stated normalizations
proves the remaining bounds.
\end{proof}

The next bound follows from these inequalities; related projection
bounds appear in \citet[Supplement, Lemmas S3 and S10]{davezies2021}
and \citet[Appendix D]{chiang2023inference}.

\begin{lemma}[Bound for the second-order projection]
\label{lem:uniform-degenerate}
Under Assumption~\ref{ass:sampling}, let $\calF$ be VC type with
$PF^2<\infty$. Then, for a constant $C$ independent of $n$,
\begin{equation}\label{eq:uniform-degenerate}
 E\sup_{f\in\calF}|\mathbb U_n f_2|
 \le \frac{C\|F\|_{P,2}}{n},
 \qquad E\left[\sup_{f\in\calF}|\mathbb U_n f_2|^2\right]
 \le \frac{C\|F\|_{P,2}^2}{n^2}.
\end{equation}
\end{lemma}

\begin{proof}[Proof of Lemma~\ref{lem:uniform-degenerate}]
The second-order projection satisfies $f_2(U_i,U_j,U_{ij})=\gamma(U_i,U_j,f)+e_{ij}(f)$. Apply the second and third
bounds in Lemma~\ref{lem:SA-maximal}, then the triangle and
Cauchy-Schwarz inequalities. Conditional-expectation contraction
\eqref{eq:SA-cover-transfer} bounds their envelopes by $C\|F\|_{P,2}$.
\end{proof}

For local bounds, let $d(h,h')=\|h-h'\|_{P,2}$ and assume $PH^q<\infty$,
$q>4$. The squared-difference class has envelope $4H^2$, so
\begin{equation}
 E\sup_{d(h,h')\le\delta}\Pn(h-h')^2
 \le\delta^2+C(PH^4)^{1/2}/\sqrt n.
 \label{eq:SA-radius}
\end{equation}
The same bound holds for the decoupled node-pair measure.
After truncating the envelope at $M$, the local chaining bounds are
\[
 \int_0^a\sqrt{1+\log(CM/u)}\,du,
 \qquad
 \int_0^a\{1+\log(CM/u)\}\,du,
\]
where $a$ is the empirical local radius in \eqref{eq:SA-radius}.
The discarded tail contributes at most
$C\|H\ind\{H>M\}\|_{P,2}$.
Taking $n\to\infty$, $\delta\downarrow0$, and then $M\to\infty$
gives vanishing expected moduli for the $\sqrt n$-normalized node
sum and the $\sqrt{N_n}$-normalized second-order sums.
For the second-order multiplier bootstrap, the conditional expected
modulus vanishes in probability. These bounds and finite $L_2(P)$
covers give asymptotic equicontinuity and tightness.

For the marks $r_{\bm x}$, define
\begin{equation}
 d_F(\bm{x},\bm{y})=\sum_{k=1}^d|F_k(x_k)-F_k(y_k)|.\label{eq:SA-rank-metric}
\end{equation}
An indicator difference is supported on a set of probability at most
$d_F(\bm{x},\bm{y})$. H\"older's inequality and bounded regression derivatives
therefore give
\begin{equation}
 \|r_{\bm{x}}-r_{\bm{y}}\|_{P,2}\le C d_F(\bm{x},\bm{y})^{(q-2)/(2q)}.\label{eq:SA-increment}
\end{equation}
Conditional projections satisfy the analogous inequality by contraction.
This bound permits discrete covariates and zero process variance at
some values of $\bm x$. For the difference between $r_{\bm x}$ and
the zero mark at $0_*$, replace $d_F$ by $\min_k F_k(x_k)$.
When necessary, identify indices at zero $d_F$ distance and complete
the resulting metric space. The limits have continuous extensions to
this space by the same increment bound.

\subsection{Functional sampling limits}

The next lemma specializes the functional sampling limits of
\citet[Theorems 2.1 and 3.1]{davezies2021} to symmetric dyads and our
normalizations. Its proof records the separate interaction and
dyad-specific components needed for Theorem~\ref{thm:quadratic}.

\begin{lemma}[Dyadic empirical-process limits]\label{lem:ddg}
Under Assumption~\ref{ass:sampling}, let $\calF$ be VC type with
$PF^q<\infty$ for some $q>4$.
\begin{enumerate}[label=(\roman*),leftmargin=*]
\item In $\ell^\infty(\calF)$, $\{\sqrt n(\Pn f-Pf):f\in\calF\}\weak\mathbb G$, where
$\mathbb G$ is centered Gaussian with covariance
\[
 K_1(f,h)=4E[f_1(U_i)h_1(U_i)].
\]
\item If $f_1=0$ for every $f\in\calF$, then, in
$\ell^\infty(\calF)$,
\begin{equation}
 \{\sqrt{N_n}(\Pn f-Pf):f\in\calF\}\weak\mathbb L,
 \qquad
 \mathbb L(f)=\mathbb W(\gamma(\cdot,\cdot,f))+\mathbb G_e(f).
 \label{eq:degenerate-limit}
\end{equation}
Here $\mathbb G_e$ is centered Gaussian with covariance
$E[e_{ij}(f)e_{ij}(h)]$, independent of $\mathbb W$.
For a common orthonormal basis $(\phi_k)$ of the centered
node $L_2$ space,
\begin{equation}
 \begin{aligned}
 \mathbb W(\gamma(\cdot,\cdot,f))
 &=\frac1{\sqrt2}\sum_{k,l\ge1}
 \beta_{kl}(f)(Z_kZ_l-\ind\{k=l\}),\\
 \beta_{kl}(f)
 &=E[\gamma(U_i,U_j,f)\phi_k(U_i)\phi_l(U_j)],
 \end{aligned}
 \label{eq:SA-chaos}
\end{equation}
where the $Z_k$ are independent standard normal variables and
the series is interpreted in $L_2$.
The covariance of $\mathbb W(\gamma(\cdot,\cdot,f))$ and $\mathbb W(\gamma(\cdot,\cdot,h))$ is
$E[\gamma(U_i,U_j,f)\gamma(U_i,U_j,h)]$. The same sequence $(Z_k)$ is used for every $f$.
\end{enumerate}
\end{lemma}

\begin{proof}[Proof of Lemma~\ref{lem:ddg}]
\emph{(i) The node limit.}
The projected class is Donsker by \eqref{eq:SA-cover-transfer}.
Lemma~\ref{lem:uniform-degenerate} gives, uniformly in $f\in\calF$,
\[
 \sqrt n(\Pn f-Pf)=\frac2{\sqrt n}\sum_i f_1(U_i)+o_p(1).
\]
The i.i.d. functional CLT proves (i).

\emph{(ii) The interaction limit.}
Choose a bounded orthonormal basis $(\phi_k)$ of the centered node
$L_2$ space. Let $\gamma^{(m)}$ truncate \eqref{eq:interaction-expansion}
to $k,l\le m$ and put $Z_{n,k}=n^{-1/2}\sum_i\phi_k(U_i)$.
By symmetry,
\[
 \frac1{\sqrt{N_n}}\sum_{i<j}\gamma^{(m)}(U_i,U_j,f)
 =\frac n{2\sqrt{N_n}}\sum_{k,l\le m}\beta_{kl}(f)
 \left\{Z_{n,k}Z_{n,l}-\frac1n\sum_i\phi_k(U_i)\phi_l(U_i)\right\}.
\]
The multivariate CLT and LLN give the truncated limit in
\eqref{eq:SA-chaos}. Degeneracy gives the exact tail bound
\[
 E\left|\frac1{\sqrt{N_n}}\sum_{i<j}
          (\gamma-\gamma^{(m)})(U_i,U_j,f)\right|^2
 =\|\gamma(\cdot,\cdot,f)-\gamma^{(m)}(\cdot,\cdot,f)\|_{P,2}^2\to0.
\]
The limiting series has the same tail variance. Finite-dimensional
convergence follows by taking $n\to\infty$ and then $m\to\infty$.

\emph{(iii) The dyad-specific limit.}
For any fixed linear combination $e_t$ of the dyad-specific components, conditional
independence and the node-pair LLN give
\begin{align*}
 \frac1{N_n}\sum_{i<j}E[e_{t,ij}^2\mid(U_i)]&\to_p Ee_{t,ij}^2,\\
 E\left[\frac1{N_n}\sum_{i<j}E\{e_{t,ij}^2
 \ind(|e_{t,ij}|>\epsilon\sqrt{N_n})\mid(U_i)\}\right]
 &=E[e_{t,ij}^2\ind(|e_{t,ij}|>\epsilon\sqrt{N_n})]\to0.
\end{align*}
Hence the conditional CLT yields
\[
 E\left[\exp\left\{iN_n^{-1/2}\sum_{i<j}e_{t,ij}\right\}
                \middle|\,(U_i)\right]
 \to\exp(-Ee_{t,ij}^2/2)\quad\text{in }L_1.
\]
Multiplying by a bounded function of the nodes proves independence
of $\mathbb G_e$ and the interaction limit. The local bounds following
\eqref{eq:SA-radius} give tightness of both components, completing (ii).
\end{proof}

\subsection{Conditional second-order multiplier bootstrap limit}

The following lemma extends Theorem~\ref{thm:quadratic} to joint
convergence with node averages. It is used for the parameter-estimation
corrections in Section~\ref{sec:nongaussian}.
Let $\mathbb B$ be a centered Gaussian linear map on the centered node
$L_2$ space, with $E[\mathbb B(a)\mathbb B(b)]=E[ab]$ and
$\mathbb B(\phi_k)=Z_k$ in \eqref{eq:chaos}, independent of
$\mathbb G_e$. For a dyad function $h$, define
$\widehat h_i=(n-1)^{-1}\sum_{j\ne i}h(\bm Z_{ij})-\Pn h$ and
$\mathbb M_n^*(h)=2n^{-1/2}\sum_i\xi_i\widehat h_i$,
using the same multipliers as $\mathbb Q_n^*$.

\begin{lemma}[Joint approximation of the second-order process and node sums]\label{lem:joint-quadratic}
Under Assumption~\ref{ass:sampling}, let $\calF,\calH$ be fixed VC-type
classes with envelopes having finite $q$th moments for some $q>4$.
Suppose $f_1=0$ for every $f\in\calF$. Then
\begin{align*}
 \bigl(\{\sqrt{N_n}(\Pn f-Pf)\}_{f\in\calF},
       \{\sqrt n(\Pn h-Ph)\}_{h\in\calH}\bigr)
 &\weak
 \bigl(\{\mathbb L(f)\}_{f\in\calF},
       \{2\mathbb B(h_1)\}_{h\in\calH}\bigr),\\
 \bigl(\{\mathbb Q_n^*(f)\}_{f\in\calF},
       \{\mathbb M_n^*(h)\}_{h\in\calH}\bigr)
 &\weak
 \bigl(\{\mathbb L(f)\}_{f\in\calF},
       \{2\mathbb B(h_1)\}_{h\in\calH}\bigr)
\end{align*}
as joint indexed processes, with the supremum norm on each class.
The second convergence is conditional on the observations, in
probability.
\end{lemma}

\begin{proof}[Proof of Lemma~\ref{lem:joint-quadratic}]
\emph{(i) Sampling and centering.}
Include finitely many $h_1$ in the node CLT of the preceding proof.
Its conditional characteristic-function argument makes $\mathbb G_e$
independent of the entire node Gaussian vector. The local bounds
give joint tightness and prove the sampling assertion.

In the bootstrap argument, condition first on all node and dyad variables;
only the independent multipliers are random. Center at $Pf$, writing
$\mathbb Q_{n,0}^*(f)=N_n^{-1/2}\sum_{i<j}\xi_i\xi_j(f_{ij}-Pf)$.
Since $E^*|N_n^{-1/2}\sum_{i<j}\xi_i\xi_j|\le1$,
\[
 E^*\|\mathbb Q_n^*-\mathbb Q_{n,0}^*\|_{\calF}
 \le\sup_f|\Pn f-Pf|=o_p(1).
\]
Lemma~\ref{lem:SA-maximal} also gives
$E^*\|\mathbb Q_n^*\|_{\calF}=O_p(1)$ without degeneracy.

\emph{(ii) The interaction component.}
Put $Z_{n,k}^*=n^{-1/2}\sum_i\xi_i\phi_k(U_i)$. Conditionally,
\begin{align*}
 (Z_{n,k}^*)_{k\le m}&\sim N\left(\bm0,
 \left(\frac1n\sum_i\phi_k(U_i)\phi_l(U_i)\right)_{k,l\le m}\right),\\
 \frac1n\sum_i\xi_i^2\phi_k(U_i)\phi_l(U_i)
 &\to_{P^*}\ind\{k=l\}\quad\text{in probability}.
\end{align*}
The covariance matrix tends to $\bm I_m$; the conditional variance
in the second line is $O_p(n^{-1})$. Substitution in the preceding
quadratic identity gives the finite-rank bootstrap limit. Also,
\[
 EE^*\left|\frac1{\sqrt{N_n}}\sum_{i<j}\xi_i\xi_j
             (\gamma-\gamma^{(m)})(U_i,U_j,f)\right|^2
 =\|\gamma(\cdot,\cdot,f)-\gamma^{(m)}(\cdot,\cdot,f)\|_{P,2}^2\to0.
\]

\emph{(iii) The dyad-specific component.}
For a fixed linear combination $e_{ij}$, let $\bm H_n$ be symmetric,
with zero diagonal and $H_{n,ij}=e_{ij}/(2\sqrt{N_n})$ for $i\ne j$.
Its multiplier sum is $\bm\xi^\top\bm H_n\bm\xi$, and
\begin{equation}
 2\operatorname{tr}(\bm H_n^2)=\Pn e^2\to_p Ee_{ij}^2,
 \qquad \|\bm H_n\|_{\rm op}\to_p0.
 \label{eq:SA-edge-spectrum}
\end{equation}
Conditional independence and centering imply that a term in
$E\operatorname{tr}(\bm H_n^4)$ vanishes if any dyad occurs only once. The remaining $O(n^3)$ terms involve at most three nodes, giving
\[
 E\|\bm H_n\|_{\rm op}^4\le E\operatorname{tr}(\bm H_n^4)
 \le Cn^3Ee_{ij}^4/N_n^2=O(n^{-1}).
\]
Diagonalization gives $\sum_j\lambda_{nj}(\zeta_j^2-1)$, with
$\max_j|\lambda_{nj}|\to_p0$ and
$2\sum_j\lambda_{nj}^2\to_p Ee_{ij}^2$. The conditional CLT follows.

\emph{(iv) Joint convergence.}
Include $Z_{n,k}^*$ and $2n^{-1/2}\sum_i\xi_i h_1(U_i)$ in one
finite Gaussian vector. Write a linear combination as
$\bm v_n^\top\bm\xi$, where $\|\bm v_n\|=O_p(1)$. Then
\begin{align*}
 E^*e^{it\bm\xi^\top\bm H_n\bm\xi+i\bm v_n^\top\bm\xi}
 &=\det(\bm I-2it\bm H_n)^{-1/2}
   e^{-\bm v_n^\top(\bm I-2it\bm H_n)^{-1}\bm v_n/2}\\
 &=\exp\{-t^2\operatorname{tr}(\bm H_n^2)-\|\bm v_n\|^2/2\}+o_p(1),
\end{align*}
by \eqref{eq:SA-edge-spectrum} and $\operatorname{tr}\bm H_n=0$.
This factorization proves independence. Taking linear combinations
and removing the finite-rank truncation gives the finite-dimensional
joint limits.

Lemma~\ref{lem:SD-rows} permits replacing $h_1(U_i)$ by $\widehat h_i$.
Its Gaussian modulus and the modulus for the second-order process following
\eqref{eq:SA-radius} give conditional tightness. The conditional law
of the feasible pair depends only on the observed array, so convergence
on the latent representation also holds under the observed-data law.
\end{proof}

\begin{proof}[Proof of Theorem~\ref{thm:quadratic}]
Apply Lemma~\ref{lem:joint-quadratic} with $\calH=\{0\}$ and take the component corresponding to $\mathbb Q_n^*$.
\end{proof}

\subsection{Replacing indexing functions by estimates}

\begin{proof}[Proof of Corollary~\ref{cor:quadratic-plugin}]
Write
\[
 \mathbb V_n^*(h)=N_n^{-1/2}\sum_{i<j}\xi_i\xi_jh(\bm Z_{ij}),
 \qquad \mathbb Q_n^*(h)=\mathbb V_n^*(h)-\Pn h\,\mathbb V_n^*(1).
\]
For every $\eta>0$, the local quadratic bound gives
\begin{equation}
 \lim_{\delta\downarrow0}\limsup_n
 P\left\{E^*\sup_{\substack{h,h'\in\calH\\\|h-h'\|_{P,2}\le\delta}}
 |\mathbb V_n^*(h)-\mathbb V_n^*(h')|>\eta\right\}=0.
 \label{eq:SA-quadratic-modulus}
\end{equation}
The centering term has the same property since
\[
 E^*|\mathbb V_n^*(1)|\le1,\qquad
 \sup_{\|h-h'\|_{P,2}\le\delta}|\Pn(h-h')|
 \le\left\{\sup_{\|h-h'\|_{P,2}\le\delta}\Pn(h-h')^2\right\}^{1/2},
\]
and \eqref{eq:SA-radius} controls the right-hand side.
For fixed $\delta>0$, with probability tending to one all estimated
functions lie in $\calH$ and
$\sup_f\|\widehat f_n(f)-f\|_{P,2}\le\delta$.
On this event, the substitution error is bounded by the local modulus
for $\mathbb Q_n^*$. Taking $n\to\infty$ and then $\delta\downarrow0$
proves \eqref{eq:quadratic-plugin-general}. Conditional Markov's
inequality transfers the limit in Theorem~\ref{thm:quadratic}.
\end{proof}

\section{Regression Estimation, Covariance, and Grids}\label{app:regression}

\subsection{Identification and nonlinear expansion}

Write $\bm s_{ij}=\bm g(\bm X_{ij})\varepsilon_{ij}$ for the population
least-squares score.

\begin{proof}[Proof of Lemma~\ref{lem:identification}]
The finite signed measure $\nu(A)=E[\varepsilon\ind\{\bm{X}\in A\}]$
vanishes on all lower orthants if $\Delta=0$. Lower orthants,
together with the whole space obtained by monotone limits, determine
finite signed measures. The monotone-class theorem gives $\nu(A)=0$
for every Borel set, equivalent to $E[\varepsilon\mid \bm{X}]=0$ a.s.
The reverse implication follows by iterated expectation.
\end{proof}

Bounded parameter derivatives give finite-dimensional Lipschitz covers;
lower-orthant indicators form a VC class. Thus the regression, score,
Hessian, and mark classes have polynomial covering bounds with
envelope $C(1+|Y|)$. Their product classes have square-integrable envelopes
because $q>4$.

\begin{proof}[Proof of Lemma~\ref{lem:linearization}]
Write $\bm\delta=\widehat{\bm\theta}-\bm\theta_*$ and
$\bm H=\partial^2m(\bm X,\bm\theta_*)/\partial\bm\theta\partial\bm\theta^\top$.
By Assumption~\ref{ass:moments}, the estimator is consistent and the
interior score equations hold with probability tending to one. We work on this event.
Under $H_0$, $E[\varepsilon\bm H]=\bm0$. The projection bounds in
case (i), and i.i.d. maximal bounds in case (ii), give
\begin{equation}
 \bm0=\Pn\bm s-\{\bm Q+O_p(\rho_n)\}\bm\delta
                    +O_p(\|\bm\delta\|^2),
 \qquad \Pn\bm s=O_p(\rho_n).
 \label{eq:SB-score-expansion}
\end{equation}
Since $\bm Q$ is nonsingular, first
$\bm\delta=O_p(\rho_n)$ and then
$\bm\delta=\bm Q^{-1}\Pn\bm s+O_p(\rho_n^2)$.
Uniformly over $\bm x\in\calX$,
\begin{align*}
 \Pn\widehat r(\bm x)
 &=\Pn(\varepsilon\ind_{\bm x})
   -\{\bm M(\bm x)+O_p(\rho_n)\}^\top\bm\delta
   +O_p(\|\bm\delta\|^2)\\
 &=\Pn(\varepsilon\ind_{\bm x})
   -\bm M(\bm x)^\top\bm Q^{-1}\Pn\bm s+O_p(\rho_n^2)\\
 &=\Pn r_{\bm x}+o_p(\rho_n).
\end{align*}
\end{proof}

\begin{proof}[Proof of Theorem~\ref{thm:weak}]
The class of marks $\calR=\{r_{\bm x}:\bm x\in\calX\}$ has the
covering and moment properties established above. In case (i),
Lemma~\ref{lem:ddg}(i) and the uniform expansion for the estimated process give the
Gaussian limit with covariance $\Omega$. In case (ii), apply the
ordinary i.i.d. functional CLT to the $N_n$ complete dyadic
observations and the same uniform expansion. Its covariance is
$E[r_{\bm x}r_{\bm y}]=\Gamma(\bm x,\bm y)$ under $H_0$.
\end{proof}

\subsection{Replacing population marks by estimated marks}

Under $H_0$, let $\bm{A}(\bm{x})=\bm Q^{-1}\bm M(\bm{x})$ and
$\widehat{\bm{A}}(\bm{x})=\widehat{\bm Q}^{-1}\widehat{\bm M}(\bm{x})$.
Uniform laws and the parameter rate imply
\begin{equation}
 d_n:=\sup_{\bm{x}}\|\widehat{\bm{A}}(\bm{x})-\bm{A}(\bm{x})\|=O_p(n^{-1/2}).\label{eq:SB-A-rate}
\end{equation}
To apply Corollary~\ref{cor:quadratic-plugin}, consider the fixed class
\[
 \calH_R=\left\{
 [Y-m(\bm X,\bm\theta)]
 [\ind\{\bm X\preceq\bm x\}-\bm c^\top\bm g_{\bm\theta}(\bm X)]:
 \bm x\in\calX,\ \bm\theta\in\Theta,\ \|\bm c\|\le C_0
 \right\},
\]
where $\bm g_{\bm\theta}=\partial m(\cdot,\bm\theta)/\partial\bm\theta$
and $C_0>\sup_{\bm x}\|\bm A(\bm x)\|$ is fixed.
The finite-dimensional Lipschitz and lower-orthant covering arguments
above show that this class is VC type with envelope $C(1+|Y|)$.
It contains all population marks and, with probability tending to one,
all estimated marks, by consistency and \eqref{eq:SB-A-rate}.
Moreover,
\[
 \sup_{\bm x}\|\widehat r_{\bm x}-r_{\bm x}\|_{P,2}
 \le C(d_n+\|\bm\delta\|)=o_p(1).
\]
Here $\widehat r_{\bm x}(\bm Z_{ij})=\widehat r_{ij}(\bm x)$ and
$\bm\delta=\widehat{\bm\theta}-\bm\theta_*$. The norm $\|\cdot\|_{P,2}$
integrates over an independent observation from $P$, with the estimated
quantities held fixed.
This verifies the corollary's conditions. Denote the first sum in
\eqref{eq:curvature-bootstrap} by
\begin{equation}
 R_n^{Q,*}(\bm x)
 =\frac1{N_n}\sum_{i<j}\xi_i\xi_j\widetilde r_{ij}(\bm x).
 \label{eq:quadratic-bootstrap}
\end{equation}
Under Assumption~\ref{ass:zero}, Theorem~\ref{thm:quadratic} therefore gives
\begin{equation}
 \sqrt{N_n}R_n^{Q,*}\weak\mathbb L(r)
 \quad\text{conditionally in probability}.
 \label{eq:fitted-quadratic-limit}
\end{equation}
Here $\mathbb L(r)$ denotes $\{\mathbb L(r_{\bm x}):\bm x\in\calX\}$.
This is the limit of the process based on population marks; it need not
equal the limit of the estimated process. For the rate used below, write
\begin{equation}
 \widehat r_{ij}(\bm{x})-r_{\bm{x}}(\bm{Z}_{ij})
 =-(\widehat{\bm{A}}-\bm{A})^{\top}\bm{s}_{ij}
   -\bm\delta^{\top}\{\bm{g}_{ij}q_{\bm{x}}(\bm{X}_{ij})+\varepsilon_{ij}\bm{H}_{ij}\bm{A}(\bm{x})\}
   +D_{ij}(\bm{x}),\label{eq:SB-mark-expansion}
\end{equation}
where
$\sup_{\bm{x}}|D_{ij}(\bm{x})|\le C(\|\bm\delta\|^2+\|\bm\delta\|d_n)(1+|Y_{ij}|)$
on events of probability tending to one.
To bound the effect of estimating the marks in the second-order draw, integrate the derivatives of
$[Y-m(\bm X,\bm\theta)][\ind_{\bm x}-\bm a^\top\bm g_{\bm\theta}]$
along the path from $(\bm\theta_*,\bm A(\bm x))$ to
$(\widehat{\bm\theta},\widehat{\bm A}(\bm x))$. These derivatives are
\[
 -[Y-m]\bm g_{\bm\theta},\qquad
 -\bm g_{\bm\theta}[\ind_{\bm x}-\bm a^\top\bm g_{\bm\theta}]
 -[Y-m]\bm H_{\bm\theta}\bm a.
\]
Both derivative classes are VC type with envelope $C(1+|Y|)$.
Lemma~\ref{lem:SA-maximal} therefore gives
\begin{equation}
 E^*\sup_{\bm x}\left|\frac1{\sqrt{N_n}}\sum_{i<j}\xi_i\xi_j
  \{\widehat r_{ij}(\bm x)-r_{\bm x}(\bm Z_{ij})\}\right|
 =O_p(d_n+\|\bm\delta\|)=o_p(1).
 \label{eq:SB-quadratic-plugin}
\end{equation}
Empirical centering has the same order, since
\[
 \sup_{\bm x}|\Pn(\widehat r_{\bm x}-r_{\bm x})|
 E^*\left|\frac1{\sqrt{N_n}}\sum_{i<j}\xi_i\xi_j\right|
 =O_p(d_n+\|\bm\delta\|).
\]
No restriction on the score's node projection is needed for this
replacement. In particular,
$\|R_n^{Q,*}\|_\infty=O_{P^*}(N_n^{-1/2})$ in probability
under any fixed null satisfying the regression assumptions.

\subsection{Covariance correction in the two Gaussian cases}

For $f=r_{\bm x}$, write
$\gamma_{\bm x}(u,v)=\gamma(u,v,r_{\bm x})$ and
$e_{ij}(\bm x)=e_{ij}(r_{\bm x})$.

We first use population marks, writing
$\bar r_i=(n-1)^{-1}\sum_{j\ne i}r_{ij}$ and
$\psi_i=\bar r_i-\Pn r$. Their decomposition is
\begin{equation}
 \psi_i(\bm{x})=\frac{n-2}{n-1}\{\psi_{\bm{x}}(U_i)-\bar\psi_{\bm{x}}\}
       +\bar\gamma_i(\bm{x})-\Pn \gamma_{\bm{x}}+\bar e_i(\bm{x})-\Pn e_{\bm{x}},\label{eq:SB-profile}
\end{equation}
where bars with subscript $i$ denote averages over dyads containing
node $i$, and
$\bar\psi_{\bm{x}}=n^{-1}\sum_i \psi_{\bm{x}}(U_i)$.
Let $C_n^0$ denote \eqref{eq:covariance-estimator} formed from these
population marks and their sample means.

\begin{lemma}[Covariance approximation]\label{lem:SB-covariance}
In case (i),
\[
 \sup_{\bm{x},\bm{y}}|n C_n^0(\bm{x},\bm{y})-4E[\psi_{\bm{x}}\psi_{\bm{y}}]|=O_p(n^{-1/2}).
\]
In case (ii),
\[
 \sup_{\bm{x},\bm{y}}|N_n C_n^0(\bm{x},\bm{y})-E[r_{\bm{x}}r_{\bm{y}}]|=O_p(n^{-1/2}).
\]
Both statements hold with the estimated marks, that is, with
$C_n^0$ replaced by $\widehat C_n$.
\end{lemma}

\begin{proof}
\textit{(i) The nondegenerate case.}
Conditional symmetrization gives
\[
 \frac1n\sum_i\sup_{\bm x}
 |\bar\gamma_i(\bm x)-\Pn\gamma_{\bm x}
   +\bar e_i(\bm x)-\Pn e_{\bm x}|^2=O_p(n^{-1}).
\]
Apply the uniform LLN to the leading node products in
\eqref{eq:SB-profile}, and Cauchy-Schwarz to their cross terms:
\[
 \sup_{\bm x,\bm y}\left|
 \frac4n\sum_i\psi_i(\bm x)\psi_i(\bm y)
 -4E[\psi_{\bm x}\psi_{\bm y}]\right|=O_p(n^{-1/2}).
\]
The covariance subtraction, multiplied by $n$, is $O_p(n^{-1})$.

\textit{(ii) Independent dyads (a special case of degeneracy).}
In this case the observations admit a representation using dyad variables alone. Then
$\psi_{\bm x}=\gamma_{\bm x}=0$ and the marks are independent and centered.
Expanding the node averages gives
\begin{align*}
 N_n\frac4{n^2}\sum_i\bar r_i(\bm x)\bar r_i(\bm y)
 &=2\Pn(r_{\bm x}r_{\bm y})
   +\frac{2W_n(\bm x,\bm y)}{n(n-1)},\\
 W_n(\bm x,\bm y)
 &=\sum_i\sum_{j\ne i}\sum_{k\ne i,j}
       r_{ij}(\bm x)r_{ik}(\bm y).
\end{align*}
Index the $N_n$ independent dyads by $a$ and write $a\sim b$ if
distinct dyads share a node. The sum $W_n$ runs over ordered pairs
$a\sim b$, and each product $r_a(\bm x)r_b(\bm y)$ is centered
given either dyad. Decoupling as in Lemma~\ref{lem:SA-maximal}
introduces an independent copy of the dyads in the second factor and
independent signs on each factor. The resulting product class has
polynomial covering numbers. Its squared envelope sum is bounded in
expectation by $Cn^3(PF^2)^2$: there are $n(n-1)(n-2)$ ordered
pairs. The two-sign bound and chaining used in that lemma give
\[
 E\sup_{\bm x,\bm y}|W_n(\bm x,\bm y)|^2\le Cn^3PF^4,
 \qquad
 \sup_{\bm x,\bm y}\frac{|W_n(\bm x,\bm y)|}{n(n-1)}
 =O_p(n^{-1/2}).
\]
Sample centering contributes $O_p(n^{-1})$ at this scale.
Subtracting one copy of the dyad covariance gives
\[
 N_n C_n^0(\bm x,\bm y)
 =\Pn(r_{\bm x}r_{\bm y})+O_p(n^{-1/2})
 =\Gamma(\bm x,\bm y)+O_p(n^{-1/2})
\]
uniformly in $(\bm x,\bm y)$.

\textit{(iii) Estimated marks.}
For an array $f_{ij}(\bm x)$, define
\[
 \|f\|_{R,\bm x}^2=\frac4{n^2}\sum_i(\bar f_i(\bm x)-\Pn f_{\bm x})^2,
 \qquad
 \|f\|_{D,\bm x}^2=\frac1{N_n}\Pn(f_{\bm x}-\Pn f_{\bm x})^2.
\]
An envelope $F$ bounds these by $4\Pn F^2/n$ and $\Pn F^2/N_n$.
For a centered scalar score with projection variances
$\sigma_1^2=Es_1^2$ and $\sigma_2^2=Es_2^2$, counting gives
\begin{align}
 E\|s\|_R^2
 &=\frac{4(n-2)^2}{n^2(n-1)}\sigma_1^2
   +\frac{4(n-2)}{n^2(n-1)}\sigma_2^2,\\
 E\|s\|_D^2
 &=\frac1{N_n}\{2\sigma_1^2+\sigma_2^2
                   -4\sigma_1^2/n-\sigma_2^2/N_n\}.
 \label{eq:SB-score-seminorm}
\end{align}
Both score seminorms are $O_p(\rho_n)$ in the relevant case.
By \eqref{eq:SB-mark-expansion}, uniformly in $\bm x$, either seminorm
of the mark difference is
\[
 O_p\!\left(d_n\rho_n+\frac{\|\bm\delta\|}{\sqrt n}
       +\frac{\|\bm\delta\|^2+\|\bm\delta\|d_n}{\sqrt n}\right)
 =O_p(\rho_n/\sqrt n).
\]
The population-mark seminorms are $O_p(\rho_n)$ by (i)-(ii).
Cauchy-Schwarz thus bounds the covariance change by
$O_p(\rho_n^2/\sqrt n)$, proving both feasible statements.
\end{proof}

\subsection{KS grid approximation and empirical CvM integration}

\label{app:grid-allocation}

We first verify that the KS grid approximates the full index set.
For $\bm x,\bm y\in\R^d$, define
\[
 \widehat d_n(\bm x,\bm y)
 =
 \sum_{k=1}^d
 \bigl|\widehat F_{n,k}(x_k)-\widehat F_{n,k}(y_k)\bigr|.
\]
For each $\bm x$ that coincides with a grid point, first choose the
smallest index $g$ with $\widehat{\bm x}_g=\bm x$.
Otherwise, choose the smallest index $g$ minimizing
$\widehat d_n(\bm x,\widehat{\bm x}_g)$.
Assign $\bm x$ to this index unless
\[
 \min_k\widehat F_{n,k}(x_k)
 \le
 \min_{1\le h\le G_n}
 \widehat d_n(\bm x,\widehat{\bm x}_h).
\]
In that event, assign it to index zero, representing the empty
lower orthant. Denote the assigned index by $g_n(\bm x)$ and set
\[
 \bm\pi_n(\bm x)
 =
 \begin{cases}
  \widehat{\bm x}_{g_n(\bm x)},&g_n(\bm x)>0,\\
  0_*,&g_n(\bm x)=0,
 \end{cases}
\]
Here $0_*$ is an auxiliary index at which all marks and processes
are defined to be zero. This map is used only to establish the KS
grid approximation. Both CvM statistics use the full empirical average.

The indicator class and the projection bounds give
\begin{equation}
 r_n:=\max_k\sup_t|\widehat F_{n,k}(t)-F_k(t)|=O_p(n^{-1/2}).
 \label{eq:SB-rank-lln}
\end{equation}
On the coverage event in Assumption~\ref{ass:grid}, if
$\min_k\widehat F_{n,k}(x_k)>h_n$, the grid
provides quantile levels $\bm u_g$ satisfying
\[
 0\le\widehat F_{n,k}(x_k)-u_{gk}\le h_n,\qquad
 \widehat x_{gk}\le x_k,\qquad
 0\le\widehat F_{n,k}(x_k)-\widehat F_{n,k}(\widehat x_{gk})\le h_n.
\]
The closest grid point therefore has empirical distance at most $dh_n$.
The allocation rule and \eqref{eq:SB-rank-lln} give
\begin{equation}
 \begin{cases}
 d_F(\bm x,\bm\pi_n(\bm x))\le dh_n+2dr_n,
       &\bm\pi_n(\bm x)\ne0_*,\\
 \min_kF_k(x_k)\le dh_n+2dr_n,
       &\bm\pi_n(\bm x)=0_*.
 \end{cases}
 \label{eq:SB-grid-cover}
\end{equation}
If $\min_k\widehat F_{n,k}(x_k)\le h_n$, the same two bounds follow
directly from the allocation rule. The increment bounds in \eqref{eq:SA-increment}
make the normalized grid errors $o_p(1)$. These bounds hold for the
observed marginal cdfs, so independence between the grid and process
is unnecessary.

For empirical integration, let $K_0$ be a compact set of bounded,
$d_F$-continuous functions on the completed index space, with
$\sup_{z\in K_0}\|z\|_\infty\le M$. A finite $\epsilon$-net
$z_1,\ldots,z_J$ satisfies
\[
 \Var\{\Pn z_j(\bm X)^2\}\le CM^4/n.
\]
Consequently,
\begin{equation}
 \begin{aligned}
 \sup_{z\in K_0}\left|\int z^2d\widehat F_n-\int z^2dF_{\bm X}\right|
 &\le\max_{j\le J}|(\Pn-P)z_j(\bm X)^2|
       +4M\epsilon+2\epsilon^2\to_p0,
 \end{aligned}
 \label{eq:SB-empirical-integration}
\end{equation}
where $n\to\infty$ precedes $\epsilon\downarrow0$.
Tightness places the normalized processes arbitrarily close to such a
compact set with arbitrarily high probability. This proves empirical
integration for the estimated process, an independent copy of its Gaussian
limit, and the second-order and augmented draws considered below;
independence from $\widehat F_n$ is not required.

\paragraph{Gaussian comparison.}
\label{app:gaussian-comparison}

Let $\rho_n$ be as in Lemma~\ref{lem:linearization}, and let $K$ denote
the limiting covariance of $\rho_n^{-1}\Pn\widehat r$ in either Gaussian
case. Apply the argument separately to each statistic, using the
$J\times J$ covariance matrices defined in Section~\ref{sec:corrected},
where $J=G_n$ for KS and $J=N_n$ for CvM. Define
\[
 \bm A=\rho_n^{-2}\widehat{\bm C}_n,\qquad
 \bm A_+=\rho_n^{-2}\widehat{\bm C}_{n,+},
\]
and
\[
 \bm B=
 \begin{cases}
  \bigl(K(\widehat{\bm x}_g,\widehat{\bm x}_h)\bigr)_{g,h=1}^{G_n},
   &\text{KS},\\[3pt]
  \bigl(K(\bm X_{ij},\bm X_{k\ell})\bigr)_{i<j,\,k<\ell},
   &\text{CvM}.
 \end{cases}
\]
Lemma~\ref{lem:SB-covariance}, which is uniform over $\calX\times\calX$,
gives $\delta_n:=\|\bm A-\bm B\|_{\max}=O_p(n^{-1/2})$.
Since $\bm B\succeq0$, optimality in \eqref{eq:psd-maxnorm} gives
\begin{equation}
 \|\bm A_+-\bm A\|_{\max}\le\delta_n,\qquad
 \Delta_n:=\|\bm A_+-\bm B\|_{\max}
 \le2\delta_n=O_p(n^{-1/2}).
 \label{eq:SB-PSD-max}
\end{equation}
The square-integrable envelope of the mark class also implies
$\max_g\{(A_+)_{gg}\vee B_{gg}\}=O_p(1)$.
The argmin in \eqref{eq:psd-maxnorm} is nonempty, compact, and convex:
zero is feasible, and a minimizing sequence is bounded in maximum norm.
Its value is continuous in the input matrix, so the minimum-Frobenius-norm
tie rule gives a unique measurable selector.
Averaging rows and columns corresponding to repeated evaluation points
preserves positive semidefiniteness and cannot increase the maximum-norm
error. This averaging strictly reduces the Frobenius norm unless those
rows and columns are already identical. Thus the positive-semidefinite projection gives
identical Gaussian evaluations at repeated covariate values.

\begin{lemma}[Gaussian comparison for KS and CvM]
\label{lem:SB-gaussian-comparison}
Let $\bm X,\bm Y$ be centered Gaussian vectors in $\R^G$, with
covariance matrices $\bm\Sigma_X,\bm\Sigma_Y$, and put
$\Delta=\|\bm\Sigma_X-\bm\Sigma_Y\|_{\max}$ and
$V=\max_g\{(\Sigma_X)_{gg}\vee(\Sigma_Y)_{gg}\}$.
Singular covariance matrices and zero coordinate variances are allowed.
For deterministic $w_g\ge0$ with $\sum_gw_g\le1$, write
$M(\bm z)=\max_g|z_g|$ and $Q_w(\bm z)=\sum_gw_gz_g^2$.
For every bounded twice continuously differentiable $\varphi$ with
bounded first and second derivatives,
\begin{align}
 |E\varphi(M(\bm X))-E\varphi(M(\bm Y))|
 &\le \tfrac12\|\varphi''\|_\infty\Delta
   +2\|\varphi'\|_\infty\sqrt{2\Delta\log(2G)},
 \label{eq:SB-KS-comparison}\\
 |E\varphi(Q_w(\bm X))-E\varphi(Q_w(\bm Y))|
 &\le\Delta\{\|\varphi'\|_\infty+2V\|\varphi''\|_\infty\}.
 \label{eq:SB-CvM-comparison}
\end{align}
\end{lemma}

\begin{proof}
For KS, the smooth-maximum comparison of
\citet[Theorem 1]{cck2015comparison}, applied to
$(\bm X^\top,-\bm X^\top)^\top$ and
$(\bm Y^\top,-\bm Y^\top)^\top$, bounds the difference by
\[
 \tfrac12\Delta\|\varphi''\|_\infty+
 \|\varphi'\|_\infty\{\beta\Delta+2\log(2G)/\beta\},\qquad \beta>0.
\]
Choosing $\beta=\sqrt{2\log(2G)/\Delta}$ proves
\eqref{eq:SB-KS-comparison}; the case $\Delta=0$ is immediate.

For CvM, take independent versions of $\bm X,\bm Y$ and set
$\bm Z_t=\sqrt t\bm X+\sqrt{1-t}\bm Y$.
For $F(\bm z)=\varphi(Q_w(\bm z))$,
\[
 \partial_{jk}F(\bm z)
 =2w_j\ind\{j=k\}\varphi'(Q_w(\bm z))
  +4w_jw_kz_jz_k\varphi''(Q_w(\bm z)).
\]
Since $\sum_gw_g\le1$ and $E Z_{t,g}^2\le V$,
\[
 \sum_{j,k}E|\partial_{jk}F(\bm Z_t)|
 \le2\|\varphi'\|_\infty
    +4\|\varphi''\|_\infty E\Big(\sum_gw_g|Z_{t,g}|\Big)^2
 \le2\|\varphi'\|_\infty+4V\|\varphi''\|_\infty.
\]
Integrating
$\frac{d}{dt}EF(\bm Z_t)
=\frac12\sum_{j,k}(\Sigma_X-\Sigma_Y)_{jk}
E[\partial_{jk}F(\bm Z_t)]$
proves \eqref{eq:SB-CvM-comparison}. Adding $\epsilon\bm I$ to both
covariance matrices and taking $\epsilon\downarrow0$ justifies the
same argument when they are singular.
\end{proof}

\section{Corrected First-Order Inference}\label{app:inference}

\subsection{Corrected first-order multiplier critical values}

\begin{proof}[Proof of Theorem~\ref{thm:bootstrap-tests}]
Let $\kappa_{\mathrm{KS}}=1$ and $\kappa_{\mathrm{CvM}}=2$,
with $\rho_n$ as in Lemma~\ref{lem:linearization}. Theorem~\ref{thm:weak}, the grid and integration
arguments, and the conditional Gaussian-statistic comparison in
Section~\ref{app:gaussian-comparison} give
\[
 \rho_n^{-\kappa_S}T_n^S\weak\Phi_S(\mathbb Z),\qquad
 \rho_n^{-\kappa_S}T_{n,1}^{S,*}\weak\Phi_S(\mathbb Z)
 \quad\text{conditionally in probability}.
\]
To include simulation of the critical value, define
\[
 F_n^{S,*}(t)=P^*\{\rho_n^{-\kappa_S}T_{n,1}^{S,*}\le t\},\qquad
 \widehat F_{n,B}^{S,*}(t)=\frac1B\sum_{b=1}^B
       \ind\{\rho_n^{-\kappa_S}T_{n,b}^{S,*}\le t\}.
\]
The conditional Dvoretzky-Kiefer-Wolfowitz inequality gives
\begin{equation}
 P^*\left\{\sup_t|\widehat F_{n,B}^{S,*}(t)-F_n^{S,*}(t)|>\eta\right\}
 \le2e^{-2B\eta^2},\qquad\eta>0.
 \label{eq:SC-empirical-cdf}
\end{equation}
Let $F_S$ be the limiting cdf and $c_S$ its $(1-\alpha)$ quantile.
By Assumption~\ref{ass:critical}, for arbitrarily small $\epsilon>0$
with $c_S\pm\epsilon$ continuity points,
\[
 F_S(c_S-\epsilon)<1-\alpha<F_S(c_S+\epsilon).
\]
Conditional convergence and \eqref{eq:SC-empirical-cdf} imply the same
inequalities for $\widehat F_{n,B}^{S,*}$ with probability tending to one.
Hence, for any $B\to\infty$,
\begin{equation}
 \rho_n^{-\kappa_S}c_n^{S,*}(1-\alpha)\to_p c_S.
 \label{eq:SC-quantile}
\end{equation}
Writing $X_n=\rho_n^{-\kappa_S}T_n^S$ and
$\widehat c_n=\rho_n^{-\kappa_S}c_n^{S,*}(1-\alpha)$,
\begin{align*}
 P\{X_n>c_S+\epsilon\}-P\{|\widehat c_n-c_S|>\epsilon\}
 &\le P\{T_n^S>c_n^{S,*}(1-\alpha)\}\\
 &\le P\{X_n>c_S-\epsilon\}
       +P\{|\widehat c_n-c_S|>\epsilon\}.
\end{align*}
Let $n,B\to\infty$, then $\epsilon\downarrow0$.
Continuity at $c_S$ gives size $\alpha$.
\end{proof}

\subsection{Fixed alternatives}

\begin{proof}[Proof of Theorem~\ref{thm:consistency}]
Consistency and the uniform LLN give
$\|\Pn\widehat r-\Delta\|_\infty\to_p0$.
Since $\Delta$ is uniformly $d_F$-continuous, grid coverage and
empirical integration yield
\[
 T_n^{\mathrm{KS}}\to_p\sup_{\bm x}|\Delta(\bm x)|,
 \qquad T_n^{\mathrm{CvM}}\to_p\int\Delta^2dF_{\bm X}.
\]
It remains to show that the simulated critical values vanish.
The envelope of the estimated marks and Cauchy-Schwarz give, uniformly in the grid,
$\|\widehat{\bm C}_n\|_{\max}=O_p(n^{-1})$.
Since zero is feasible in \eqref{eq:psd-maxnorm},
\[
 \|\widehat{\bm C}_{n,+}-\widehat{\bm C}_n\|_{\max}
 \le\|\widehat{\bm C}_n\|_{\max},\qquad
 \max_g(\widehat C_{n,+})_{gg}
 \le2\|\widehat{\bm C}_n\|_{\max}=O_p(n^{-1}).
\]
A Gaussian union bound and $\log(2G_n)/n\to0$ imply
$P^*\{T_{n,1}^{\mathrm{KS},*}>\epsilon\}\to_p0$.
For CvM, the same diagonal bound gives
\[
 E^*T_{n,1}^{\mathrm{CvM},*}
 =\frac1{N_n}\sum_{i<j}
   \Var^*\{R_{n,1}^{*}(\bm X_{ij})\}
 \le\max_g(\widehat C_{n,+})_{gg}=O_p(n^{-1}).
\]
Conditional Markov's inequality and \eqref{eq:SC-empirical-cdf} therefore give
$c_n^{S,*}(1-\alpha)\to_p0$ for either statistic.
The stated positive population limits prove consistency.
\end{proof}

\subsection{Local alternatives}

For the strict-power remark following Theorem~\ref{thm:local},
the sufficient conditions are
\begin{equation}
 \begin{aligned}
 \mathrm{KS}:&\quad \sup_{\bm x\in\calX}|\mu_d(\bm x)|>0,\qquad
 P\{\|\mathbb Z-\mu_d\|_\infty<\epsilon\}>0
 \quad\text{for every }\epsilon>0,\\
 \mathrm{CvM}:&\quad
 \int\mu_d(\bm x)^2\,dF_{\bm X}(\bm x)>0.
 \end{aligned}
 \label{eq:local-drift}
\end{equation}
The KS condition means that $\mu_d$ is nonzero and belongs to the support
of the limiting Gaussian law under the supremum norm.

\begin{proof}[Proof of Theorem~\ref{thm:local}]
All population derivatives and projections in this proof are evaluated
under the null baseline at $\bm\theta_0$; write
$r_{0,\bm x}=\varepsilon_0q_{\bm x}$.

\textit{(i) Estimated process.}
Under the coupling in \eqref{eq:local-model}, the score at
$\bm\theta_0$ is $\Pn(\bm g\varepsilon_0)+\rho_n\Pn(\bm gd)$.
Taylor expansion gives
\[
 \rho_n^{-1}(\widehat{\bm\theta}-\bm\theta_0)
 =\bm Q^{-1}\{\rho_n^{-1}\Pn(\bm g\varepsilon_0)+E[\bm gd]\}+o_p(1).
\]
The quadratic and empirical-derivative remainders are respectively
$O_p(\rho_n)$ and $O_p(n^{-1/2})$ after normalization.
Uniformly in $\bm x$,
\begin{align*}
 \rho_n^{-1}\Pn\widehat r(\bm x)
 &=\rho_n^{-1}\Pn r_{0,\bm x}+E[d\ind_{\bm x}]
       -\bm M(\bm x)^\top\bm Q^{-1}E[\bm gd]+o_p(1)\\
 &=\rho_n^{-1}\Pn r_{0,\bm x}+\mu_d(\bm x)+o_p(1).
\end{align*}

\textit{(ii) Critical values and rejection probability.}
The mark perturbation from the baseline has $R$- and $D$-seminorms
$O_p(\rho_n n^{-1/2})$ and $O_p(\rho_n N_n^{-1/2})$.
The baseline marks have both seminorms $O_p(\rho_n)$.
Cauchy-Schwarz gives
\[
 \sup_{\bm x,\bm y}\rho_n^{-2}
 |\widehat C_n^{\rm local}(\bm x,\bm y)
  -\widehat C_n^{\rm null}(\bm x,\bm y)|=O_p(n^{-1/2}).
\]
Before the positive-semidefinite projection, the covariance estimator retains its uniform
$O_p(n^{-1/2})$ error relative to the baseline kernel at scale $\rho_n^{-2}$.
Optimality in \eqref{eq:psd-maxnorm} transfers this rate to the projected
matrix by \eqref{eq:SB-PSD-max}. The Gaussian-statistic comparison and
\eqref{eq:SC-empirical-cdf} thus imply
\[
 \rho_n^{-\kappa_S}c_n^{S,*}(1-\alpha)\to_p c_{S,1-\alpha},
 \qquad
 \rho_n^{-\kappa_S}T_n^S\weak\Phi_S(\mathbb Z+\mu_d).
\]
The rejection-probability bounds in the preceding size proof, with
$P\{\Phi_S(\mathbb Z+\mu_d)=c_{S,1-\alpha}\}=0$, prove
\eqref{eq:local-power} for the simulated critical value.

\textit{(iii) Power comparison.}
The acceptance sets $\{z:\Phi_S(z)\le c\}$ are symmetric and convex.
Anderson's inequality gives limiting power at least $\alpha$.
For KS under \eqref{eq:local-drift}, the strict Anderson inequality
\citep{lewandowski1995anderson} gives
\[
 P\{\|\mathbb Z+\mu_d\|_\infty\le c_{\mathrm{KS},1-\alpha}\}
 <P\{\|\mathbb Z\|_\infty\le c_{\mathrm{KS},1-\alpha}\}=1-\alpha.
\]
The stipulated continuity makes open and closed balls equivalent here.

For CvM, let $V$ be the closed linear support of $\mathbb Z$ in
$L_2(F_{\bm X})$, and write $\mu_d=\mu_\parallel+\mu_\perp$,
where $\mu_\parallel\in V$ and $\mu_\perp\perp V$. Then
\[
 \|\mathbb Z+\mu_d\|_2^2
 =\|\mathbb Z+\mu_\parallel\|_2^2+\|\mu_\perp\|_2^2.
\]
Put $c=c_{\mathrm{CvM},1-\alpha}$ and $a=\|\mu_\perp\|_2^2$.
If $a\ge c$, acceptance has probability zero. If $0<a<c$,
\[
 P\{\|\mathbb Z+\mu_d\|_2^2\le c\}
 \le P\{\|\mathbb Z\|_2^2\le c-a\}<1-\alpha.
\]
If $a=0$ and $\int\mu_d^2dF_{\bm X}>0$, then
$\mu_\parallel\ne0$, and strict Anderson gives the same conclusion.
This proves the strict-power remark.
\end{proof}

\paragraph{Strict KS power.}
A nonzero $\mu_d$ alone need not increase KS power when the Gaussian
covariance is singular. For example, let $\bm Z=(Z,0)^\top$,
$Z\sim N(0,1)$, and $\bm\mu=(0,a)^\top$, with
$0<a<c$ and $P\{|Z|>c\}=\alpha$. Then
\[
 P\{\|\bm Z+\bm\mu\|_\infty>c\}
 =P\{\max(|Z|,a)>c\}=\alpha.
\]
Condition~\eqref{eq:local-drift} rules out this possibility; it is
needed only for strict power, not for the local limit formula.

\section{Second-Order Estimation and Envelope Inference}\label{app:second-order}

This section concerns a fixed null distribution with
$\psi_{\bm x}=0$ for every $\bm x\in\calX$. It does not impose a zero
score projection. All vector empirical-process statements below
are understood componentwise, with fixed parameter dimension.
Define the node projections
\begin{equation}
 \bm v(u)=\bm Q^{-1}E[\bm s_{ij}\mid U_i=u],\qquad
 \bm d_{\bm x}(u)
 =E[\bm g(\bm X_{ij})q_{\bm x}(\bm X_{ij})\mid U_i=u].
 \label{eq:estimation-projections}
\end{equation}
The gradient is evaluated at $\bm\theta_*$. Both projections have
mean zero under $H_0$; neither is assumed to vanish.

\subsection{Second-order expansion of the estimated process}

\begin{proof}[Proof of Proposition~\ref{prop:second-order}]
Recall $\bm A(\bm x)=\bm Q^{-1}\bm M(\bm x)$, and put
$\bm Q_n^0=\Pn(\bm g\bm g^\top)$ and
$\bm M_n^0(\bm x)=\Pn(\bm g\ind_{\bm x})$, using derivatives at
$\bm\theta_*$. The score expansion and projection bound give
\begin{equation}
 \bm\delta_n=\bm Q^{-1}\Pn\bm s+O_p(n^{-1})
             =\frac2n\sum_i\bm v(U_i)+O_p(n^{-1}).
 \label{eq:SD-delta-first}
\end{equation}

\textit{(i) Linear identity.}
For $m=\bm g^\top\bm\theta$, the score equation is
$\bm Q_n^0\bm\delta_n=\Pn\bm s$. Hence, exactly,
\begin{align*}
 \Pn\widehat r(\bm x)-\Pn r_{\bm x}
 &=-(\bm M_n^0(\bm x))^\top\bm\delta_n
          +\bm A(\bm x)^\top\Pn\bm s\\
 &=-\{\bm M_n^0(\bm x)-\bm Q_n^0\bm A(\bm x)\}^\top\bm\delta_n
   =-\bm D_n(\bm x)^\top\bm\delta_n.
\end{align*}

\textit{(ii) Consequence of degeneracy.}
Under $H_0$ and Assumption~\ref{ass:zero},
\begin{equation}
 E[\varepsilon\ind_{\bm x}\mid U_i]
 =\bm M(\bm x)^\top\bm v(U_i).
 \label{eq:SD-conditional-measure}
\end{equation}
Intersect the probability-one sets for a countable determining family
of lower orthants. Uniqueness of finite signed measures extends this
identity to every bounded measurable scalar $h$:
\begin{equation}
 E[\varepsilon h(\bm X)\mid U_i]
 =E[h(\bm X)\bm g(\bm X)^\top]\bm v(U_i).
 \label{eq:SD-conditional-h}
\end{equation}
Apply it entrywise to $\bm H$. Since $E[\varepsilon\bm H]=\bm0$,
the projection bound and \eqref{eq:SD-delta-first} imply
\begin{equation}
 \begin{aligned}
 \Pn(\varepsilon\bm H)
 &=E\left[\bm H\left(\bm g^\top\frac2n\sum_i\bm v(U_i)\right)\right]+O_p(n^{-1})\\
 &=E[\bm H(\bm g^\top\bm\delta_n)]+O_p(n^{-1}).
 \end{aligned}
 \label{eq:SD-hessian-cancellation}
\end{equation}
The random coefficients are held fixed inside these population expectations.

\textit{(iii) Nonlinear expansion.}
Expanding the score to second order gives
\begin{align}
 \bm Q\bm\delta_n
 ={}&\Pn\bm s-(\bm Q_n^0-\bm Q)\bm\delta_n
       +\Pn(\varepsilon\bm H)\bm\delta_n\nonumber\\
 &-\tfrac12 E[\bm g(\bm\delta_n^\top\bm H\bm\delta_n)]
       -E[(\bm H\bm\delta_n)(\bm g^\top\bm\delta_n)]
       +o_p(n^{-1}).
 \label{eq:SD-score-second}
\end{align}
The omitted third-derivative term is $o_p(n^{-1})$, because
\[
 \Pn[\varepsilon\,\partial_{\bm\theta}^3m(\bm X,\bm\theta_*)]
 =o_p(1),\qquad \|\bm\delta_n\|^2=O_p(n^{-1}).
\]
Its population coefficient is zero under $H_0$.
Boundedness and continuity of the third derivative control the local
Taylor remainder by dominated convergence; a fourth derivative is
unnecessary. Uniform LLNs justify replacing the other quadratic
coefficients by their expectations.

Expanding the estimated process gives, uniformly in $\bm x$,
\[
 \Pn\widehat r(\bm x)=\Pn(\varepsilon\ind_{\bm x})
  -(\bm M_n^0(\bm x))^\top\bm\delta_n
  -\tfrac12\bm\delta_n^\top E[\bm H\ind_{\bm x}]\bm\delta_n
  +o_p(n^{-1}).
\]
Combining it with \eqref{eq:SD-score-second} yields
\begin{align}
 \Pn\widehat r(\bm x)-\Pn r_{\bm x}
 ={}&-\bm D_n(\bm x)^\top\bm\delta_n
      -\tfrac12\bm\delta_n^\top\bm C(\bm x)\bm\delta_n\nonumber\\
 &-\bm A(\bm x)^\top\Pn(\varepsilon\bm H)\bm\delta_n
   +E[(\bm A(\bm x)^\top\bm H\bm\delta_n)
            (\bm g^\top\bm\delta_n)]+o_p(n^{-1}).
 \label{eq:SD-before-cancellation}
\end{align}
The last two terms cancel to $O_p(n^{-3/2})$ by
\eqref{eq:SD-hessian-cancellation}, uniformly in $\bm x$ because
$\sup_{\bm x}\|\bm A(\bm x)\|<\infty$.
This proves \eqref{eq:second-order-expansion}.
\end{proof}

\subsection{Feasible node draws and joint convergence}

\begin{lemma}[Approximation by node averages]\label{lem:SD-rows}
Let $\calH$ be VC type with $PH^q<\infty$ for some $q>4$, and write
$\bar h_i=(n-1)^{-1}\sum_{j\ne i}h(\bm Z_{ij})$ and
$\widehat h_i=\bar h_i-\Pn h$. Then
\begin{equation}
 E\left[\frac1n\sum_i\sup_{h\in\calH}
 |\bar h_i-Ph-h_1(U_i)|^2\right]\le C PH^2/n.
 \label{eq:SD-row-bound}
\end{equation}
Moreover, replacing $h$ by random estimates $\widehat h_n(h)$
within a fixed VC-type class with a finite $q$th-moment envelope
does not change $\mathbb M_n^*(h)$ asymptotically if
$\sup_h\|\widehat h_n(h)-h\|_{P,2}=o_p(1)$.
\end{lemma}

\begin{proof}
Conditional on $U_i$, the observations $\bm Z_{ij}$, $j\ne i$,
are i.i.d. Symmetrization and the squared VC maximal inequality give
\[
 E\left[\sup_{h\in\calH}|\bar h_i-Ph-h_1(U_i)|^2\mid U_i\right]
 \le\frac{C E[H^2\mid U_i]}{n-1}.
\]
Integrating and averaging proves \eqref{eq:SD-row-bound}.

For multiplier increments, Jensen's inequality gives
\[
 E^*|\mathbb M_n^*(h)-\mathbb M_n^*(h')|^2
 =\frac4n\sum_i(\widehat h_i-\widehat h_i')^2
 \le C\Pn(h-h')^2.
\]
The same contraction transfers the covering bound to centered node
averages. Equation~\eqref{eq:SA-radius} and Gaussian chaining then
give an $o_p(1)$ conditional expected modulus on any $o_p(1)$
population $L_2$ neighborhood. This bound permits estimated functions
based on the full sample.

To replace node averages by node projections, use
\eqref{eq:SD-row-bound}: their empirical squared distance is
$O_p(n^{-1})$. Chaining again makes the difference negligible.
Centering the projections contributes at most
\[
 2\sup_h\left|\frac1n\sum_i h_1(U_i)\right|
   \left|\frac1{\sqrt n}\sum_i\xi_i\right|=o_{P^*}(1)
\]
in probability, completing the proof.
\end{proof}

\subsection{Validity of the augmented second-order multiplier bootstrap}

Use the Gaussian map $\mathbb B$ in Lemma~\ref{lem:joint-quadratic},
independent of $\mathbb G_e$, and set
$\bm Z_\theta=2\mathbb B(\bm v)$ and
$\bm{\mathbb D}(\bm x)=2\mathbb B(\bm d_{\bm x})$, componentwise.
Define
\begin{align}
 \mathbb T_L(\bm x)
 &=\mathbb L(r_{\bm x})
   -\frac1{\sqrt2}\bm{\mathbb D}(\bm x)^\top\bm Z_\theta,
 \label{eq:linear-fitted-limit}\\
 \mathbb T_N(\bm x)
 &=\mathbb T_L(\bm x)
   -\frac1{2\sqrt2}\bm Z_\theta^\top\bm C(\bm x)\bm Z_\theta.
 \label{eq:nonlinear-fitted-limit}
\end{align}
Write $\mathbb T=\mathbb T_N$. For a linear regression,
$\bm C=\bm0$ and $\mathbb T=\mathbb T_L$.

\begin{lemma}[Second-order limit of the estimated process and its bootstrap approximation]
\label{lem:SD-fitted}
Under $H_0$ and Assumptions~\ref{ass:sampling}, \ref{ass:moments},
and~\ref{ass:zero},
\[
 \sqrt{N_n}\Pn\widehat r\weak\mathbb T,
 \qquad
 \sqrt{N_n}R_n^{A,*}\weak\mathbb T
 \quad\text{conditionally in probability},
\]
in $\ell^\infty(\calX)$, where $R_n^{A,*}$ is the augmented draw in
\eqref{eq:curvature-bootstrap}. Neither $\bm v=\bm0$ nor
$\bm d_{\bm x}=\bm0$ is required.
\end{lemma}

\begin{proof}
The coordinates of $\bm Q^{-1}\bm s$ and $\bm gq_{\bm x}$ have
zero means and first-order node projections $\bm v$ and $\bm d_{\bm x}$.
Lemma~\ref{lem:joint-quadratic} and \eqref{eq:SD-delta-first} give
\begin{equation}
 \left(\sqrt{N_n}\Pn r,\sqrt n\,\bm\delta_n,\sqrt n\,\bm D_n\right)
 \weak\left(\mathbb L(r),\bm Z_\theta,\bm{\mathbb D}\right).
 \label{eq:SD-linear-joint}
\end{equation}
All indexed components use their supremum norms.

For the feasible draws, consistency and bounded derivatives give
\begin{align*}
 \sup_{\bm x}\|\widehat r_{\bm x}-r_{\bm x}\|_{P,2}&=o_p(1),&
 \|\widehat{\bm g}\widehat\varepsilon-\bm s\|_{P,2}&=o_p(1),\\
 \sup_{\bm x}\|\widehat{\bm g}\widehat q_{\bm x}-\bm gq_{\bm x}\|_{P,2}
 &=o_p(1),&
 \sup_{\bm x}\|\widehat{\bm C}(\bm x)-\bm C(\bm x)\|&=o_p(1).
\end{align*}
These functions lie in the fixed classes described in
Section~\ref{app:regression}. Corollary~\ref{cor:quadratic-plugin},
Lemma~\ref{lem:SD-rows}, and consistency of $\widehat{\bm Q}$ imply
\begin{equation}
 \left(\sqrt{N_n}R_n^{Q,*},\sqrt n\,\bm\delta_n^*,\sqrt n\,\bm D_n^*\right)
 \weak\left(\mathbb L(r),\bm Z_\theta,\bm{\mathbb D}\right)
 \quad\text{conditionally in probability}.
 \label{eq:SD-bootstrap-joint}
\end{equation}
Using the same multipliers preserves the common $\mathbb B$ in this
joint limit.

Proposition~\ref{prop:second-order} gives
\begin{align*}
 \sqrt{N_n}\Pn\widehat r(\bm x)
 ={}&\sqrt{N_n}\Pn r_{\bm x}
 -\frac{\sqrt{N_n}}n\{\sqrt n\bm D_n(\bm x)\}^\top
                         (\sqrt n\bm\delta_n)\\
 &-\frac{\sqrt{N_n}}{2n}(\sqrt n\bm\delta_n)^\top
                 \bm C(\bm x)(\sqrt n\bm\delta_n)+o_p(1).
\end{align*}
The augmented draw has the same expression with starred quantities
and $\widehat{\bm C}$. Since $\sqrt{N_n}/n\to1/\sqrt2$, the continuous
mapping theorem proves both limits. For linear regression,
$\bm C=\bm0$, giving $\mathbb T_L$. The products remain uncentered,
so their possibly nonzero means are retained.
\end{proof}

If $\bm v=\bm0$, both estimation terms vanish in this limit.
For linear regression, the parameter-estimation term also vanishes when
$\bm d_{\bm x}=\bm0$ for every $\bm x$. Thus these cases are
included without imposing either condition on the full class.
Section~\ref{app:examples} illustrates both negligible corrections
and cancellation between interaction and estimation effects.

\subsection{KS and CvM critical values for the envelope test}
\label{app:envelope-inference}

The node-moment class has
$\|\bm gq_{\bm x}-\bm gq_{\bm y}\|_{P,2}
\le C d_F(\bm x,\bm y)^{1/2}$, and the bounded Hessian gives
$\|\bm C(\bm x)-\bm C(\bm y)\|\le C d_F(\bm x,\bm y)$.
Together with \eqref{eq:SA-increment}, these bounds give continuous
versions of the limits of the estimated process on the completed index space
defined after \eqref{eq:SA-increment}.
The tightness and empirical-integration arguments of
Section~\ref{app:regression} therefore apply to the augmented
bootstrap processes as well. In particular, taking a product of tight,
equicontinuous processes preserves the needed modulus bound.

\begin{proof}[Proof of Theorem~\ref{thm:envelope}]
Let $\kappa_{\mathrm{KS}}=1$ and $\kappa_{\mathrm{CvM}}=2$.

\textit{(i) The degenerate case.}
Lemma~\ref{lem:SD-fitted} and the grid and integration arguments give
\[
 N_n^{\kappa_S/2}T_n^S\weak\Phi_S(\mathbb T),\qquad
 N_n^{\kappa_S/2}T_{n,1}^{S,A,*}\weak\Phi_S(\mathbb T)
 \quad\text{conditionally in probability}.
\]
The quantile condition and \eqref{eq:SC-empirical-cdf} imply convergence
of the normalized simulated augmented critical value to the limiting
quantile. Thus
\[
 P\{T_n^S>c_{A,n}^{S,*}(1-\alpha)\}\to\alpha,
 \qquad
 \{T_n^S>c_{E,n}^{S,*}\}\subseteq\{T_n^S>c_{A,n}^{S,*}\}.
\]
The inclusion proves the conservative bound without requiring a limit
for the corrected first-order critical value.

\textit{(ii) The nondegenerate case.}
The maximal bounds give, conditionally in probability,
\[
 \|R_n^{Q,*}\|_\infty=O_{P^*}(n^{-1}),\qquad
 \|\bm\delta_n^*\|+\sup_{\bm x}\|\bm D_n^*(\bm x)\|
 =O_{P^*}(n^{-1/2}),\qquad
 \sup_{\bm x}\|\widehat{\bm C}(\bm x)\|=O_p(1).
\]
Hence $\|R_n^{A,*}\|_\infty=O_{P^*}(n^{-1})$. With $\rho_n=n^{-1/2}$,
\[
 \rho_n^{-\kappa_S}c_{A,n}^{S,*}(1-\alpha)\to_p0,
 \qquad \rho_n^{-\kappa_S}c_n^{S,*}(1-\alpha)\to_p c_S>0.
\]
Their maximum has the same limit as the corrected critical value.

\textit{(iii) Independent dyads (a special case of degeneracy).}
A representation using only dyad variables gives
$\bm v=\bm0$ and $\bm d_{\bm x}=\bm0$. Equation~\eqref{eq:SD-bootstrap-joint}
then implies
\[
 \sqrt{N_n}\|R_n^{A,*}-R_n^{Q,*}\|_\infty=o_{P^*}(1).
\]
The augmented and corrected draws converge to the same Gaussian law.
With $\rho_n=N_n^{-1/2}$, both normalized simulated critical values, and
their maximum, converge to $c_S$.

In (ii)-(iii), Theorem~\ref{thm:bootstrap-tests} and the same
rejection-probability bounds prove exactness. The two collections of
draws need not be independent: \eqref{eq:SC-empirical-cdf} applies
separately to each.
\end{proof}

\paragraph{Fixed and local alternatives.}
Under a fixed alternative, the classes containing the estimated functions retain bounded envelope
moments. Therefore
\[
 \|R_n^{Q,*}\|_\infty=O_{P^*}(n^{-1}),\qquad
 \|\bm\delta_n^*\|+\|\bm D_n^*\|_\infty=O_{P^*}(n^{-1/2}),
 \qquad \|\widehat{\bm C}\|_\infty=O_p(1).
\]
The augmented critical value vanishes by \eqref{eq:SC-empirical-cdf}.
Together with Theorem~\ref{thm:consistency}, this proves consistency
of the envelope test.

Under the Gaussian local alternatives, the estimated mark and node-moment
classes approach their baseline counterparts in $L_2(P)$.
Corollary~\ref{cor:quadratic-plugin} and Lemma~\ref{lem:SD-rows} preserve
the bootstrap limits: the augmented draw is negligible at the
$\sqrt n$ rate in case (i), and has the same Gaussian limit as the
corrected draw in case (ii). Thus the envelope retains the critical
value and local-power conclusions of Theorem~\ref{thm:local}.

\subsection{Examples separating interaction and estimation}\label{app:examples}

These examples also check the signs and joint dependence in the
expansion. All node and dyad variables mentioned as independent form
mutually independent collections. Bounded variables can be used
throughout to satisfy the moment assumptions.

\paragraph{A nonzero score projection with no parameter-estimation term.}
Let the regression model contain only an intercept and let
$Y_{ij}=\theta_*+A_i+A_j+v_{ij}$, with centered i.i.d. $A_i$ and
centered dyad noise $v_{ij}$. Let dyad covariates be i.i.d. and
independent of these variables. Then
$q_{\bm x}=\ind_{\bm x}-F_{\bm X}(\bm x)$,
$\psi_{\bm x}=0$, $\gamma_{\bm x}=0$, and
$E[s_{ij}\mid U_i]=A_i$. However $d_{\bm x}=0$ and
$D_n(\bm x)=O_p(n^{-1})$ uniformly, whereas
$\delta_n=O_p(n^{-1/2})$. Their product is negligible.
This example shows why a zero score projection is an unnecessarily
strong sufficient condition for the whole degenerate class.

\paragraph{Estimation alone can leave a non-Gaussian term.}
Let $B_i$ be independent Rademacher variables and fix $0<c<1/4$.
Given the nodes, draw $X_{ij}\in\{0,1\}$ independently with probability
$p_{ij}=1/2+c(B_i+B_j)$ of one. Write
$w_{ij}=(1/2)/\{p_{ij}^{X_{ij}}(1-p_{ij})^{1-X_{ij}}\}$ and let
$Y_{ij}=(B_i+B_j)w_{ij}+v_{ij}$, where $v_{ij}$ is independent,
centered dyad noise. Fit an intercept, with true value zero.
For each $x\in\{0,1\}$,
$E[(B_i+B_j)w_{ij}\ind\{X_{ij}=x\}]=E[B_i+B_j]/2=0$,
so the null holds. At $x=0$,
$q_0=\ind\{X=0\}-1/2$. Direct conditional expectation gives
\[
 \psi_0=\gamma_0=0,\qquad v(B_i)=B_i,\qquad d_0(B_i)=-cB_i.
\]
Thus, if $Z\sim N(0,1)$,
\[
 n\Pn\widehat r(0)\weak
       \sqrt2\,\mathbb G_e(0)+4cZ^2,
 \qquad Z\ \text{independent of }\mathbb G_e(0).
\]
The mean $4c$ would be lost by centering the bootstrap product.

\paragraph{Interaction and estimation can cancel.}
Let $B_i$ be Rademacher, let $A_i$ be centered and independent of
$B_i$, and set
$X_{ij}=B_i+B_j$,
$Y_{ij}=\theta_0+\theta_1X_{ij}+A_i+A_j+v_{ij}$.
Fit an intercept and slope. Then
$\bm Q=\operatorname{diag}(1,2)$, and, at $x=-2$,
\begin{align*}
 q_{-2}&=B_iB_j/4,&
 \gamma_{-2}(U_i,U_j)&=(A_i+A_j)B_iB_j/4,\\
 \bm v(U_i)&=(A_i,A_iB_i/2)^\top,&
 \bm d_{-2}(U_i)&=(0,B_i/4)^\top.
\end{align*}
Here $\psi_x=0$ for every $x$, so the process is degenerate.
The interaction contribution to $n\Pn r_{-2}$ is
$\tfrac12\mathbb B(AB)\mathbb B(B)$, exactly the limit of
$n\bm D_n(-2)^\top\bm\delta_n$. They cancel in the estimated process.
This cancellation is preserved by common multipliers and generally
lost with independent multiplier sequences. The independent dyad
noise can be retained to leave a nontrivial Gaussian limit.

\paragraph{Nonlinear curvature can remain without interaction.}
Let $m(X,\theta)=\theta+\theta^2h(X)$, with $Eh(X)=0$,
and let the true parameter be zero. Take independent dyad covariates
and $Y_{ij}=A_i+A_j+v_{ij}$. Then $g=1$, $H=2h$,
$\psi_x=\gamma_x=d_x=0$, but
$C(x)=2E[h(X)\ind\{X\le x\}]$ need not vanish.
If $\sigma_A^2=EA_i^2$, the curvature contribution to
$n\Pn\widehat r(x)$ is
$-4\sigma_A^2 E[h(X)\ind\{X\le x\}]Z^2$.
The third term of \eqref{eq:curvature-bootstrap} reproduces it.

\section{Verification of the simulation designs}\label{app:simulation-checks}

Write $A_i=U_i^\varepsilon$ and $B_i=U_i^x$.
The representations in \eqref{eq:sim-squared} and \eqref{eq:sim-strong}
are symmetric and dissociated. Bounded primitives and compact parameter
sets give the required moment and derivative bounds, while independence
of the error and covariate primitives gives $E[\varepsilon_{ij}\mid X_{ij}]=0$.
Each design has at least three distinct covariate values. Equality of
the squared affine means at these values implies
\[
 \theta_0^2=\theta_0\theta_1=\theta_1^2=1,
 \qquad (\theta_0,\theta_1)\in[1/2,3/2]^2
 \quad\Longrightarrow\quad \bm\theta_*=(1,1)^\top.
\]
Also $\bm Q$ is positive definite. Thus all four designs satisfy
Assumptions~\ref{ass:sampling} and~\ref{ass:moments} for every reported
$\lambda$. The consistent least-squares estimator is interior with probability
tending to one because the true parameter is interior.

\paragraph{The Gaussian designs.}
For $\omega=0$, the dyads are independent, a special case of degeneracy.
Dyad-specific noise gives $\Gamma(x,x)>0$ for some $x$, as required by
Assumption~\ref{ass:nondeg}(ii).
The coefficient $\lambda$ does not affect this design. For the nondegenerate weak Gaussian
design ($\omega=1/2$), direct enumeration gives
\[
 \psi_{-3}(U_i)=\lambda\frac{3}{368}A_iB_i,
 \qquad \Omega(-3,-3)=\lambda^2\frac{3}{33856}>0.
\]
For the nondegenerate strong Gaussian design, the loading $1+X^2$ in \eqref{eq:sim-strong} gives
\[
 \psi_0(U_i)=\lambda\frac{205-118B_i}{368}A_i.
\]
Both satisfy Assumption~\ref{ass:nondeg}(i) for every $\lambda>0$.

Let $H(X)$ denote the error loading, equal to $1+X$ or $1+X^2$. Then
\[
 \Gamma(x,x)=\frac{2\lambda^2\omega^2+1}{3}E[H(X)^2q_x(X)^2].
\]
By \eqref{eq:correction-identity}, the share of integrated variance
due to covariance between distinct dyads sharing a node is
\[
 \frac{(n-2)\int\Omega(x,x)\,dF_X(x)}
 {(n-2)\int\Omega(x,x)\,dF_X(x)+2\int\Gamma(x,x)\,dF_X(x)}.
\]
Table~\ref{tab:node-share} gives this population quantity. Increasing
$\lambda$ strengthens the node contribution relative to independent
dyad noise without changing the Gaussian case.

\begin{table}[!htbp]\centering
\caption{Integrated variance share from covariance between distinct dyads}
\label{tab:node-share}
\begingroup\setstretch{1}\small
\begin{tabular}{lcrr}\toprule
Gaussian design & $\lambda$ & $n=30$ & $n=200$\\\midrule
Nondegenerate weak & 1 & 0.041 & 0.234\\
Nondegenerate weak & 4 & 0.103 & 0.449\\
Nondegenerate weak & 10 & 0.113 & 0.473\\
Nondegenerate strong & 1 & 0.673 & 0.936\\
Nondegenerate strong & 4 & 0.846 & 0.975\\
Nondegenerate strong & 10 & 0.858 & 0.977\\
\bottomrule\end{tabular}\endgroup\end{table}

\paragraph{The degenerate non-Gaussian design.}
For $\omega=1$, $X$ is the sum of four independent symmetric signs.
For any function $a$ on its five possible values,
\begin{align}
 E[a(X)\mid B_i]
 &=E[a(X)]+\frac{B_i}{4}E[Xa(X)],\\
 E[a(X)\mid B_i,B_j]
 &=E[a(X)]+\frac{B_i+B_j}{4}E[Xa(X)]
   +\frac{B_iB_j}{12}E[(X^2-4)a(X)].
\end{align}
Here $\bm g(X)=2(1+X)(1,X)^\top$. Projection orthogonality gives
$E[(1+X)q_x(X)]=E[X(1+X)q_x(X)]=0$. Hence
\[
 \psi_x(U_i)=\lambda A_iE[(1+X)q_x(X)\mid B_i]=0
 \quad\text{for every }x.
\]
Define, for this verification only,
\[
 a_x=E[q_x(X)],\qquad t_x=E[X^2(1+X)q_x(X)].
\]
Elementary moment calculations give
\begin{align}
 \gamma(U_i,U_j,r_x)
 &=\frac{\lambda t_x}{12}(A_i+A_j)B_iB_j,\\*
 \bm v(U_i)&=\lambda\left(A_i/2,A_iB_i/8\right)^\top,\\*
 \bm d_x(U_i)&=\left(0,B_it_x/2\right)^\top,\\*
 \bm C(x)&=2a_x\begin{pmatrix}1&-1\\-1&1\end{pmatrix}.
\end{align}
At $x=-4$, $a_x=1/52$ and $t_x=-87/52$. Thus the interaction,
parameter-estimation term, and curvature term are all
nonzero for every $\lambda>0$.

Let $Z_A=\mathbb B(A)$, $Z_{AB}=\mathbb B(AB)$, and $Z_B=\mathbb B(B)$.
These are independent centered Gaussian variables with variances
$1/3$, $1/3$, and $1$, respectively. Under the normalization $n$,
\begin{equation}
 n\Pn\widehat r(x)\weak
 \sqrt2\mathbb G_{e,\lambda}(x)
 -\frac{\lambda t_x}{12}Z_{AB}Z_B
 -\lambda^2a_x\left(Z_A-Z_{AB}/4\right)^2.
 \label{eq:sim-fitted-limit}
\end{equation}
The Gaussian process $\mathbb G_{e,\lambda}$ is independent of these
three variables. Its variance is
\[
 E[\mathbb G_{e,\lambda}(x)^2]
 =\frac{2\lambda^2+1}{3}E[(1+X)^2q_x(X)^2]
  -\frac{\lambda^2t_x^2}{216}.
\]
The interaction and parameter-estimation term partially cancel,
leaving the product term in \eqref{eq:sim-fitted-limit}. Curvature
contributes the final squared Gaussian term, which grows with
$\lambda^2$. The variance of the sum of the two non-Gaussian terms is
\[
 \frac{\lambda^2t_x^2}{432}
 +2\lambda^4a_x^2\left(\frac{17}{48}\right)^2.
\]
Integrating this variance with respect to $F_X$ and dividing by the
integrated variance of the full limit gives 0.039, 0.371, and 0.785
for $\lambda=1,4,10$, respectively. These shares concern variance;
the mean of the non-Gaussian component is
$-17\lambda^2a_x/48$.

Independent dyad noise gives a nonzero Gaussian component for each
$\lambda$. Since $1+X$ is nonzero at all five covariate values, any nonzero linear
combination of the functions $q_x$ has positive variance from this noise. The limiting KS
and CvM statistics have continuous, strictly increasing distributions
at their positive critical values. The Gaussian designs also have
nontrivial continuous limiting statistics. The finite covariate set
allows these calculations to be verified by finite sums.

\subsection{Additional simulation results}

All size calculations use the nominal 5\% level, 20,000 replications,
and 999 draws. The four procedures use the same generated data in each
replication. For a given baseline design and sample size, primitive
variables and multiplier draws are also shared across $\lambda$;
the two $\omega=1/2$ designs share these draws. Replications are
independent. Tables~\ref{tab:size-independent}-\ref{tab:size-nong}
report all 70 distinct size settings.

\begin{table}[!htbp]\centering
\caption{Size: independent dyads}\label{tab:size-independent}
\begingroup\setstretch{1}\small\setlength{\tabcolsep}{4pt}
\begin{tabular}{ccrrrrrrrr}\toprule
 & & \multicolumn{2}{c}{Uncorrected} & \multicolumn{2}{c}{Corrected} & \multicolumn{2}{c}{Augmented} & \multicolumn{2}{c}{Envelope}\\
\cmidrule(lr){3-4}\cmidrule(lr){5-6}\cmidrule(lr){7-8}\cmidrule(lr){9-10}
$\lambda$ & $n$ & KS & CvM & KS & CvM & KS & CvM & KS & CvM\\\midrule
- & 30 & 0.013 & 0.013 & 0.130 & 0.130 & 0.047 & 0.047 & 0.040 & 0.040\\
 & 40 & 0.010 & 0.010 & 0.106 & 0.106 & 0.047 & 0.047 & 0.040 & 0.040\\
 & 60 & 0.008 & 0.008 & 0.083 & 0.083 & 0.048 & 0.048 & 0.041 & 0.041\\
 & 90 & 0.008 & 0.008 & 0.074 & 0.074 & 0.050 & 0.050 & 0.043 & 0.043\\
 & 120 & 0.007 & 0.007 & 0.066 & 0.066 & 0.049 & 0.049 & 0.042 & 0.042\\
 & 150 & 0.006 & 0.006 & 0.060 & 0.060 & 0.047 & 0.047 & 0.041 & 0.041\\
 & 200 & 0.006 & 0.006 & 0.060 & 0.060 & 0.049 & 0.049 & 0.043 & 0.043\\
\addlinespace[3pt]
\bottomrule\end{tabular}
\par\smallskip\begin{minipage}{.97\textwidth}\footnotesize Nominal level 5\%; 20,000 replications and 999 draws. The dyad-noise coefficient is one. The coefficient $\lambda$ drops out under independent dyads.\end{minipage}\endgroup\end{table}

\begin{table}[!htbp]\centering
\caption{Size: nondegenerate strong Gaussian}\label{tab:size-strong}
\begingroup\setstretch{1}\small\setlength{\tabcolsep}{4pt}
\begin{tabular}{ccrrrrrrrr}\toprule
 & & \multicolumn{2}{c}{Uncorrected} & \multicolumn{2}{c}{Corrected} & \multicolumn{2}{c}{Augmented} & \multicolumn{2}{c}{Envelope}\\
\cmidrule(lr){3-4}\cmidrule(lr){5-6}\cmidrule(lr){7-8}\cmidrule(lr){9-10}
$\lambda$ & $n$ & KS & CvM & KS & CvM & KS & CvM & KS & CvM\\\midrule
1 & 30 & 0.046 & 0.050 & 0.108 & 0.102 & 0.208 & 0.231 & 0.103 & 0.100\\
 & 40 & 0.044 & 0.048 & 0.091 & 0.086 & 0.259 & 0.287 & 0.091 & 0.086\\
 & 60 & 0.046 & 0.049 & 0.075 & 0.073 & 0.358 & 0.388 & 0.075 & 0.073\\
 & 90 & 0.045 & 0.048 & 0.066 & 0.064 & 0.442 & 0.473 & 0.066 & 0.064\\
 & 120 & 0.046 & 0.050 & 0.064 & 0.062 & 0.511 & 0.542 & 0.064 & 0.062\\
 & 150 & 0.044 & 0.046 & 0.056 & 0.054 & 0.553 & 0.579 & 0.056 & 0.054\\
 & 200 & 0.050 & 0.051 & 0.059 & 0.060 & 0.614 & 0.635 & 0.059 & 0.060\\
\addlinespace[3pt]
4 & 30 & 0.067 & 0.076 & 0.106 & 0.108 & 0.242 & 0.243 & 0.100 & 0.103\\
 & 40 & 0.069 & 0.074 & 0.098 & 0.097 & 0.308 & 0.299 & 0.097 & 0.096\\
 & 60 & 0.069 & 0.071 & 0.087 & 0.086 & 0.399 & 0.381 & 0.087 & 0.085\\
 & 90 & 0.062 & 0.068 & 0.073 & 0.076 & 0.486 & 0.464 & 0.073 & 0.076\\
 & 120 & 0.063 & 0.064 & 0.070 & 0.070 & 0.551 & 0.530 & 0.070 & 0.070\\
 & 150 & 0.057 & 0.060 & 0.063 & 0.064 & 0.596 & 0.574 & 0.063 & 0.064\\
 & 200 & 0.061 & 0.062 & 0.065 & 0.065 & 0.655 & 0.634 & 0.065 & 0.065\\
\addlinespace[3pt]
10 & 30 & 0.114 & 0.141 & 0.152 & 0.172 & 0.134 & 0.145 & 0.095 & 0.109\\
 & 40 & 0.105 & 0.132 & 0.134 & 0.156 & 0.169 & 0.173 & 0.099 & 0.116\\
 & 60 & 0.090 & 0.116 & 0.110 & 0.132 & 0.225 & 0.219 & 0.092 & 0.111\\
 & 90 & 0.080 & 0.100 & 0.092 & 0.110 & 0.293 & 0.274 & 0.082 & 0.097\\
 & 120 & 0.077 & 0.092 & 0.087 & 0.100 & 0.349 & 0.317 & 0.081 & 0.093\\
 & 150 & 0.071 & 0.083 & 0.079 & 0.089 & 0.375 & 0.340 & 0.076 & 0.086\\
 & 200 & 0.075 & 0.084 & 0.081 & 0.089 & 0.431 & 0.390 & 0.080 & 0.087\\
\addlinespace[3pt]
\bottomrule\end{tabular}
\par\smallskip\begin{minipage}{.97\textwidth}\footnotesize Nominal level 5\%; 20,000 replications and 999 draws. The dyad-noise coefficient is one. Augmented includes both estimation corrections.\end{minipage}\endgroup\end{table}

\begin{table}[!htbp]\centering
\caption{Size: nondegenerate weak Gaussian}\label{tab:size-weak}
\begingroup\setstretch{1}\small\setlength{\tabcolsep}{4pt}
\begin{tabular}{ccrrrrrrrr}\toprule
 & & \multicolumn{2}{c}{Uncorrected} & \multicolumn{2}{c}{Corrected} & \multicolumn{2}{c}{Augmented} & \multicolumn{2}{c}{Envelope}\\
\cmidrule(lr){3-4}\cmidrule(lr){5-6}\cmidrule(lr){7-8}\cmidrule(lr){9-10}
$\lambda$ & $n$ & KS & CvM & KS & CvM & KS & CvM & KS & CvM\\\midrule
1 & 30 & 0.015 & 0.014 & 0.136 & 0.139 & 0.048 & 0.045 & 0.042 & 0.039\\
 & 40 & 0.012 & 0.012 & 0.116 & 0.116 & 0.055 & 0.050 & 0.046 & 0.043\\
 & 60 & 0.011 & 0.009 & 0.090 & 0.088 & 0.060 & 0.051 & 0.049 & 0.042\\
 & 90 & 0.009 & 0.008 & 0.074 & 0.073 & 0.072 & 0.056 & 0.052 & 0.046\\
 & 120 & 0.010 & 0.009 & 0.068 & 0.068 & 0.085 & 0.065 & 0.058 & 0.051\\
 & 150 & 0.011 & 0.008 & 0.062 & 0.062 & 0.094 & 0.070 & 0.058 & 0.051\\
 & 200 & 0.012 & 0.008 & 0.060 & 0.057 & 0.110 & 0.077 & 0.059 & 0.052\\
\addlinespace[3pt]
4 & 30 & 0.019 & 0.019 & 0.119 & 0.113 & 0.066 & 0.062 & 0.047 & 0.040\\
 & 40 & 0.017 & 0.015 & 0.096 & 0.092 & 0.075 & 0.064 & 0.049 & 0.039\\
 & 60 & 0.015 & 0.011 & 0.076 & 0.070 & 0.087 & 0.066 & 0.052 & 0.036\\
 & 90 & 0.015 & 0.011 & 0.060 & 0.055 & 0.108 & 0.073 & 0.053 & 0.036\\
 & 120 & 0.016 & 0.011 & 0.057 & 0.053 & 0.128 & 0.084 & 0.054 & 0.041\\
 & 150 & 0.017 & 0.011 & 0.053 & 0.046 & 0.147 & 0.092 & 0.053 & 0.040\\
 & 200 & 0.020 & 0.013 & 0.051 & 0.044 & 0.181 & 0.112 & 0.051 & 0.042\\
\addlinespace[3pt]
10 & 30 & 0.020 & 0.022 & 0.093 & 0.086 & 0.038 & 0.026 & 0.026 & 0.015\\
 & 40 & 0.019 & 0.018 & 0.075 & 0.064 & 0.046 & 0.029 & 0.026 & 0.014\\
 & 60 & 0.015 & 0.010 & 0.052 & 0.044 & 0.060 & 0.039 & 0.024 & 0.013\\
 & 90 & 0.013 & 0.009 & 0.041 & 0.031 & 0.079 & 0.055 & 0.023 & 0.011\\
 & 120 & 0.014 & 0.009 & 0.037 & 0.027 & 0.093 & 0.064 & 0.024 & 0.009\\
 & 150 & 0.014 & 0.007 & 0.035 & 0.023 & 0.102 & 0.071 & 0.024 & 0.009\\
 & 200 & 0.015 & 0.007 & 0.034 & 0.022 & 0.122 & 0.081 & 0.027 & 0.009\\
\addlinespace[3pt]
\bottomrule\end{tabular}
\par\smallskip\begin{minipage}{.97\textwidth}\footnotesize Nominal level 5\%; 20,000 replications and 999 draws. The dyad-noise coefficient is one. Augmented includes both estimation corrections.\end{minipage}\endgroup\end{table}

\begin{table}[!htbp]\centering
\caption{Size: dependent non-Gaussian design}\label{tab:size-nong}
\begingroup\setstretch{1}\small\setlength{\tabcolsep}{4pt}
\begin{tabular}{ccrrrrrrrr}\toprule
 & & \multicolumn{2}{c}{Uncorrected} & \multicolumn{2}{c}{Corrected} & \multicolumn{2}{c}{Augmented} & \multicolumn{2}{c}{Envelope}\\
\cmidrule(lr){3-4}\cmidrule(lr){5-6}\cmidrule(lr){7-8}\cmidrule(lr){9-10}
$\lambda$ & $n$ & KS & CvM & KS & CvM & KS & CvM & KS & CvM\\\midrule
1 & 30 & 0.017 & 0.017 & 0.176 & 0.170 & 0.043 & 0.045 & 0.037 & 0.039\\
 & 40 & 0.012 & 0.012 & 0.151 & 0.143 & 0.047 & 0.046 & 0.040 & 0.039\\
 & 60 & 0.009 & 0.010 & 0.117 & 0.113 & 0.046 & 0.047 & 0.039 & 0.041\\
 & 90 & 0.008 & 0.008 & 0.093 & 0.088 & 0.050 & 0.048 & 0.042 & 0.041\\
 & 120 & 0.007 & 0.006 & 0.082 & 0.081 & 0.046 & 0.047 & 0.040 & 0.040\\
 & 150 & 0.008 & 0.007 & 0.077 & 0.075 & 0.050 & 0.050 & 0.043 & 0.043\\
 & 200 & 0.008 & 0.008 & 0.070 & 0.070 & 0.051 & 0.050 & 0.044 & 0.044\\
\addlinespace[3pt]
4 & 30 & 0.022 & 0.020 & 0.143 & 0.124 & 0.043 & 0.038 & 0.034 & 0.027\\
 & 40 & 0.018 & 0.016 & 0.116 & 0.097 & 0.045 & 0.043 & 0.032 & 0.026\\
 & 60 & 0.012 & 0.011 & 0.084 & 0.068 & 0.051 & 0.049 & 0.030 & 0.021\\
 & 90 & 0.008 & 0.006 & 0.066 & 0.054 & 0.051 & 0.049 & 0.023 & 0.014\\
 & 120 & 0.008 & 0.005 & 0.056 & 0.045 & 0.051 & 0.051 & 0.020 & 0.012\\
 & 150 & 0.006 & 0.004 & 0.054 & 0.042 & 0.050 & 0.051 & 0.019 & 0.010\\
 & 200 & 0.006 & 0.004 & 0.048 & 0.037 & 0.052 & 0.053 & 0.019 & 0.010\\
\addlinespace[3pt]
10 & 30 & 0.070 & 0.060 & 0.200 & 0.175 & 0.018 & 0.012 & 0.016 & 0.010\\
 & 40 & 0.073 & 0.064 & 0.185 & 0.157 & 0.021 & 0.012 & 0.018 & 0.010\\
 & 60 & 0.070 & 0.063 & 0.146 & 0.122 & 0.029 & 0.017 & 0.025 & 0.013\\
 & 90 & 0.057 & 0.048 & 0.111 & 0.089 & 0.032 & 0.019 & 0.027 & 0.015\\
 & 120 & 0.050 & 0.042 & 0.088 & 0.069 & 0.036 & 0.021 & 0.030 & 0.017\\
 & 150 & 0.042 & 0.032 & 0.072 & 0.057 & 0.041 & 0.026 & 0.032 & 0.019\\
 & 200 & 0.031 & 0.023 & 0.055 & 0.039 & 0.044 & 0.030 & 0.028 & 0.018\\
\addlinespace[3pt]
\bottomrule\end{tabular}
\par\smallskip\begin{minipage}{.97\textwidth}\footnotesize Nominal level 5\%; 20,000 replications and 999 draws. The dyad-noise coefficient is one. Augmented includes both estimation corrections.\end{minipage}\endgroup\end{table}

\FloatBarrier

\paragraph{The second-order corrections.}
Table~\ref{tab:second-order} adds the two corrections in turn: the
parameter-estimation term and then the curvature term.
At $\lambda=4$ and $n=200$, adding curvature changes KS rejection
from 0.109 to 0.052 and CvM rejection from 0.134 to 0.053.
At $\lambda=10$, the simpler draws overreject while the augmented draw
is conservative. The comparison shows the importance of the estimation
terms in the second-order expansion. The counts below show how often the estimate lies on the parameter boundary at the largest $\lambda$.

\begin{table}[!htbp]\centering
\caption{Role of the second-order corrections in the non-Gaussian design}\label{tab:second-order}
\begingroup\setstretch{1}\small\setlength{\tabcolsep}{4pt}
\begin{tabular}{ccrrrrrrrr}\toprule
 & & \multicolumn{2}{c}{Second-order only} & \multicolumn{2}{c}{$+$ estimation term} & \multicolumn{2}{c}{Augmented} & \multicolumn{2}{c}{Envelope}\\
\cmidrule(lr){3-4}\cmidrule(lr){5-6}\cmidrule(lr){7-8}\cmidrule(lr){9-10}
$\lambda$ & $n$ & KS & CvM & KS & CvM & KS & CvM & KS & CvM\\\midrule
1 & 30 & 0.048 & 0.048 & 0.045 & 0.046 & 0.043 & 0.045 & 0.037 & 0.039\\
 & 90 & 0.045 & 0.046 & 0.052 & 0.052 & 0.050 & 0.048 & 0.042 & 0.041\\
 & 200 & 0.047 & 0.048 & 0.054 & 0.055 & 0.051 & 0.050 & 0.044 & 0.044\\
\addlinespace[3pt]
4 & 30 & 0.095 & 0.108 & 0.091 & 0.107 & 0.043 & 0.038 & 0.034 & 0.027\\
 & 90 & 0.096 & 0.117 & 0.108 & 0.131 & 0.051 & 0.049 & 0.023 & 0.014\\
 & 200 & 0.096 & 0.120 & 0.109 & 0.134 & 0.052 & 0.053 & 0.019 & 0.010\\
\addlinespace[3pt]
10 & 30 & 0.216 & 0.227 & 0.211 & 0.224 & 0.018 & 0.012 & 0.016 & 0.010\\
 & 90 & 0.250 & 0.275 & 0.265 & 0.291 & 0.032 & 0.019 & 0.027 & 0.015\\
 & 200 & 0.247 & 0.280 & 0.262 & 0.294 & 0.044 & 0.030 & 0.028 & 0.018\\
\addlinespace[3pt]
\bottomrule\end{tabular}
\par\smallskip\begin{minipage}{.97\textwidth}\footnotesize The first draw applies product weights to the centered estimated marks. The second subtracts the parameter-estimation term. The augmented draw also subtracts the curvature term. All three terms use the same node multipliers.\end{minipage}\endgroup\end{table}

\FloatBarrier
\paragraph{Power calculations.}
The local-power experiments use $n=120$. Figures~\ref{fig:power-lambda4}
and~\ref{fig:power-lambda10} complement the main-text results with
$\lambda=4$ and $10$. The independent-dyad design is unchanged by
$\lambda$ and is not repeated.
\begin{figure}[!htbp]\centering
\includegraphics[width=\textwidth]{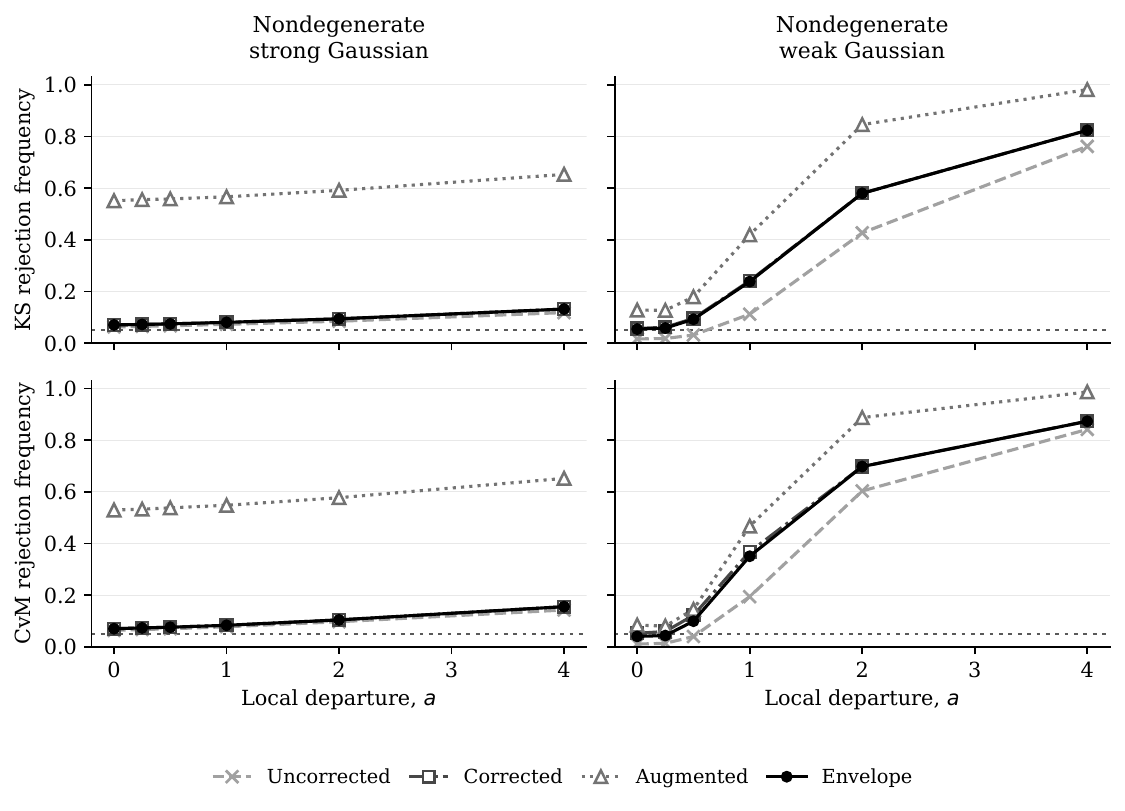}
\caption{Local power at $n=120$ and $\lambda=4$. The added conditional-mean term is $(a/\sqrt n)X^2$. Each point uses 20,000 replications and 999 draws.}\label{fig:power-lambda4}\end{figure}
\begin{figure}[!htbp]\centering
\includegraphics[width=\textwidth]{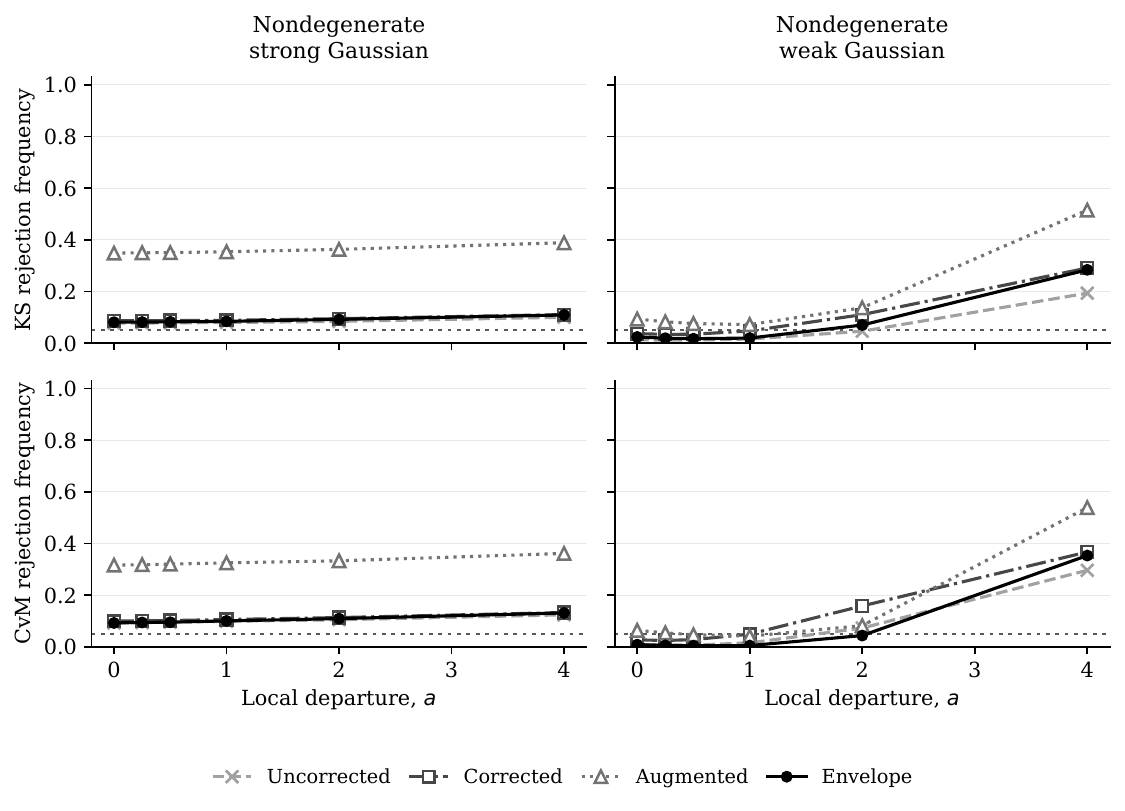}
\caption{Local power at $n=120$ and $\lambda=10$. The added conditional-mean term is $(a/\sqrt n)X^2$. Each point uses 20,000 replications and 999 draws.}\label{fig:power-lambda10}\end{figure}
\FloatBarrier
For the fixed departure $0.1X^2$, finite enumeration and global
population least squares give
\[
 \begin{array}{c|ccc}
 \omega&0&1/2&1\\\hline
 \bm\theta_*& (0.9580,1.0560)^\top&(0.9581,1.0541)^\top&(0.9563,1.0523)^\top\\
 \int\Delta^2dF_X&0.002431&0.002088&0.002318
 \end{array}
\]
These quantities are independent of $\lambda$ and the error loading.
The interior minimizers and positive CvM criteria verify the detection
conditions of Theorem~\ref{thm:consistency}.
Tables~\ref{tab:sim-fixed-power}-\ref{tab:fixed-nong} report
$n=30,90,200$; the replication files contain all seven sample sizes.
\begin{table}[!htbp]\centering
\caption{Fixed-alternative power: independent dyads}\label{tab:sim-fixed-power}
\begingroup\setstretch{1}\small\setlength{\tabcolsep}{4pt}
\begin{tabular}{ccrrrrrrrr}\toprule
 & & \multicolumn{2}{c}{Uncorrected} & \multicolumn{2}{c}{Corrected} & \multicolumn{2}{c}{Augmented} & \multicolumn{2}{c}{Envelope}\\
\cmidrule(lr){3-4}\cmidrule(lr){5-6}\cmidrule(lr){7-8}\cmidrule(lr){9-10}
$\lambda$ & $n$ & KS & CvM & KS & CvM & KS & CvM & KS & CvM\\\midrule
- & 30 & 0.325 & 0.325 & 0.641 & 0.641 & 0.563 & 0.563 & 0.526 & 0.526\\
 & 90 & 1.000 & 1.000 & 1.000 & 1.000 & 1.000 & 1.000 & 1.000 & 1.000\\
 & 200 & 1.000 & 1.000 & 1.000 & 1.000 & 1.000 & 1.000 & 1.000 & 1.000\\
\addlinespace[3pt]
\bottomrule\end{tabular}
\par\smallskip\begin{minipage}{.97\textwidth}\footnotesize Nominal level 5\%; 20,000 replications and 999 draws. The conditional mean is $(1+X)^2+0.1X^2$; the regression family remains $(\theta_0+\theta_1X)^2$.\end{minipage}\endgroup\end{table}

\begin{table}[!htbp]\centering
\caption{Fixed-alternative power: nondegenerate strong Gaussian}\label{tab:fixed-strong}
\begingroup\setstretch{1}\small\setlength{\tabcolsep}{4pt}
\begin{tabular}{ccrrrrrrrr}\toprule
 & & \multicolumn{2}{c}{Uncorrected} & \multicolumn{2}{c}{Corrected} & \multicolumn{2}{c}{Augmented} & \multicolumn{2}{c}{Envelope}\\
\cmidrule(lr){3-4}\cmidrule(lr){5-6}\cmidrule(lr){7-8}\cmidrule(lr){9-10}
$\lambda$ & $n$ & KS & CvM & KS & CvM & KS & CvM & KS & CvM\\\midrule
1 & 30 & 0.056 & 0.066 & 0.126 & 0.128 & 0.218 & 0.252 & 0.121 & 0.125\\
 & 90 & 0.082 & 0.103 & 0.111 & 0.126 & 0.502 & 0.546 & 0.111 & 0.126\\
 & 200 & 0.136 & 0.167 & 0.153 & 0.185 & 0.736 & 0.764 & 0.153 & 0.185\\
\addlinespace[3pt]
4 & 30 & 0.070 & 0.082 & 0.110 & 0.115 & 0.248 & 0.251 & 0.106 & 0.111\\
 & 90 & 0.071 & 0.077 & 0.083 & 0.088 & 0.499 & 0.481 & 0.083 & 0.088\\
 & 200 & 0.075 & 0.083 & 0.080 & 0.088 & 0.680 & 0.662 & 0.080 & 0.088\\
\addlinespace[3pt]
10 & 30 & 0.118 & 0.145 & 0.156 & 0.176 & 0.144 & 0.156 & 0.103 & 0.117\\
 & 90 & 0.083 & 0.105 & 0.096 & 0.115 & 0.302 & 0.284 & 0.089 & 0.105\\
 & 200 & 0.080 & 0.090 & 0.086 & 0.096 & 0.441 & 0.398 & 0.086 & 0.095\\
\addlinespace[3pt]
\bottomrule\end{tabular}
\par\smallskip\begin{minipage}{.97\textwidth}\footnotesize Nominal level 5\%; 20,000 replications and 999 draws. The conditional mean is $(1+X)^2+0.1X^2$; the regression family remains $(\theta_0+\theta_1X)^2$.\end{minipage}\endgroup\end{table}

\begin{table}[!htbp]\centering
\caption{Fixed-alternative power: nondegenerate weak Gaussian}\label{tab:fixed-weak}
\begingroup\setstretch{1}\small\setlength{\tabcolsep}{4pt}
\begin{tabular}{ccrrrrrrrr}\toprule
 & & \multicolumn{2}{c}{Uncorrected} & \multicolumn{2}{c}{Corrected} & \multicolumn{2}{c}{Augmented} & \multicolumn{2}{c}{Envelope}\\
\cmidrule(lr){3-4}\cmidrule(lr){5-6}\cmidrule(lr){7-8}\cmidrule(lr){9-10}
$\lambda$ & $n$ & KS & CvM & KS & CvM & KS & CvM & KS & CvM\\\midrule
1 & 30 & 0.087 & 0.140 & 0.334 & 0.415 & 0.211 & 0.299 & 0.189 & 0.271\\
 & 90 & 0.741 & 0.930 & 0.936 & 0.986 & 0.970 & 0.993 & 0.932 & 0.985\\
 & 200 & 1.000 & 1.000 & 1.000 & 1.000 & 1.000 & 1.000 & 1.000 & 1.000\\
\addlinespace[3pt]
4 & 30 & 0.019 & 0.022 & 0.124 & 0.123 & 0.052 & 0.045 & 0.038 & 0.030\\
 & 90 & 0.072 & 0.122 & 0.194 & 0.282 & 0.288 & 0.315 & 0.182 & 0.242\\
 & 200 & 0.361 & 0.590 & 0.517 & 0.698 & 0.858 & 0.892 & 0.517 & 0.698\\
\addlinespace[3pt]
10 & 30 & 0.017 & 0.017 & 0.083 & 0.076 & 0.037 & 0.023 & 0.022 & 0.012\\
 & 90 & 0.014 & 0.011 & 0.043 & 0.044 & 0.060 & 0.037 & 0.016 & 0.005\\
 & 200 & 0.038 & 0.056 & 0.079 & 0.117 & 0.148 & 0.073 & 0.063 & 0.032\\
\addlinespace[3pt]
\bottomrule\end{tabular}
\par\smallskip\begin{minipage}{.97\textwidth}\footnotesize Nominal level 5\%; 20,000 replications and 999 draws. The conditional mean is $(1+X)^2+0.1X^2$; the regression family remains $(\theta_0+\theta_1X)^2$.\end{minipage}\endgroup\end{table}

\begin{table}[!htbp]\centering
\caption{Fixed-alternative power: dependent non-Gaussian design}\label{tab:fixed-nong}
\begingroup\setstretch{1}\small\setlength{\tabcolsep}{4pt}
\begin{tabular}{ccrrrrrrrr}\toprule
 & & \multicolumn{2}{c}{Uncorrected} & \multicolumn{2}{c}{Corrected} & \multicolumn{2}{c}{Augmented} & \multicolumn{2}{c}{Envelope}\\
\cmidrule(lr){3-4}\cmidrule(lr){5-6}\cmidrule(lr){7-8}\cmidrule(lr){9-10}
$\lambda$ & $n$ & KS & CvM & KS & CvM & KS & CvM & KS & CvM\\\midrule
1 & 30 & 0.034 & 0.045 & 0.230 & 0.250 & 0.078 & 0.101 & 0.068 & 0.091\\
 & 90 & 0.225 & 0.404 & 0.567 & 0.708 & 0.511 & 0.707 & 0.466 & 0.658\\
 & 200 & 0.946 & 0.988 & 0.993 & 0.998 & 0.998 & 0.999 & 0.993 & 0.998\\
\addlinespace[3pt]
4 & 30 & 0.018 & 0.016 & 0.135 & 0.111 & 0.035 & 0.030 & 0.026 & 0.019\\
 & 90 & 0.012 & 0.014 & 0.087 & 0.088 & 0.042 & 0.036 & 0.024 & 0.016\\
 & 200 & 0.048 & 0.083 & 0.188 & 0.246 & 0.119 & 0.146 & 0.098 & 0.122\\
\addlinespace[3pt]
10 & 30 & 0.063 & 0.054 & 0.191 & 0.168 & 0.021 & 0.012 & 0.018 & 0.010\\
 & 90 & 0.053 & 0.044 & 0.104 & 0.084 & 0.032 & 0.019 & 0.027 & 0.014\\
 & 200 & 0.028 & 0.023 & 0.059 & 0.049 & 0.043 & 0.029 & 0.025 & 0.015\\
\addlinespace[3pt]
\bottomrule\end{tabular}
\par\smallskip\begin{minipage}{.97\textwidth}\footnotesize Nominal level 5\%; 20,000 replications and 999 draws. The conditional mean is $(1+X)^2+0.1X^2$; the regression family remains $(\theta_0+\theta_1X)^2$.\end{minipage}\endgroup\end{table}

\FloatBarrier
\subsection{Simulation implementation and numerical checks}\label{app:mc-implementation}

The constrained least-squares routine compares all admissible interior
and boundary stationary points and the corners of $[1/2,3/2]^2$.
Every replication uses $\widehat R_n(x)=\Pn[\widehat\varepsilon\ind\{X\le x\}]$
both for the observed process and to center the bootstrap marks:
$\widetilde r_{ij}(x)=\widehat r_{ij}(x)-\widehat R_n(x)$.
The equality $\widehat R_n(x)=\Pn\widehat r(x)$ holds at an
interior solution; it need not hold on the boundary. The estimator is
interior with probability tending to one. No replication is discarded.

The observed and augmented CvM statistics use the full empirical average.
For the two first-order draws, the simulation assigns each observed
value to the nearest KS grid point in empirical marginal cdf distance
(the smaller point in a tie), or to the zero process value if the
empirical cdf is no larger than that distance. The grid omitted an
observed covariate value in 1,947 size fits and 1,947 power fits.
If all observed values are on the grid, the resulting weighted
average equals the full empirical average for the same draw.
Each design has finite support with positive mass at every support
point. The empirical masses converge to these probabilities, and the
spacing of the quantile levels tends to zero. The grid therefore
contains the full support with probability tending to one.

For the corrected draw, the simulation replaces negative eigenvalues
by zero on the distinct KS grid points. This can differ in finite
samples from the maximum-norm projection in \eqref{eq:psd-maxnorm}.
To see why it suffices for these designs, let $\bm A$ be a symmetric
$J\times J$ matrix and $\bm A_+$ its eigenvalue truncation.
For every $\bm\Sigma\succeq0$, optimality in Frobenius norm gives
\[
 \|\bm A_+-\bm\Sigma\|_{\max}
 \le \|\bm A_+-\bm A\|_F+\|\bm A-\bm\Sigma\|_F
 \le 2J\|\bm A-\bm\Sigma\|_{\max}.
\]
Here $J\le7$ for every $n$, so truncation preserves covariance
consistency at the relevant rate in both Gaussian cases.

For the numerical counts below, an estimate within $10^{-6}$ of the
boundary is counted as a boundary fit.
Boundary fits occur in 142,939 of 1,400,000 size fits and 172,887 of
2,240,000 power fits. For the non-Gaussian design, the counts per
20,000 replications are
\begin{center}
\begin{tabular}{crr}
\toprule
$\lambda$&$n=30$&$n=200$\\\midrule
1&0&0\\
4&4,836&40\\
10&13,385&4,378\\
\bottomrule
\end{tabular}
\end{center}
These counts describe finite-sample estimation effects at an interior
population parameter.

\FloatBarrier

\section{Housing application: bootstrap results and sensitivity checks}\label{app:housing}
Table~\ref{tab:housing-components} reports the component $p$-values underlying
the envelope tests in Table~\ref{tab:empirical}. Table~\ref{tab:housing-sensitivity}
changes one feature of the application at a time, keeping the same locations
and distance bands.
\begin{table}[!htbp]\centering
\caption{Bootstrap $p$-values for the housing application}\label{tab:housing-components}
\begingroup\small\setstretch{1}\setlength{\tabcolsep}{2.4pt}
\begin{tabular}{l*{8}{r}}\toprule
 &\multicolumn{2}{c}{Uncorrected}&\multicolumn{2}{c}{Corrected}&\multicolumn{2}{c}{Augmented}&\multicolumn{2}{c}{Envelope}\\
\cmidrule(lr){2-3}\cmidrule(lr){4-5}\cmidrule(lr){6-7}\cmidrule(lr){8-9}
 & KS & CvM & KS & CvM & KS & CvM & KS & CvM\\\midrule
Distance only: line & $<0.001$ & 0.0012 & $<0.001$ & 0.0013 & $<0.001$ & $<0.001$ & $<0.001$ & 0.0013 \\
Distance only: bands & $<0.001$ & $<0.001$ & $<0.001$ & $<0.001$ & $<0.001$ & $<0.001$ & $<0.001$ & $<0.001$ \\
By access: lines & 0.1455 & 0.2593 & 0.1315 & 0.2488 & $<0.001$ & $<0.001$ & 0.1315 & 0.2488 \\
By access: power functions & 0.0442 & 0.0837 & 0.0377 & 0.0801 & $<0.001$ & $<0.001$ & 0.0377 & 0.0801 \\
\bottomrule\end{tabular}
\par\smallskip\begin{minipage}{.94\linewidth}\footnotesize\textit{Notes:}
Each component uses 19,999 draws. The envelope $p$-value is the larger of the corrected and augmented $p$-values.\end{minipage}\endgroup
\end{table}
\begin{table}[!htbp]\centering
\caption{Sensitivity of the housing envelope tests}\label{tab:housing-sensitivity}
\begingroup\small\setstretch{1}\setlength{\tabcolsep}{5pt}
\begin{tabular}{l*{6}{r}}\toprule
&\multicolumn{2}{c}{600 m}&\multicolumn{2}{c}{1,000 m}&\multicolumn{2}{c}{First sale}\\
\cmidrule(lr){2-3}\cmidrule(lr){4-5}\cmidrule(lr){6-7}
 & KS & CvM & KS & CvM & KS & CvM \\\midrule
Distance only: line & $<0.001$ & $<0.001$ & $<0.001$ & 0.0016 & $<0.001$ & 0.0024 \\
Distance only: bands & $<0.001$ & $<0.001$ & $<0.001$ & $<0.001$ & $<0.001$ & $<0.001$ \\
By access: lines & 0.4516 & 0.4426 & 0.0194 & 0.0826 & 0.0820 & 0.1804 \\
By access: power functions & 0.3018 & 0.2568 & 0.0048 & 0.0256 & 0.0128 & 0.0386 \\
\bottomrule\end{tabular}
\par\smallskip\begin{minipage}{.94\linewidth}\footnotesize\textit{Notes:}
Entries are envelope $p$-values, with 4,999 draws per component.
The first two comparisons change the 800-metre definition of proximity
and retain median prices. The last uses the first transaction by date at
each location, breaking ties by source record number, and retains 800 metres.
\end{minipage}\endgroup\end{table}
\FloatBarrier

\end{document}